\documentclass[12pt]{article} 

\usepackage{natbib,todonotes,amsmath,amssymb, array, boldline, graphicx, epstopdf, grffile, amsthm, ulem, url, xurl, hyperref, float, comment, setspace, lineno}

\usepackage[letterpaper, margin=1.0in]{geometry}
\usepackage{authblk} 

\setcitestyle{authoryear,open={(},close={)}}

\newcommand{\dint}{\text{d}}

\newcommand{\abs}[1]{\left| #1 \right|}

\newtheorem{theorem}{Theorem}

\newtheorem{lemma}[theorem]{Lemma}
\newtheorem{remark}{Remark}

\begin{document}

\title{ARMA approximation of a Non-separable Spatio-Temporal Model with Fractional Smoothnesses in Space and Time}

\date{June, 2026}

\author{S. Knutsen Furset, Geir-Arne Fuglstad, and Espen R.
Jakobsen}

\affil{Department of Mathematical Sciences, Norwegian University of Science and Technology, Norway}

\maketitle





\begin{abstract}The Matérn covariance model is ubiquitous
in spatial modelling, but there is no default choice
for spatio-temporal modelling. In this paper, we consider
the recently proposed ``diffusion-based'' extension of the spatial Matérn covariance model to a spatio-temporal non-separable covariance model that allows fractional smoothnesses in space and in time. The model is described in terms of a  space-time fractional stochastic partial differential equation, but currently proposed computational approaches have strong restrictions on the possible smoothnesses in time. We propose a discretization method based on rational approximations in time to handle arbitrary smoothnesses, which leads
to a vector autoregressive moving average process (VARMA). We prove that the covariance function of the approximation converges pointwise, determine explicit convergence rates as a function of spatial and temporal resolutions and the accuracy of the rational approximation, and conduct numerical verification to demonstrate small pointwise error for low orders of the VARMA process. Through a simulation study, we demonstrate that the parameters can be estimated back and that correctly specifying the temporal smoothness is especially important for forecasting. The approach is illustrated for three months of daily mean temperatures in mainland France.

\vspace{0.4cm}
\end{abstract}

\maketitle



\section{Introduction}\label{sec1}
Gaussian random fields (GRFs) are a key tool for modelling spatio-temporal processes across a range of disciplines \citep{gelfand2010handbook,cressie2011statistics}.
While there exists a wide variety of possible covariance models, most spatio-temporal GRFs used in practice are separable, i.e. the covariance function factors into a spatial covariance function and a temporal covariance function \citep{porcu202130}. This is, in part, motivated by  computational convenience and the fact that popular tools such R-INLA include easy-to-use separable models \citep{bakka2018spatial}. However, separable covariance structures are considered physically unrealistic as they do not allow space and time to interact, e.g., through diffusion and advection.

It is common to separate between descriptive and dynamic models \citep{cressie2011statistics}. In the former, there has been an extensive effort to develop new families of non-separable covariance functions \citep{cressie1999classes,gneiting2002nonseparable, stein2005space,rodrigues2010class,fonseca2011general}. In the latter, covariance structure arises from ``physical'' models defined through, e.g., partial differential equations or integrodifference equations \citep{wikle2010general,wikle2015modern}. This has the advantage that one does not need to show the validity of the resulting family of covariance functions, and a particularly appealing approach is to use stochastic partial differential equations (SPDEs) to build interpretable models  \citep{jones1997models,roques2022spatial}. Simple components such as diffusion and advection can give rise to complex covariance structure through stochastic advection-diffusion equations with spatially varying diffusion and advection \citep{lindgren2011,sigrist2015stochastic,liu2022statistical,clarotto2024spde,berild2024non}. 

In this paper, we consider the non-separable ``diffusion-based'' covariance model recently proposed by \citet{LindgrenBakka}. The model builds on ideas in \citet{krainski2018statistical}, and, for a dimension $d\in\mathbb{N}$, a GRF $u$ is specified as a solution of the SPDE
\begin{equation}
    C^{-1/2}L^{\beta/2}(\partial_t+r^{-1}L^\alpha)^\gamma u(\boldsymbol{s},t) = \mathcal{W}(\boldsymbol{s}, t), \quad \boldsymbol{s}\in\mathbb{R}^d, t\in\mathbb{R},\label{eq:introSPDE}
\end{equation}
where $\partial_t$ is the temporal derivative, $L = (\kappa^2 - \Delta)$ is a spatial operator involving the Laplacian $\Delta$, the constants $\kappa, C, r>0$, the exponents $\alpha, \beta, \gamma>0$, and $\mathcal{W}$ is spatio-temporal Gaussian white noise. We discuss initial conditions, boundary conditions and further restrictions on the exponents in Section \ref{sec:CovModel}. \citet{LindgrenBakka} showed that this SPDE gives rise to an extension of the spatial Matérn model that can be parametrized through spatial and temporal smoothnesses, spatial and temporal ranges, marginal variance, and a non-separability parameter. The marginal spatial behaviour is a Matérn covariance function, and the marginal temporal covariance function is Matérn when $\alpha = 0$.

\citet{kirchner2022} gives meaning to the SPDE in a broader setting on a spatial domain $\mathcal{D}\subset\mathbb{R}^d$ and temporal domain $\mathcal{T}=[0, T]$ for a $T > 0$. \citet[Corollaries 4.2, 4.3, and 5.2]{kirchner2022} show that corresponding results hold in this setting such that the marginal spatial covariance structure corresponds to a Whittle-Matérn GRF with boundary conditions, the marginal temporal covariance structure is Matérn for $\alpha = 0$ (up to the initial condition), and the expressions for smoothnesses given in \citet{LindgrenBakka} still holds. Note that having an initial condition instead of considering SPDE \eqref{eq:introSPDE} defined for all $t\in\mathbb{R}$ means that a warm-up period is required to be sufficiently close to the time-stationary behaviour of the spatio-temporal process,
and, in practice, we want to use the behaviour after the warm-up period to model spatio-temporal phenomena.

However, the discretization of SPDE \eqref{eq:introSPDE} proposed by \citet{LindgrenBakka} only allows a limited selection of spatial and temporal smoothnesses since one must have $\gamma, 2\alpha, \beta \in \mathbb{N}_0$. Their approach can be extended to arbitrary spatial smoothnesses, which requires general $\alpha, \beta>0$, through rational approximations of fractional powers of $L$  \citep{Bolin2019,bolin2024covariance}, but extending to general temporal smoothnesses, by approximating the space-time fractional operator $(\partial_t+r^{-1}L^\alpha)^\gamma$, requires a new approach. \citet{Furset2023} recently proposed an approach to handle fractional $\gamma$, but their approach is more suitable for simulation than inference because it is computationally too expensive for large datasets.

In this paper, we develop an approach to handle inference for any choice of $\gamma,\alpha,\beta>0$, i.e. any smoothnesses in space and time. 
The approach is spectral in space and uses the eigenfunctions of $-\Delta$. This essentially allows us to break the problem down to simpler univariate temporal SPDEs where the solutions must be approximated. While the presentation in the paper focuses on a rectangular domain, the approach can be used on more general domains as long as exact or numerical approximations of the eigenfunctions are known. We show that a rational approximation involving the shift operator gives rise to autoregressive moving average (ARMA) processes, which, in turn, give rise to an approximation in the form of a vector ARMA for the spatio-temporal processes. We prove that the covariance function of the approximation converges pointwise to the true covariance function as the accuracy of the rational approximation increases, the number of spatial basis functions go to infinity, and the temporal step length goes to zero. We also determine explicit convergence rates as a function of the number of spatial basis functions and the temporal step length. The theoretical results are supported by numerical verification of the pointwise convergence of the covariance function.

The vector ARMA process involves an augmented state vector and the likelihood computations can be done through a Kalman filter. We consider two prediction tasks: 1) predicting the current state (filtering), and 2) predicting future behaviour (forecasting). The proposed algorithm has computational complexity dominated by the number of spatial basis functions squared times the number of observations. 

Motivated by an illustrative application to daily mean temperature in mainland France, we construct a simulation study to evaluate the ability to accurately estimate the parameters and to predict. In particular, we focus on the ability to accurately determine the degree of non-separability. The evaluation of parameter estimation focuses on bias and variability between realizations from the chosen scenarios, and the predictions of the underlying true process are evaluated using root mean square error (RMSE) and continuous ranked probability score (CRPS). In the application to daily mean temperature, parameter estimates, predictions and prediction standard deviations based on the full dataset behave well. However, in the cross-validation with hold-out areas, we see a sensitivity to spatial resolution that must be investigated in future work.

The non-separable covariance model and associated SPDE is introduced in Section \ref{sec:CovModel}. Then, in Section \ref{sec:disc}, we motivate the proposed discretization for fractional $\gamma$ and give the full spatio-temporal discretization. This is followed, in Section \ref{sec:ErrorNumVer},  by theorems giving point-wise error bounds and convergence rates on the covariance function of the approximation, and a numerical investigation of point-wise error as a function of temporal smoothness. In 
Section \ref{sec:ModelAndInf}, we present the hierarchical spatio-temporal model and inference procedure, and, in Section \ref{sec:SimStudy}, we evaluate parameter estimation and prediction through a simulation study. Then we illustrate the applicability of the model to daily mean precipitation in Section \ref{sec:application}. The paper ends with discussion in Section \ref{sec:Discussion}. Proofs are given in Sections \ref{suppl:proof-thm-1}-\ref{suppl:proof-thm-3} in the Supplementary Materials. Additional details on the sequential Kalman filter, and supporting tables for the cross-validations performed in the illustrative example can be found in Sections \ref{suppl:seq-Kalman-filter} and \ref{suppl:case-study-parameter-estimate-tables}-\ref{suppl:case-study-cv-prediction-score-tables}, respectively, in the Supplementary Materials. Supporting code is available on Zenodo, see the \hyperref[sec:data-availability]{Data availability statement}.

\section{Covariance model}\label{sec:CovModel}
Consider a bounded spatial domain $\mathcal{D}\subset\mathbb{R}^2$ and the temporal domain $\mathcal{T} = (0, T]$ for a fixed terminal time $T>0$. We restrict the SPDE in Equation \eqref{eq:introSPDE}, as in \citet{kirchner2022}, and define a GRF $u$ on $\mathcal{D}\times \mathcal{T}$ as the solution of
\begin{align}
    \left(\partial_t + r^{-1} (\kappa^2 - \Delta)^\alpha \right)^\gamma u(\boldsymbol{s},t) &= \partial_t W_Q(\boldsymbol{s},t),\label{eq:SPDE:frac}\\
     \frac{r^\gamma}{\sqrt{C} \sigma}(\kappa^2 - \Delta)^{\beta/2}W_Q(\boldsymbol{s},t) &= W_I(\boldsymbol{s},t),  \quad \boldsymbol{s}\in \mathcal{D}, \quad t\in \mathcal{T},\label{eq:SPDE:noise}
\end{align}
where $r, \kappa, \alpha, \gamma, \sigma, \beta > 0$ are parameters, $C>0$ is chosen such that $\sigma$ represents the marginal standard deviation, and we have a zero Neumann boundary condition 
in space. Note that since $\gamma$ may be non-integer, we need to consider an initial condition $u(\cdot, t) \equiv 0$ for $t\leq 0$ due to the non-Markovian nature of the SPDE; see \citet{kirchner2022} for details. The solution $u$ exists in a mean-square sense if $\beta + \alpha (2 \gamma - 1) - 1 > 0$ and $\gamma>\tfrac12$, and if $\mathcal{D}$ is a rectangle, then $u$ has pointwise bounded marginal variance under the same assumption. The generalized GRF $W_I$ is a cylindrical Wiener process, which can be heuristically understood in terms of (i) $W_I(\cdot, 0) \equiv 0$, (ii) independent increments, and (iii) that each increment $W_I(\cdot,t_2)-W_I(\cdot, t_1)$, $t_1\leq t_2$, behaves as Gaussian white noise scaled by $\sqrt{t_2-t_1}$.\footnote{$W_I(\cdot,t_2)-W_I(\cdot, t_1)\overset{\mathcal D}{=}\sum_{k = 1}^\infty \sqrt{t_2-t_1}\Delta w_{k}\, e_k(\boldsymbol{s})$ where $\{\Delta w_{k}\}_k$ are independent $N(0,1)$ random variables and $\{e_k\}_k$ is any orthonormal basis for $L^2(\mathcal D)$.}
Equation \eqref{eq:SPDE:noise} describes spatial smoothing of the cylindrical Wiener process and gives rise to a Q-Wiener process with covariance operator $Q = C\sigma^2 r^{-2\gamma}(\kappa^2-\Delta)^{-\beta}$. Note that both $W_I$ and $W_Q$ might be generalized GRFs without a meaningful point-wise definition.

Assume that the negative Laplacian $-\Delta$ on $\mathcal D$ 
with zero Neumann boundary conditions admits an orthonormal eigenbasis $\{f_k\}$ with associated eigenvalues $\{\xi_k\}$ (see e.g. \cite{davies1995}). Then $\{f_k\}$ is an eigenbasis also for $L_1 = r^{-1}(\kappa^2-\Delta)^\alpha$ and $Q$ with corresponding eigenvalues $\{\mu_k\}$ and $\{\lambda_k\}$ given in terms of $\xi_k$ as
\begin{align}\label{m-l}
 \mu_k = r^{-1} (\kappa^2 + \xi_k)^{\alpha}\qquad\text{and}\qquad \lambda_k = C \sigma^2 r^{- 2\gamma} (\kappa^2 + \xi_k)^{-\beta}.   
\end{align}
The solutions $u$ can then be expressed as an expansion in $\{f_k\}$,
\begin{equation} \label{eq:spectral-representation}
u(\boldsymbol{s},t) = \sum_{k = 1}^\infty c_k(t)\, f_k(\boldsymbol{s}), \quad \boldsymbol{s}\in\mathcal{D}, \quad t\in\mathcal{T},
\end{equation}
 where $\{c_k\}$ are independent Gaussian temporal processes on $\mathcal{T}$. 
%
\citet{kirchner2022} directly specified the solution $u$ through a convolution involving a semi-group. However, following \citet{Furset2023}, one can see that this gives rise to the following coefficients 
in Equation \eqref{eq:spectral-representation}, 
\begin{equation}
    c_k(t) = \frac{\sqrt{\lambda_k}}{\Gamma(\gamma)} \int_0^t \mathrm{e}^{-\mu_k (t-v)} (t-v)^{\gamma - 1}\text{d} w_k(v), \quad t \geq 0 \, .\label{eq:cov:time}
\end{equation}
where $\Gamma$ is the gamma function, and $\{w_k\}$ are independent Brownian motions. This means that the temporal processes $c_k$ can be seen as solutions of stochastic differential equations (SDEs)
\begin{align}
\begin{cases}
    \left(\partial_t + \mu_{k}\right)^\gamma c_{k}(t) = \sqrt{\lambda_{k}} \partial_t w_{k}(t), & t\in \mathcal{T}, \\[0.1cm]
    \hspace{1.46cm} c_k(t) = 0, & t \leq 0.
    \end{cases}
    \label{eq:SPDE:SDE}
\end{align}
These time-fractional SDEs are the main obstacle in the discretization, and the finite element method discretization presented in \citet{LindgrenBakka} only handles the case where $\gamma\in\mathbb{N}$ and the fractional space-time operator is avoided.

The aim of this paper is to consider the limiting behaviour of $u$ when $T\rightarrow \infty$ and $t \gg 0$. More precisely, we aim to construct a time-stationary spatio-temporal process with interpretable parameters. A consequence of \citet[Corollary 4.2]{kirchner2022}, is that each $c_k$ converges to a stationary process. By the Itô isometry and (\ref{eq:cov:time}), the covariance function of this stationary process is given by
\begin{equation} \label{eq:time-process-covariance-function}
C_{c_k}(h) := \lim_{t \rightarrow \infty} \mathrm{E}[c_k(t + h) c_k(t)] = \frac{\lambda_k \mathrm{e}^{-\mu_k |h|}}{(2 \mu_k)^{2 \gamma - 1} \Gamma(\gamma)^2} \int_0^\infty (u + 2 \mu_k |h|)^{\gamma - 1} u^{\gamma - 1} \mathrm{e}^{-u} \, \mathrm{d}u, \, \quad h\in \mathbb{R}.
\end{equation}
This means that the covariance function of $u$ after reaching the time-stationary behaviour can be written as
\begin{equation} \label{eq:process-covariance-function}
C_{u}(\boldsymbol{s}_1, \boldsymbol{s}_2, h) := \lim_{t \rightarrow \infty} \mathrm{E}[u(\boldsymbol{s}_1, t + h)\, u(\boldsymbol{s}_2, t)] = \sum_{k=1}^\infty  C_{c_k}(h)\, f_k(\boldsymbol{s}_1)\,f_k(\boldsymbol{s}_2), \quad  \boldsymbol{s}_1, \boldsymbol{s}_2\in\mathcal{D}, \ h\in \mathbb{R},
\end{equation}
and we construct a discrete-time process that approximates this covariance function and not
a general solution to the original SPDE.

In the case $\mathcal{D} = \mathbb{R}^d$ and $\mathcal{T} = \mathbb{R}$, \citet{LindgrenBakka} showed that the parameters could be reparametrised in terms of six parameters with intuitive interpretations: spatial smoothness $\nu_\mathrm{s}$, temporal smoothness $\nu_\mathrm{t}$, spatial range $r_\mathrm{s}$, temporal range $r_\mathrm{t}$, non-separability $\beta_\mathrm{s}$, and marginal standard deviation $\sigma$. Let $\beta^* = \frac{\nu_\mathrm{s}}{\nu_\mathrm{s} + \frac{d}{2}}$, then the reparametrisation is an invertible representation of the original parameters
given by
\begin{equation} \label{eq:reparametrisation}
 \gamma = \nu_\mathrm{t} \max \left(1, \frac{\beta_\mathrm{s}}{\beta^*} \right) + \frac{1}{2}, \quad \alpha = \frac{\nu_\mathrm{s}}{2\nu_\mathrm{t}} \min \left(1, \frac{\beta_\mathrm{s}}{\beta^*} \right), \quad
\beta = \frac{1 - \beta_\mathrm{s}}{\beta^*} \nu_\mathrm{s}, \quad 
\kappa = \frac{\sqrt{8 \nu_\mathrm{s}}}{r_\mathrm{s}}, \quad r = \frac{r_\mathrm{t} \kappa^{2 \alpha}}{\sqrt{8(\gamma - \tfrac{1}{2})}} \, ,
\end{equation}
%
or conversely
\begin{align*}
\nu_\mathrm{s} = \beta + (2\gamma - 1)\alpha - \tfrac{d}{2}&, \quad
\nu_\mathrm{t} = \gamma - \tfrac{1}{2} + \tfrac{1}{\alpha}\min \left(\beta - \tfrac{d}{2}, 0 \right), \quad
\beta_\mathrm{s} = \frac{(2\gamma - 1)\alpha}{\beta + (2\gamma - 1)\alpha},\\
&r_\mathrm{s} = \kappa^{-1} \sqrt{8 \nu_\mathrm{s}}, \quad
r_\mathrm{t} = r \kappa^{-2 \alpha} \sqrt{8 \left(\gamma - \tfrac{1}{2}\right)} \, .
\end{align*}
The smoothness parameters $\nu_\mathrm{t}$ and $\nu_\mathrm{s}$ quantify smoothness of the process, in the sense that $u$ has $\lceil \nu_\mathrm{t} \rceil-1$ 
temporal and 
$\lceil \nu_\mathrm{s} \rceil-1$ spatial derivatives in the mean square sense, with the fractional part quantifying the remaining Hölder-continuity after differentiation. 
The range parameters $r_\mathrm{s}$ and $r_\mathrm{t}$ quantify the distance where the correlation has dropped to approximately $0.13$ for respectively space and time. The parameter $\beta_\mathrm{s}$ quantify the non-separability of the model, and can be interpreted as the fraction of spatial regularity stemming from the fractional time-space Equation (\ref{eq:SPDE:frac}), as opposed to space only regularisation given by Equation \eqref{eq:SPDE:noise}.
For example, in the case where $\beta_\mathrm{s} = 0$, all the spatial regularity of $u$ comes from the noise-term in Equation (\ref{eq:SPDE:frac}), and the model is separable, i.e. the covariance function of $u$ can be factored into a spatial and a temporal covariance function. $\beta_\mathrm{s} > 0$ gives a non-separable model, with $\beta_\mathrm{s} = 1.0$ giving a ``fully non-separable'' model, where Equation \eqref{eq:SPDE:frac} is driven by 
space-time white noise, and all the regularity of $u$ stems from the diffusion term. See \citet[Section 3]{LindgrenBakka} for a more in-depth discussion of the parameters. 

\citet{LindgrenBakka} also demonstrated, again in the case $\mathcal{D} = \mathbb{R}^d$ and $\mathcal{T} = \mathbb{R}$, that the spatial covariance function $C_u(\boldsymbol{s}_1, \boldsymbol{s}_2, 0)$, $\boldsymbol{s}_1,\boldsymbol{s}_2\in\mathcal{D}$, is Matérn and that for $\beta_\mathrm{s} = 0$, the temporal covariance function $C_u(\boldsymbol{s}, \boldsymbol{s}, h)$, $h\in\mathbb{R}$, is Matérn for any $\boldsymbol{s}\in\mathcal{D}$. This is their motivation for terming this covariance model as a non-separable spatio-temporal extension of the Matérn covariance function. \citet{kirchner2022} give corresponding results for bounded domains $\mathcal{D}\subset\mathbb{R}^d$ where the spatial covariance function instead arises from a Whittle-Matérn SPDE with boundary conditions.

\section{Discretization method}\label{sec:disc}

\subsection{Spatial discretization}\label{sec:disc:space}

We apply a spectral approximation in space by truncating the eigenbasis representation of the solution $u$ in Equation 
\eqref{eq:spectral-representation}. For $M\in\mathbb N$, we define the (spectral) approximate solution $u_M$ by
\begin{equation} \label{eq:spectral-approximation}
u_M(\boldsymbol{s},t) =\sum_{k = 1}^M c_k(t)\, f_k(\boldsymbol{s}), \quad \boldsymbol{s}\in\mathcal{D}, \quad t\in\mathcal{T}.
\end{equation}
%
Bounds on the error incurred in this approximation is discussed in Section \ref{sec:theory:spatial}. 
\medskip

\noindent\textbf{On the domains $\mathcal D$.} \ To use this method we need practical access to the eigenbasis $\{f_k\}$ (and eigenvalues, see below) of $-\Delta$ on $\mathcal D$ with zero Neumann boundary conditions.
In the numerical examples below we will always use 
rectangular domains where eigenfunction are explicit cosine functions (see next paragraph), but the approach would work equally well on, say $\mathcal{D}=\mathbb{S}^2$, where the eigenfunctions of $-\Delta$ are known to be the spherical harmonics.  In fact our method works in any domain where an eigenbasis for the Neumann Laplacian is known or can be efficiently computed, see e.g. \cite{KuSi84,SuZh17}.\medskip

\noindent\textbf{Rectangular domains $\mathcal D$.} \ On $\mathcal D=(0, A_1)\times (0, A_2)$ for $A_1, A_2>0$, the eigenvalues and corresponding eigenfunctions of $-\Delta$ under zero Neumann boundary conditions are given by
\begin{align}
\xi_{k} &= (i_k^2/A_1^2+j_k^2/A_2^2)\pi^2,\nonumber\\[0.2cm]
\label{eq:rectangle-eigenfunctions}
f_{k}(\boldsymbol{s}) &= \frac{2^{1 - \frac12\delta_{i_k,0} - \frac12\delta_{j_k,0}}}
{\sqrt{A_1 A_2}}\cos(\pi i_k s_1/A_1) \cos(\pi j_k s_2/A_2), \quad \boldsymbol{s}=(s_1,s_2)^\mathrm{T}\in\mathcal{D},
\end{align}
where $\delta_{m,n}$ is the Kronecker-delta, 
and 
 $i_k,j_k\in\mathbb N_0$  
are the associated frequencies in $s_1$-direction and $s_2$-direction, respectively. 
There are several ways to order the frequencies and corresponding eigenvalues. Throughout this paper, we assume an arbitrary ordering 
such that
$\xi_1 \leq \xi_2 \leq \ldots$.
In implementations in later sections, we have used the ordering from the \textsf{R} function \texttt{order}.
%

\subsection{Temporal discretization for one frequency}
\label{sec:disc:time}
\subsubsection{Motivation and heuristic derivation}

Our goal is to approximate the stationary distribution 
of the Gaussian process $c_k(t)$, $t \geq 0$, defined by Equation \eqref{eq:SPDE:SDE}, with stationary covariance function (\ref{eq:time-process-covariance-function}). We introduce a temporal grid $t_n = n \Delta t$, $n \in \mathbb Z$, for step size $\Delta t>0$, and introduce the approximation 
 $c_k^n\approx c_k(t_n)$, $n \in \mathbb N$. 
 
To motivate the discrete scheme we consider first $\gamma = 1$. Then Equation \eqref{eq:SPDE:SDE} defines an Ohrnstein-Uhlenbeck process, which by Equation \eqref{eq:cov:time} and the Itô isometry can be represented in distribution by the recursive formula
$$
    c_k(t_{n}) = \mathrm{e}^{-\mu_k\Delta t}c_k(t_{n-1})+\frac{\sqrt{\lambda_k}}{\Gamma(\gamma)} \int_{t_{n-1}}^{t_{n}} \mathrm{e}^{-\mu_k (t_n-s)} \text{d} w_k(s) \
    { \overset{\mathrm{d}}{=} } \ \mathrm{e}^{-\mu_k\Delta t}c_k(t_{n-1}) + \frac{\sqrt{\lambda_k}}{\sqrt{2\mu_k}}\sqrt{1-\mathrm{e}^{-2\mu_k\Delta t}}\, z_k^n , \quad n\in\mathbb{N},$$
where $z_k^1, z_k^2,\ldots \overset{\text{iid}}{\sim} \mathcal{N}(0,1)$. The stationary distribution of this process is achieved by initialising it in its stationary state, $c_k(0)\sim \mathcal{N}(0, C_{c_k}(0))$.
Since  $\frac{\sqrt{\lambda_k}}{\sqrt{2\mu_k}}\sqrt{1-\mathrm{e}^{-2\mu_k\Delta t}} \approx \sqrt{\lambda_k \Delta t}$ for small $\Delta t$,  a natural approximation $c_k^n$ can be obtained by solving
\begin{align}\label{approx_gamma=1}
    (1-\mathrm{e}^{-\mu_k\Delta t}B)\,c_k^n = 
    { \sqrt{\lambda_k \Delta t} }\, z_k^n, \quad n\in\mathbb{N},
\end{align}
where $B$ is the backshift operator defined by
$Bc_k^n = c_k^{n-1}$, and $c_k^0 \sim\mathcal{N}(0,\lambda_k\Delta t(1-\mathrm{e}^{-2\mu_k\Delta_t})^{-1})$ is the stationary state. This defines an autoregressive process of order 1 (AR(1)). Note that $\frac1{\Delta t}(1-\mathrm{e}^{-\mu_k\Delta t}B)$ is a discrete approximation of $(\partial_t + \mu_k)$, so Equation \eqref{approx_gamma=1} is an approximation of Equation \eqref{eq:SPDE:SDE} when $\gamma=1$. We extend this approximation  inductively to the case of integers $\gamma \geq 2$ by 
\[
    \left[\frac{1}{\Delta t}(1-\mathrm{e}^{-\mu_k\Delta t}B) \right]^\gamma c_k^{n} = \sqrt{\frac{\lambda_k}{\Delta t}}z_k^n, \quad n\in\mathbb{N},
\]
or equivalently
\begin{equation}
    \label{eq:integer-gamma-approximation}
    (1-\mathrm{e}^{-\mu_k\Delta t}B)^\gamma c_k^{n} = \epsilon_k^n, \quad n\in\mathbb{N},
\end{equation} 
where $(c_k^0, \cdots, c_k^{-\gamma+1})^\mathrm{T}$ is initialised in the process's stationary state, i.e., the appropriate Gaussian multivariate distribution, and $\epsilon_k^1,\epsilon_k^2,\ldots \overset{\mathrm{iid}}{\sim}\mathcal{N}(0, \sigma_k^2)$ with variance $\sigma_k^2 = \lambda_k (\Delta t)^{2 \gamma - 1}$. 
In practice, the time-stationary marginal variance $C_{c_k}(0)=\frac{\lambda_k \Gamma(2 \gamma - 1)}{(2 \mu_k)^{2 \gamma - 1} \Gamma(\gamma)^2}$ is known analytically, and we adjust $\sigma_k$ to achieve the desired marginal variance by a post hoc procedure discussed in the next section. 
When $\gamma$ is an integer, Equation \eqref{eq:integer-gamma-approximation} describes an AR($\gamma$) process. For non-integer $\gamma$, one could then imagine an approximation of the form
\[
    (1-\mathrm{e}^{-\mu_k\Delta t}B)^{\gamma-\lfloor\gamma\rfloor}(1-\mathrm{e}^{-\mu_k\Delta t}B)^{\lfloor\gamma\rfloor} c_k^{n} = \epsilon_k^n, \quad n\in\mathbb{N},
\]
where the challenging part is to give meaning to and
to approximate the fractional power $(1-\mathrm{e}^{-\mu_k\Delta t}B)^{\gamma-\lfloor\gamma\rfloor}$. We consider rational approximations, which in spirit is similar to \cite{Bolin2019}, but while they approximate a spatial operator with real discrete spectrum, 
we approximate a temporal operator with a complex continuous spectrum; different mathematical tools are therefore needed to analyse the rational approximation in our case. Our rational approximations are given by
\begin{equation}
p(\mathrm{e}^{-\mu_k\Delta t}B)(1-\mathrm{e}^{-\mu_k\Delta t}B)^{\lfloor\gamma\rfloor} c_k^{n-1} = q(\mathrm{e}^{-\mu_k\Delta t}B)\epsilon_k^n,\label{eq:time-discrete}
\end{equation}
where $p,q:\mathbb{C}\rightarrow\mathbb{C}$ are polynomials  of order $m$ with real coefficients that are extended to discrete-time operators by considering them as polynomials in $\mathrm{e}^{-\mu_k\Delta t}B$, and $(c_k^0, \cdots, c_k^{-m-\lfloor \gamma\rfloor+1)})^\mathrm{T}$ is initialised in the processes stationary state. This gives an autoregressive moving average ARMA($m+\lfloor\gamma\rfloor$, $m$) process. As discussed above, the discrete-time process in Equation \eqref{eq:time-discrete} is adjusted through a post-hoc adjustment to 
$\sigma_k$ to achieve the known true time-stationary marginal variance, $C_{c_k}(0)=\frac{\lambda_k \Gamma(2 \gamma - 1)}{(2 \mu_k)^{2 \gamma - 1} \Gamma(\gamma)^2}$.

\subsubsection{Constructing the approximation in three steps}

\noindent \textbf{Step 1:} \
{\em Determine the polynomials $p$ and $q$  in \eqref{eq:time-discrete} to have small approximation error.
} By
Theorem \ref{thm:method-convergence} in Section \ref{sec:theory:time}, this can be done by finding $p$ and $q$ with real coefficients such that the rational function $r(z) = p(z)/q(z)$ approximates $(1-z)^{\gamma-\lfloor\gamma\rfloor}$ well on the complex unit disc. The existence of such a rational function is a consequence of Mergelyan's theorem, see Remark \ref{rem:existence-of-rational-function}(c) for details. Intuitively, we need to consider the complex unit disc since it coincides with the $\ell^2$-spectrum  of the backward shift operator $B$. In practice, we to 
select polynomials $p$ and $q$ of order m with no common factors such that the difference between $p(x)/q(x)$ and $(1-x)^{\gamma-\lfloor \gamma \rfloor}$ is minimized on a discrete grid of $[0, 1]$ with step size $0.001$. This approach has no theoretical guarantees on optimality, but works well in practice.\medskip

\noindent \textbf{Step 2:} {\em Express approximation \eqref{eq:time-discrete} as a Markov process.} For given polynomials $p$ and $q$ of order $m$, Equation \eqref{eq:time-discrete} gives rise to the ARMA($m+\lfloor\gamma\rfloor$, $m$) process of the form
\[
    c_k^n= \sum_{i = 1}^{m+\lfloor\gamma\rfloor}\phi_{k,i}c_k^{n-i}+ \epsilon_k^n+\sum_{i=1}^m \theta_{k,i} \epsilon_{k}^{n-i},
\]
where $\phi_{k,1}, \ldots, \phi_{k,m+\lfloor\gamma\rfloor}$ and $\theta_{k,1}, \ldots, \theta_{k,m}$ can be computed by expanding the polynomials and collecting the terms. Let $\zeta_i = \min\left\{\max\{m, \lfloor\gamma\rfloor\}, i, m+\lfloor\gamma\rfloor-i\right\}$, then the explicit formulas are
\begin{align*}
    & \qquad \phi_{k,i} = -e^{-\mu_k i \Delta t}\sum_{j = 0}^{\zeta_i}(-1)^{i-j} \binom{\lfloor \gamma\rfloor}{i-j} \frac{p_j}{p_0},\quad i = 1, \ldots, m+\lfloor\gamma\rfloor,\\
    \qquad\mathrm{and} & \qquad\theta_{k,i} = \mathrm{e}^{-\mu_k i}\frac{q_i}{q_0},\quad i = 1, \ldots, m,
\end{align*}
where $p_0, \ldots, p_{m+\lfloor\gamma\rfloor}$ and $q_0, \ldots, q_m$ are the coefficients of the polynomials $p(z) =\sum_{i=0}^{m+\lfloor\gamma\rfloor}p_{i} z^{i}$ and $q(z)=\sum_{i=0}^mq_{i} z^{i}$, respectively. Here we have divided Equation \eqref{eq:time-discrete} by $p_0$ and redefined $\epsilon_k^n$ as $\frac{q_0}{p_0}\epsilon_k^n$. This division is unproblematic since we can guarantee that $q_0,p_0 \neq 0$ if the approximation is sufficiently accurate. More precisely, the function $(1-z)^{\gamma - \lfloor \gamma \rfloor}$ has no pole and no zero at $z = 0$, so any sufficiently good approximation satisfies $p_0 = p(0) \neq 0$ and $q_0 = q(0) \neq 0$, as we are assuming the polynomials have no common factors. 

The ARMA process can be described as a Markov process by introducing the augmented state vector
$\boldsymbol{x}_k^n=(c_k^n, c_k^{n-1}, \ldots, c_k^{n-m-\lfloor\gamma\rfloor}, \epsilon_k^{n}, \ldots, \epsilon_k^{n-m})^\mathrm{T}$, which follows a linear dynamical model
\begin{align}\label{MP}
    \boldsymbol{x}_k^n = \mathbf{F}_k\boldsymbol{x}_k^{n-1}+\boldsymbol{v}_k^n, \quad n = 1, \ldots, N,
\end{align}
where $N$ is a finite time horizon, $\boldsymbol{x}_k^0 = \boldsymbol{0}$, $\mathbf{F}_k$ is a $(2m+\lfloor\gamma\rfloor)\times (2m+\lfloor\gamma\rfloor)$ matrix given by
\[
    \mathbf{F}_k = \begin{bmatrix}
        (\phi_{k,1}, \ldots, \phi_{k,m+\lfloor\gamma\rfloor-1}) & \phi_{k, m+\lfloor\gamma\rfloor} & (\theta_{k,1},\ldots, \theta_{k, m-1}) & \theta_{k,m} \\
        \mathbf{I}_{m+\lfloor\gamma\rfloor-1} & 0 & \boldsymbol{0} & 0 \\
        \boldsymbol{0} & 0 & \boldsymbol{0} & 0 \\
        \mathbf{0} & \mathbf{0} & \mathbf{I}_{m-1} & \mathbf{0}
    \end{bmatrix}, \qquad \text{and}\qquad \boldsymbol{v}_k^{n}=\begin{bmatrix}
     1\\ \boldsymbol{0} \\ 1 \\  \boldsymbol{0} 
    \end{bmatrix}  \epsilon_k^n .
\]
It follows that $\boldsymbol{v}_k^n\sim\mathcal{N}(0, \boldsymbol{\Sigma}_k)$,  where  $(\boldsymbol{\Sigma}_k)_{i,j}=\sigma_k^2$ for $i,j\in\{1,m+\lfloor\gamma\rfloor+1\}$ and zero otherwise.
\medskip

\noindent \textbf{Step 3:} {\em Initialise the Markov process \eqref{MP} in the stationary state.} 
To do this,
we compute the variance $\mathbf{S}^n_k$ of $\mathbf x_k^n$  for a large number of iterations $n=N_{\mathrm{init}}$ corresponding to a time $T'=N_{\mathrm{init}}\Delta t$ 
which is much larger than the temporal range $r_\mathrm{t}$ of the process. 
To reduce the long term effect of discretisation errors, we then
adjust the the innovation variance $\sigma_k^2$ such that the resulting process $\{\mathbf{x}_{k,1}^{n}\}=\{c_k^{n}\}$ achieves the true time-stationary marginal variance $C_{c_k}(0)$ at $n=N_{\mathrm{init}}$. Finally, we recompute the variance $\mathbf{S}_k^{N_{\mathrm{init}}}$ and use it to simulate $\boldsymbol{x}_k^0$, the time-stationary starting value for \eqref{MP}. The precise procedure is as follows:
\begin{enumerate}
    \item Set $\sigma_k = 1$ and $\mathbf{S}_k^0 = \mathbf{I}$.\medskip 
    \item Let $N_{\mathrm{init}} \gg r_\mathrm{t}/\Delta t$  (cf. Equation \eqref{eq:reparametrisation}). For $n = 1, \ldots, N_\mathrm{init}$, compute sequentially \ 
    \(
        \mathbf{S}_k^n = \mathbf{F}_k\mathbf{S}_k^{n-1}\mathbf{F}_k^\mathrm{T}+\boldsymbol{\Sigma}_k.
    \)\medskip
    \item  Set \
    \(
        \sigma_k^2 = \dfrac{C_{c_k}(0)}{{ (\mathbf{S}_{k}^{N_\mathrm{\mathrm{init}}}})_{1,1}}\qquad \text{and}\qquad \boldsymbol{x}_k^0\sim \mathcal N(0,\sigma_k^2
\mathbf{S}_k^{N_\mathrm{\mathrm{init}}}). 
    \)  
\end{enumerate}
Starting from $\sigma_k = 1$, the procedure produce the scaling factor to apply to the process such that the approximated marginal variance of $c_k^{N_\mathrm{init}}$ equals $C_{c_k}(0)$.
The marginal covariance matrix of the corresponding rescaled process $\{\boldsymbol{x}_k^{n}\}$ then becomes $
\frac{C_{c_k}(0)}{(\mathbf{S}_{k}^{N_\mathrm{\mathrm{init}}})_{1,1}}
\mathbf{S}_{k}^{N_{\mathrm{init}}}$ at $n=N_{\mathrm{init}}$.
%

\begin{remark}
An alternative to Step 3 is to use the Durbin-Levinson algorithm to compute the limiting covariance matrix $\mathbf{S}_{k}^{\infty}$; see \cite[Chapter 2.5.3]{Brockwell2016} for a description of this algorithm. However, for our application this proved to be numerical unstable when $\Delta t$ is small or $\mu_k$ is large; we discuss this more in Section \ref{sec:rational-numerical-validation}. \smallskip

\end{remark}

\subsection{Full discretization and Markov representation}\label{sec:DiscFull}
Select $M$ spectral basis functions, a temporal step size $\Delta t=T/N$, a temporal grid $t_n = n \Delta t$, $n = 0, 1, \ldots, N$, and let $u_M^n(\cdot) \approx u_M(\cdot, t_n)$. Then the full approximation of the stationary state of \eqref{eq:SPDE:frac} -- \eqref{eq:SPDE:noise} can then be written as
\begin{equation} \label{eq:full-discretisation}
    u^n_M(\boldsymbol{s}) = \sum_{k = 1}^M c_k^n \,f_k(\boldsymbol{s}), \quad \boldsymbol{s}\in\mathcal{D}, \quad n = 0,\ldots, N,
\end{equation}
where the independent discrete-time process $\{c_k^n\}$ are determined as described
in Section \ref{sec:disc:time}.

This gives rise to an AR(1) process with
the augmented state vector $\boldsymbol{x}^n = \left((\boldsymbol{x}_1^n)^\mathrm{T}, \ldots, (\boldsymbol{x}_M^n)^\mathrm{T}\right)^\mathrm{T}$. The augmented state vector follows 
the linear dynamical model initialised in the stationary state,
\begin{equation} \label{eq:state-space-model}
    \boldsymbol{x}^n = \mathbf{F}\boldsymbol{x}^{n-1}+\boldsymbol{v}^n, \quad n = 1, \ldots, N, \qquad \boldsymbol{x}^0
    \sim \mathcal N(0, \tilde{\mathbf{S}}^{N_\mathrm{\mathrm{init}}}
    ), 
\end{equation}
where 
$\mathbf{F}$ and $\tilde{\mathbf{S}}^{N_\mathrm{\mathrm{init}}}$ are $M(2m+\lfloor\gamma\rfloor)\times M(2m+\lfloor\gamma\rfloor)$ block diagonal matrices with respectively blocks $\mathbf{F}_1,\ldots\mathbf{F}_M$ and $\sigma_1^2\mathbf{S}_1^{N_\mathrm{\mathrm{init}}},\dots,\sigma_M^2\mathbf{S}^{N_\mathrm{\mathrm{init}}}_M$, and $\boldsymbol{v}^n=\left((\boldsymbol{v}_1^n)^\mathrm{T},\ldots,(\boldsymbol{v}_M^n)^\mathrm{T}\right)^\mathrm{T}\sim\mathcal{N}(\boldsymbol{0}, \Sigma)$ with $\Sigma$ block diagonal with blocks $\Sigma_1,\ldots, \Sigma_M$. The innovations $\boldsymbol{v}^1, \ldots, \boldsymbol{v}^N$ are independent. 

\section{Error bounds and numerical verification}\label{sec:ErrorNumVer}
\subsection{Spatial error}\label{sec:theory:spatial}

The spectral discretisation (\ref{eq:spectral-representation}) is known to have 
root mean-square (RMS) convergence order $M^{-\frac{\nu_{\mathrm{s}}}{d}}$, where $\nu_\mathrm{s} = \beta + (2\gamma - 1) \alpha - \frac{d}{2}$ is the mean-square spatial smoothness of $u(\cdot, t)$ (\cite{Furset2023}). For practical inference and prediction, we are less interested in RMS (i.e. "strong") convergence and more interested in convergence in distribution. We now give such a result; indeed in the next result we show that covariance functions converge with a rate which is (almost) twice as good as the RMS rate. 

\begin{theorem} \label{thm:weak-spectral-convergence}
Assume $\gamma>\tfrac12$, $\beta + \alpha (2 \gamma - 1) - 1 > 0$, and $\mathcal{D}$ is a bounded domain in $\mathbb{R}^d$ with zero Neumann or Dirichlet boundary condition or a $d$-dimensional compact Riemannian manifold. Then for spatial locations $\mathbf{s}_1, \mathbf{s}_2 \in \mathcal{D}$ and time lag $h \in \mathbb{R}$
\begin{align*}
    \left| C_{u}(\boldsymbol{s}_1, \boldsymbol{s}_2, h) - C^M_{u}(\boldsymbol{s}_1, \boldsymbol{s}_2, h)\right| \lesssim M^{-\frac{2 \nu_\mathrm{s}}{d} + \frac{d - 1}{d}}, \quad \boldsymbol{s}_1,\boldsymbol{s}_2\in \mathcal D, \quad h \in \mathbb{R} \, ,
\end{align*}
where $C_{u}(\boldsymbol{s}_1, \boldsymbol{s}_2, h)$ is the covariance function defined in \eqref{eq:process-covariance-function}, and $C^M_{u}(\boldsymbol{s}_1, \boldsymbol{s}_2, h)$ is the covariance function of the spatial discretisation \eqref{eq:spectral-approximation}. If we restrict to a rectangular domain $\mathcal{D}=\times_{i = 1}^d [0, A_i] \subset \mathbb{R}^d$, where $A_i > 0$ for all $i$, then
\begin{align*}
    \left| C_{u}(\boldsymbol{s}_1, \boldsymbol{s}_2, h) - C^M_{u}(\boldsymbol{s}_1, \boldsymbol{s}_2, h)\right| \lesssim M^{-\frac{2 \nu_\mathrm{s}}{d}}, \quad \boldsymbol{s}_1,\boldsymbol{s}_2\in \mathcal D, \quad h \in \mathbb{R} \, .
\end{align*}
\end{theorem}
The proof can be found in Section \ref{suppl:proof-thm-1} in the Supplementary Materials. Note that the convergence is uniform in space $(s_1,s_2)$ and $h$. On rectangles we obtain twice the strong (RMS) rate, and in 2D, the error then becomes $O(M^{1 -\beta - \alpha(2 \gamma - 1)})$. Double rates for convergence in distribution compared with strong convergence is observed also for other discretisations of SPDEs and SDEs (e.g. \cite{KLL12,KP92}). For general domains our rate is worse off by a term $\frac{d - 1}{d}$ and uniform convergence of covariance functions is lost for small enough $\nu_\mathrm{s}$. This actually occurs on the 2D-sphere for $\nu_\mathrm{s} < 0.5$. Convergence is lost at the poles because there exist sequences of the spherical harmonics whose mass concentrates there (\cite{Sogge2016}).

\subsection{Temporal error}\label{sec:theory:time}

In practice we observe that the dominating error of the temporal discretisation described in Section \ref{sec:disc:time} is the one induced by the rational approximation in the stationary setting, and not the remaining error from the warm-up period affected by the initial condition. Here we give a result about this dominating error, i.e. an error bound for convergence in distribution of the ARMA-process $c_k^n$ defined by \eqref{eq:time-discrete} to the exact stationary process $c_k(t)$, $t\in\mathbb{R}$, defined by \eqref{eq:time-process-covariance-function}. The result is formulated in terms of stationary correlation functions $\rho$, i.e.
\begin{equation} \label{process-correlation-function}
\rho_{c_k}(h) = \frac{\mathrm{e}^{-\mu_k |h|} }{\Gamma(2 \gamma - 1)} \int_0^\infty (u + 2 \mu_k |h|)^{\gamma - 1} u^{\gamma - 1} \mathrm{e}^{-u} \, \mathrm{d} u, \quad h\in\mathbb{R} \, ,\end{equation}
as long as the polynomials $p$ and $q$ in \eqref{eq:time-discrete} are chosen correctly. 

\begin{theorem} \label{thm:method-convergence} 
Assume $\gamma>\tfrac12$, $\beta + \alpha (2 \gamma - 1) - 1 > 0$, $\eta := \gamma - \lfloor \gamma \rfloor$, $\Delta t>0$, $\epsilon > 0$, and $p$ and $q$ of order $m$ polynomials with real coefficients satisfying
\begin{equation} \label{ass:rational-approximation}
\sup_{z \in \mathbb{C}, |z| \leq 1} \left|(1 - z)^{\eta} - p(z) / q(z) \right| < \epsilon \, ,
\end{equation}
and $\tilde{c}_k = \{\tilde{c}_k^n\}_n$ is the ARMA$(m + \lfloor \gamma \rfloor, m)$-process satisfying (cf. \eqref{eq:time-discrete})
\begin{equation*}
p(\mathrm{e}^{-\mu_k h } B) (1 - \mathrm{e}^{-\mu_k h } B)^{\lfloor \gamma \rfloor} \tilde{c}_k^n = q(\mathrm{e}^{-\mu_k h } B) z_k^n, \qquad n\in\mathbb Z, \qquad z_k^n\overset{\mathrm{iid}}{\sim} \mathcal{N}(0, 1) \, .
\end{equation*}
If in addition $r := (1 - \mathrm{e}^{-\mu_k \Delta t})^{-\eta} \epsilon < 1$, then the correlation functions of $c$ and $\tilde c$ satisfy 
\begin{equation*}
\left|\rho_{c_k}(n \Delta t) - \rho_{\tilde{c}_k}(n) \right| \lesssim 
\mu_k^{2 \gamma - 1} \left(\frac{\epsilon \Delta t^{-\eta}}{1 - r} + \Delta t \abs{\log(\Delta t)} + \Delta t^{\min(2\gamma - 1, 1)} \right),\qquad n\in \mathbb N_0 \, ,
\end{equation*}
where the constant is independent of $k$ and $n$, but depends on $\gamma$. 

\end{theorem}
The proof can be found in Section \ref{suppl:proof-thm-2} of the Supplementary Materials.

\begin{remark} (a) \ The function $(1 - z)^{\eta}$ has no poles or zeroes in 
$|z| \leq \mathrm{e}^{-h \mu_k} < 1$. 
By Assumption (\ref{ass:rational-approximation}) the approximation $p(z) / q(z)$ will therefore also have no poles $|z| \leq \mathrm{e}^{-\mu_k h}$. Thus 
$q(\mathrm{e}^{- \mu_k h}z)$ have no zeroes on the unit disc, and the ARMA-approximation is therefore always invertible. Similarly, assuming $\epsilon$ is sufficiently small, $p(z
) / q(z)$ will 
have no zeroes in $|z| \leq \mathrm{e}^{-h \mu_k}$. In this case $p(\mathrm{e}^{-h \mu_k}z)(1 - \mathrm{e}^{-h \mu_k}z)^{\lfloor \gamma \rfloor}$ have no zeroes on the unit disc, and the ARMA approximation is causal.

\medskip
\noindent (b) \ The factor $\Delta t^{-\eta}$ indicates an increase in the error when $\Delta t \rightarrow 0$ which is compensated by corresponding decreases $\epsilon$. This relationship between the step-size and the rational approximation is natural and expected; when $\Delta t \rightarrow 0$ the continuous-time "window" $[t, t - m\Delta t]$ seen by the ARMA-process shrinks, so in order to maintain accuracy when reducing $\Delta t$ we have to increase the rational order $m$ to compensate. When $\eta = 0$ (i.e. $\gamma \in \mathbb{N}$) the continuous time process $c_k$ is (generalised) Markov and this trade-off thus becomes more benign when $\eta$ is small and disappears completely in the limiting case $\eta = 0$.

\medskip
\noindent (c) \ \label{rem:existence-of-rational-function} For any $\epsilon > 0$, the existence of a rational function $\Tilde{r} := p/q$ satisfying (\ref{ass:rational-approximation})  is a consequence of Mergelyan's theorem (\cite{Mergelyan1952}). By the triangle inequality and symmetry of $(1 - z)^\eta$, the symmetrisation $r(z) := \frac12(\Tilde{r}(z) + \overline{\Tilde{r}(\overline{z})})$, 
also satisfies (\ref{ass:rational-approximation}) and has real coefficients.

\end{remark}
\subsection{Total error}

In Theorem \ref{thm:full-convergence} below, we combine Theorem \ref{thm:weak-spectral-convergence} and Theorem \ref{thm:method-convergence} to get a bound on the convergence in distribution for the full time and space discretisation (\ref{eq:full-discretisation}). We formulate the result in terms of convergence of covariance functions, and for simplicity, only on rectangular domains $\mathcal{D} \subset \mathbb{R}^d$. Results for general domains and manifolds follow from a small modification of the proof. 

\begin{theorem} \label{thm:full-convergence} Assume $\gamma>\tfrac12$, $\beta + \alpha (2 \gamma - 1) - 1 > 0$, $\mathcal{D} = \times_{i = 1}^d [0, A_i]$ with zero Neumann or Dirichlet boundary conditions,  $\eta := \gamma - \lfloor \gamma \rfloor$, $h>0$, $\epsilon > 0$, and $p$ and $q$ of order $m$ polynomials with real coefficients satisfying \eqref{ass:rational-approximation}.
For each $k = 1,2,3,...$, let $\tilde{c}_k = \{\tilde{c}_k^n\}_n$ be an ARMA$(m + \lfloor \gamma \rfloor, m)$-process satisfying \eqref{eq:time-discrete}. Then the covariance function $C_{u}(\boldsymbol{s}_1, \boldsymbol{s}_2, t)$ defined in \eqref{eq:process-covariance-function} and the covariance function $\Tilde{C}^M_{u}(\boldsymbol{s}_1, \boldsymbol{s}_2, n)$ of \eqref{eq:full-discretisation} at spatial locations $\mathbf{s}_1, \mathbf{s}_2 \in \mathcal{D}$ and time lag $
n \Delta t$, $n \in \mathbb{N}_0$, satisfy
\begin{equation*}
    \left|C_{u}(\boldsymbol{s}_1, \boldsymbol{s}_2, n \Delta t) - \Tilde{C}^M_{u}(\boldsymbol{s}_1, \boldsymbol{s}_2, n) \right| \lesssim M^{-\frac{2 \nu}{d}} + C(M) \left( \frac{\epsilon \Delta t^{-\eta}}{1 - r} + \Delta t \abs{\log(\Delta t)} + \Delta t^{\min(2\gamma - 1, 1)} \right)\, ,
\end{equation*}
where 
\begin{align*}
    C(M) \lesssim \begin{cases}
        M^{1 - \frac{2 \beta}{d}} \, & , \quad \beta < \frac{d}{2} \, , \\
        \log(M)  \, & , \quad \beta = \frac{d}{2} \, , \\
        1  \, & , \quad  \beta > \frac{d}{2} \, ,
    \end{cases} \, .
\end{align*}
\end{theorem}
The proof can be found in Section \ref{suppl:proof-thm-3} in the Supplementary Materials. 

\begin{remark} 
When $\beta \geq \frac{d}{2}$, the constant $C(M)$ blows up as $M\to \infty$ and must be compensated by decreasing $\Delta t$ sufficiently fast. Minimizing the error estimate for fixed $\Delta t$ yields the best possible value of $M=M(\Delta t)$ and a minimal error bound 
with a reduced rate in $\Delta t$ 
compared to when $\beta <\frac d2$.
\end{remark}

\subsection{Numerical validation of the rational approximation} \label{sec:rational-numerical-validation}

\begin{figure}
    \includegraphics[width=0.45\textwidth]{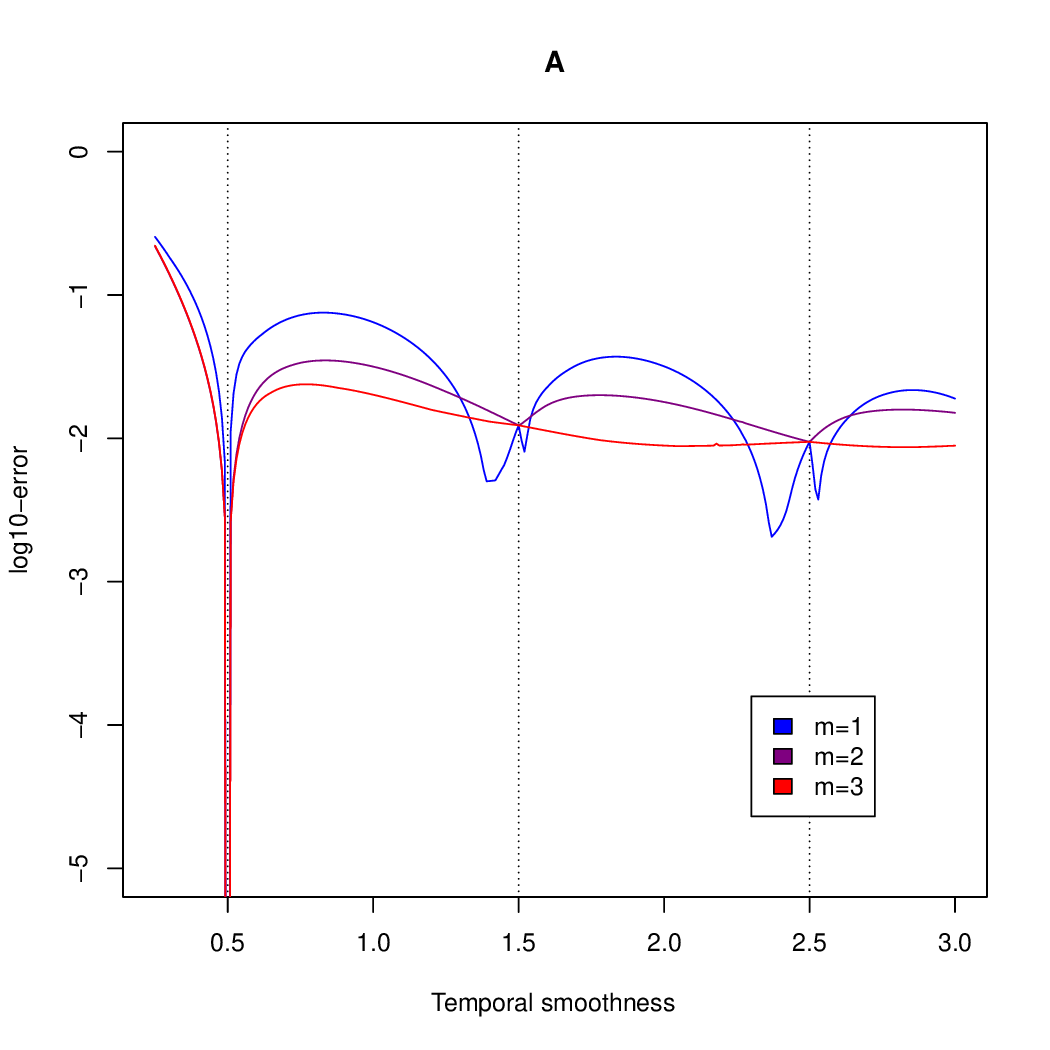}
    \includegraphics[width=0.45\textwidth]{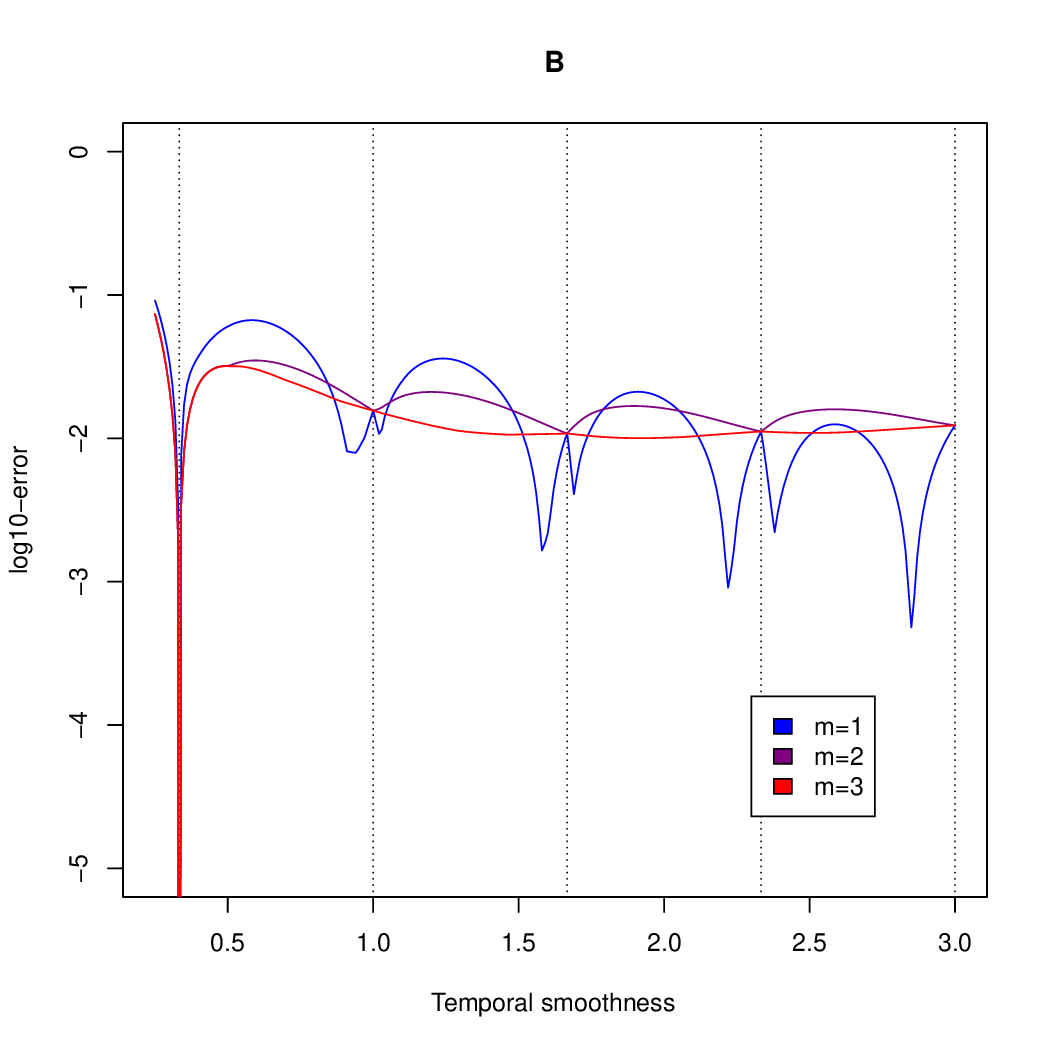}
    \includegraphics[width=0.45\textwidth]{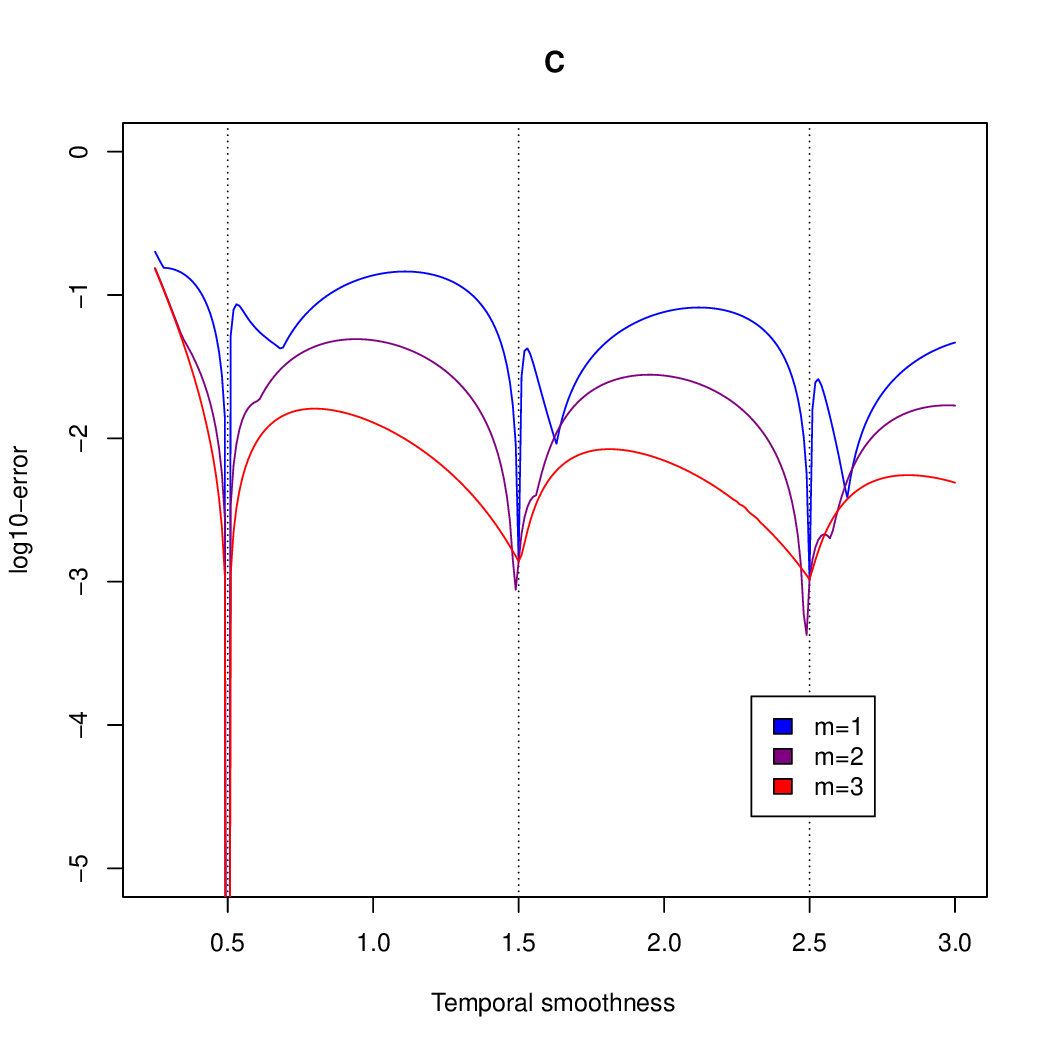}
    \includegraphics[width=0.45\textwidth]{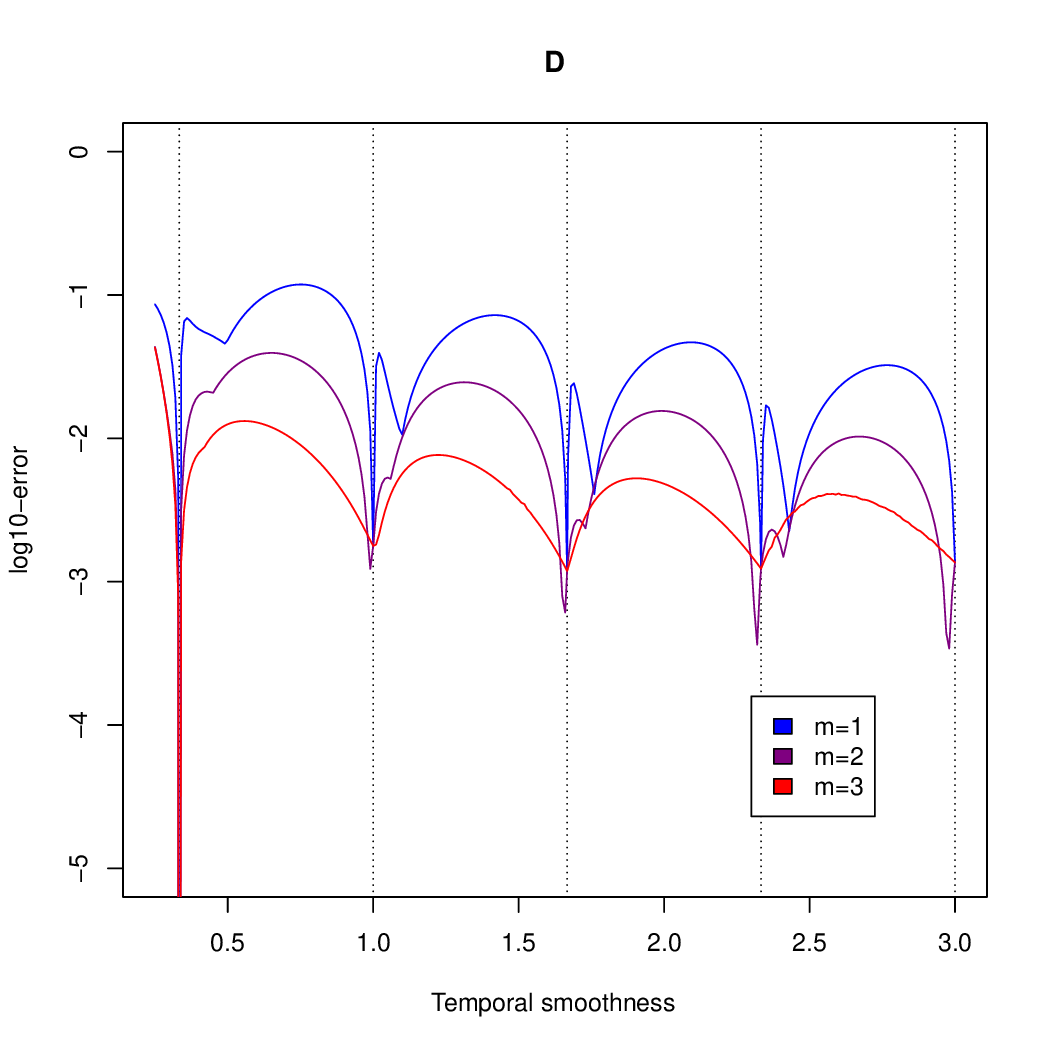}
    \centering
    \caption{Sup-norm $\log_{10}$-error in covariance function of the rational approximation using $M = 16^2$ spatial basis function. Different colours correspond to different orders of approximation. The x-axis represents the temporal smoothness $\nu_t$. The dashed vertical bars represent integer values of $\gamma$. (A) is $r_\mathrm{t} = 1.0$, $\beta_\mathrm{s} = 0.25$. (B) is $r_\mathrm{t} = 1.0$, $\beta_\mathrm{s} = 0.50$. (C) is $r_\mathrm{t} = 3.0$, $\beta_\mathrm{s} = 0.25$. (D) is $r_\mathrm{t} = 3.0$, $\beta_\mathrm{s} = 0.50$. All cases have $\nu_\mathrm{s} = 0.5$, $r_\mathrm{s} = 0.25$, and $\sigma = 1.0$.}
    \label{fig:rat_error}
\end{figure}

We numerically validate the accuracy of the covariance function of the (vector) ARMA process compared to the target covariance function \eqref{eq:process-covariance-function}. We consider three orders of the rational approximation $m = 1, 2, 3$, where the polynomials used in the rational approximation were chosen using the procedure in Section \ref{sec:disc:time} through the \textsf{R}-package \texttt{pracma} \cite{Borchers2011}. We fix $\nu_\mathrm{s} = 0.5$, $r_\mathrm{s} = 0.25$, and $\sigma = 1$, and create four cases by combinations of $r_\mathrm{t}\in\{1.0, 3.0\}$ and $\beta_\mathrm{s}\in\{0.25, 0.5\}$. For each case, we consider temporal smoothnesses $\nu_\mathrm{t} \in\{ 0.25,0.26, \ldots, 3.00\}$.

For simplicity, we consider a one-dimensional spatial domain $\mathcal{D} = (0,1)$ and a temporal domain $\mathcal{T}=[0,\lceil \frac{r_t}{0.05} \rceil]$. We fix the spatial resolution $M = 16^2$ and define the true covariance function, $C_u^{M}$, by truncating the sum in Equation \eqref{eq:process-covariance-function}. For each combination of parameters, we compute the supremum norm of the error between $C_u^{M}$ and the covariance function $\Tilde{C}_u^{M}$ of the full discretization with $M$ spatial basis functions and time step $\Delta t = 0.05$, i.e.
$$
    \text{Sup-Error}=\sup_{s_1,s_2\in\mathcal{D}, h\in\tilde{\mathcal{T}}} \left|C_u^{M}(s_1,s_2, h)-\tilde{C}_u^{M}(s_1,s_2,h)\right|,
$$
where  $\tilde{\mathcal{T}} = \{0, \Delta t, 2\Delta t, \ldots, \lceil \frac{r_t}{0.05} \rceil]\}$ since $\tilde{C}_{u}^{M}$ is only defined for discrete time lags. To compute (an approximation of) Sup-Error, we replace $\mathcal{D}$ with an equidistant grid $\tilde{\mathcal{D}} = \{0, 0.05, \ldots, 1\}$ with 21 grid locations.

For each spatial basis function, the temporal covariance function for the ARMA-approximation of the temporal processes associated to that basis function was computed using \verb|tacvfARMA| from the \textsf{R}-package \texttt{ltsa}, which uses Durbin-Levinson recursion; see \cite{McLeod2007} for details on this package. Approximations of higher order than $3$ proved to be unstable, i.e. the approximating ARMA-process would often be computationally non-stationary and the matrices used in the Durbin-Levinson recursion would therefore be ill-conditioned. Note, however, that this is only a problem when computing the correlation function of the ARMA-process using Durbin-Levinson recursion, it can still be used in Kalman filtering as discussed in the next section. 

The Sup-Error is plotted in Figure \ref{fig:rat_error}. In the test cases, the error is almost always below $0.1$, and often much smaller. Sup-Error generally declines with increasing rational order $m$. However, approximations using first-order rational approximations sometimes outperform higher order functions. As expected, the approximation 
generally performs better when $\gamma$ is near to integer values (see Figure \ref{fig:rat_error}), and exceptionally well close to $\gamma = 1$, where the correlation function is exact for the discrete time steps. 

\section{Hierarchical model and inference}\label{sec:ModelAndInf}
\subsection{Full model}\label{sec:hierFull}
We consider $n_{\mathrm{obs}} $ observations $y_1, \ldots, y_{n_{\mathrm{obs}}}$ at spatial locations $\boldsymbol{s}_1, \ldots \boldsymbol{s}_{n_\mathrm{obs}}\in \mathcal{D}$ and time points
$t_1, \ldots t_{n_{\mathrm{obs}}}\in\{0, \Delta t, \ldots, N\Delta t\}\subset[0,T]$. Typically, $\Delta t$ would be daily or hourly observation.  We assume a Gaussian observation model with conditionally independent observations
\[
    y_i|\mu(\boldsymbol{s}_i, t_i), \sigma_{\mathrm{obs}}^2 \sim \mathcal{N}(\mu(\boldsymbol{s}_i, t_i), \sigma_{\mathrm{obs}}^2), \quad i = 1, \ldots, n_{\mathrm{obs}},
\]
where 
$\mu$ is the underlying spatio-temporal signal.

We model $\mu$ as a GRF
\[
    \mu(\boldsymbol{s}, t) = \boldsymbol{g}(\boldsymbol{s},t)^\mathrm{T}\boldsymbol{\beta}+u(\boldsymbol{s}, t), \quad \boldsymbol{s}\in\mathcal{D}, \quad t\geq 0,
\]
where $\boldsymbol{g}$ is a $p_{\mathrm{cov}}$-dimensional vector of spatio-temporal covariates, $\boldsymbol{\beta}$ is a $p_{\mathrm{cov}}$-dimensional vector of coefficients, and
$u$ is a GRF with the covariance function given in Equation
\eqref{eq:process-covariance-function}. The parameters of the latent model and the observation model are $\boldsymbol{\theta} = (\beta_1, \ldots, \beta_{p_{\mathrm{cov}}}, r_\mathrm{s}, r_\mathrm{t}, \nu_\mathrm{s},\nu_\mathrm{t}, \beta_\mathrm{s}, \sigma, \sigma_{\mathrm{obs}})^\mathrm{T}$.

We have organised the observations in vectors $\boldsymbol{y}^n$ of observations available in time step $n$ for $n = 1, \ldots, N$. These vectors may be of differing lengths because a different number of measurement stations are available in each time step. Using the approximation in Section \ref{sec:disc},  gives rise to a linear dynamical model,
\begin{align}
    \boldsymbol{y}^n|\boldsymbol{x}^n, \boldsymbol{\theta} &\sim \mathcal{N}(\mathbf{G}_n\boldsymbol{\beta}+\mathbf{H}_n\boldsymbol{x}^n, \sigma_{\mathrm{obs}}^2\mathbf{I}) \notag \\
    \boldsymbol{x}^n &= \mathbf{F}\boldsymbol{x}^{n-1} + \boldsymbol{v}^n, \quad n = 1, \ldots, N, \label{eq:LinDynModel}
\end{align}
where $\mathbf{H}_n$ is the matrix mapping coefficients $\boldsymbol{x}^n$ to the locations observed in time step $n$, $\mathbf{G}_n$ is similarily the design matrix of covariates for the spatial locations observed at time step $n$, and $\boldsymbol{v}^1,\ldots,\boldsymbol{v}^N|\boldsymbol{\theta} \overset{\text{iid}}{\sim}\mathcal{N}(\boldsymbol{0}, \Sigma)$. The dependence on $\boldsymbol{\theta}$ is supressed in the notation $\Sigma$; see Section \ref{sec:disc} for a description of how $\Sigma$ varies as a function of $\boldsymbol{\theta}$. In practice, we set $\boldsymbol{x}^0 = 0$, but propagate the model $N_{\mathrm{init}}$ steps before starting to conditioning on observations to start approximately in the time-stationary distribution. We discuss the details in the implementation described in the next section.

\subsection{Parameter inference and prediction}
\label{sec:HierParInf}
Consider $\boldsymbol{\theta}$ fixed.
Assume initial covariance matrix $\mathbf{S}_0 = 0\cdot \mathbf{I}$. We propagate the dynamic model without observation for $N_{\mathrm{init}}$ steps to give
\[
    \mathbf{S}_n = \mathbf{F}\mathbf{S}_{n-1}\mathbf{F}^\mathrm{T} + \Sigma, \quad n = 1, \ldots, N_{\mathrm{init}}.
\]
This gives an approximate covariance matrix for the time-stationary
distribution given by $\mathbf{S}_{N_{\mathrm{init}}}$. The expectation is $\boldsymbol{0}$.

For $n = 1, \ldots, N$,  let $\tilde{\boldsymbol{m}}_n = \mathrm{E}[\boldsymbol{x}^n|\{\boldsymbol{y}^i\}_{i\leq n-1}]$, $\hat{\boldsymbol{m}}_n = \mathrm{E}[\boldsymbol{x}^n|\{\boldsymbol{y}^i\}_{i\leq n}]$, $\tilde{\mathbf{S}}_n = \mathrm{Cov}[\boldsymbol{x}^n|\{\boldsymbol{y}^i\}_{i\leq n-1}]$, and $\hat{\mathbf{S}}_n = \mathrm{Cov}[\boldsymbol{x}^n|\{\boldsymbol{y}^i\}_{i\leq n}]$. We compute the log-likelihood of $\boldsymbol{\theta}$ for the model in Equation \eqref{eq:LinDynModel} using a sequential Kalman filter, where we assume inital values set to
$\hat{\boldsymbol{m}}_0 = \boldsymbol{0}$ and $\hat{\mathbf{S}}_0 = \mathbf{S}_{N_{\mathrm{init}}}$.
We compute the distribution of $\boldsymbol{x}^n|\{\boldsymbol{y}^i\}_{i\leq n-1}$ by computing, respectively, the mean and
the covariance matrix
\[
    \tilde{\boldsymbol{m}}_n = \mathbf{F}\hat{\boldsymbol{m}}_{n-1}, \quad \text{and}\quad \tilde{\mathbf{S}}_{n} = \mathbf{F}\hat{\mathbf{S}}_{n-1}\mathbf{F}^\mathrm{T} + \Sigma.
\]
Then we compute the Kalmain gain and one-step ahead point-forecasts
\begin{equation*}
    \mathbf{K}_n = \Tilde{\mathbf{S}}_n \mathbf{H}_n^\mathrm{T}(\mathbf{H}_n\tilde{\mathbf{S}}_n\mathbf{H}_n^\mathrm{T}+\sigma_{\mathrm{obs}}^2\mathbf{I})^{-1}, \quad \quad\text{and}\quad
    \tilde{\mathbf{y}}^n = \mathbf{G}_n\boldsymbol{\beta} + \mathbf{H}_n \tilde{\mathbf{m}}_n, \,
\end{equation*}
and find the distribution of $\boldsymbol{x}^n|\{\boldsymbol{y}^i\}_{i\leq n}$ by computing, respectively, the mean and the covariance matrix
\begin{equation*}
    \hat{\boldsymbol{m}}_n = \Tilde{\boldsymbol{m}}_n + \mathbf{K}_n (\mathbf{y}^n - \tilde{\mathbf{y}}^n), \quad \text{and}\quad
    \Hat{\mathbf{S}}_n = \Tilde{\mathbf{S}}_n - \mathbf{K}_n \mathbf{H}_n \Tilde{\mathbf{S}}_n \, .
\end{equation*}
This is done iteratively for $n = 1, \ldots, N$. The log-likelihood is then computed by
\begin{equation*}
    \ell(\boldsymbol{\theta}; \boldsymbol{y}^1, \ldots, \boldsymbol{y}^N) = \sum_{n = 1}^{N} \log(\phi(\mathbf{y}^n; \tilde{\mathbf{y}}^n, \mathbf{A}_n)) \, ,
\end{equation*}
where $\mathbf{A}_n=\mathbf{H}_n\tilde{\mathbf{S}}_n\mathbf{H}_n^\mathrm{T}+\sigma_{\mathrm{obs}}^2\mathbf{I}$, and $\phi(\cdot; \boldsymbol{\mu}, \Sigma)$ is the multivariate Gaussian probability density function of $\mathcal{N}(\boldsymbol{\mu}, \Sigma)$. Let $n_{\mathrm{obs}}^n$ be the number of observations at time step $n = 1, \ldots, N$, and let $b$ be the dimension of the latent state vector. When discussing the complexity, we assume that the order of the rational approximation is fixed and that $b$ is proportional to the number of spatial basis functions $M$. Then the above computations have a complexity $\mathcal{O}((n_{\mathrm{obs}}^n)^3+M(n_{\mathrm{obs}}^n)^2+ M^2n_{\mathrm{obs}}^n)$ for time step $n$. We assume that $b < n_{\mathrm{obs}}^n$. So using the Woodbury matrix identity would result in a better complexity $\mathcal{O}(M^3+M^2(n_{\mathrm{obs}}^n)+ M(n_{\mathrm{obs}}^n)^2)$. However, the quadratic complexity in $n_{\mathrm{obs}}^n$ would remain since the observation matrix $\mathbf{H}_n$ is dense and the Woodbury matrix identity would involve the term $\mathbf{H}_n^\mathrm{T}\mathbf{H}_n$. Therefore, a better strategy is to update sequentially based on the measurements $y_1^n, y_2^n, \ldots, y_{n_{\mathrm{obs}}^n}^2$ in the vector $\boldsymbol{y}^n$ instead of doing a joint update. This gives a complexity $\mathcal{O}(M^2(n_{\mathrm{obs}}^n))$. Details on the sequential Kalman filter are found in Section \ref{suppl:seq-Kalman-filter} in the Supplementary Materials. 

In the optimization of the log-likelihood, we used the reparametrised parameters $\log(\frac{\nu_t - 0.25}{2.5}) - \log(1 - \frac{\nu_t - 0.25}{3})$, $\log(\nu_s - 0.25)$, $\frac{1}{3}\log( -\frac{2\beta_s}{\beta_s - 1})$, $\log(r_t - 0.005)$, $\log(r_s - 0.005)$, $\log(\sigma - 0.005)$, and $\log(\sigma_{\mathrm{obs}})$. Note that this parametrisation restricts the domain of the parameters to $0.25 < \nu_t < 3.25$, $\nu_\mathrm{s}>0.25$, and $r_\mathrm{s}, r_\mathrm{s}, \sigma_{\mathrm{obs}}>0.005$. These restrictions were made to keep the parameter inference away from extreme values where we encounter numerical instabilities. The log-likelihood was optimised using the function \texttt{optimParallel} from the \textsf{R}-package \texttt{optimParallel}, running on $9$ cores. This is a parallelised implementation of the L-BFGS-B method, see \cite{Gerber2018} for details. Unless otherwise stated, we use $M_{inf}^2 = 8^2 = 64$ basis functions for inference.

After fixing the parameters $\boldsymbol{\theta} = \hat{\boldsymbol{\theta}}_{\mathrm{MLE}}$, the Kalman filter produces predictive distributions $\boldsymbol{x}^n|\boldsymbol{y}^1, \ldots, \boldsymbol{y}^n$ for filtered estimates  through mean vector $\hat{\mathbf{m}}_n$ and covariance matrix $\hat{\mathbf{S}}_n$, and predictive distributions $\boldsymbol{x}^{n+1}|\boldsymbol{y}^1, \ldots, \boldsymbol{y}^n$ for one-step ahead forecasts  through mean vector $\Tilde{\mathbf{m}}_{n+1}$ and covariance matrix $\Tilde{\mathbf{S}}_{n+1}$. These predictive distributions can be used to compare the predictions under different models. If it is necessary to predict new observations, one would transform using basis functions and add measurement noise variance $\sigma_{\mathrm{obs}}^2$.

\section{Simulation study}\label{sec:SimStudy}

\subsection{Motivation and scenarios}
\label{sec:SimStudy:Mot}
We aim to assess the ability to estimate parameters and to make predictions using the proposed model. The main interest lies in the behaviour for different non-separability parameters and different signal-to-noise ratios. We choose the spatial domain $\mathcal{D} = [0,1]^2$ and a temporal domain of length $T = 45$, we fix the true parameters $\nu_t = 1.0$, $\nu_s = 1.0$, $r_t = 10.0$, $r_s = 1.0$, and $\sigma = 3.5$. We then create four scenarios for the remaining two parameters: LL is $\beta_\mathrm{s} = 0.25$ and $\sigma_{\mathrm{obs}} = 0.35$, LH is $\beta_\mathrm{s} = 0.25$ and $\sigma_{\mathrm{obs}} = 0.75$, HL is $\beta_\mathrm{s} = 0.75$ and $\sigma_{\mathrm{obs}} = 0.35$, and HH is $\beta_\mathrm{s} = 0.75$ and $\sigma_{\mathrm{obs}} = 0.75$, i.e. LL and LH are low non-separability with $95.6\%$ and $99.0\%$ signal-to-noise ratio (SNR), respectively, and HL and HH are high non-separability with $95.6\%$ and $99.0\%$ SNR, respectively. Here SNR is measured by $\sigma^2/(\sigma^2+\sigma_{\mathrm{obs}}^2)$.

When simulating the truth, we consider $N = 45$ time steps, denoted by $t=1,2\ldots,45$, and use $M_{\mathrm{sim}}=32^2=4096$ spatial basis functions. For each basis function we construct the $45 \times 45$ covariance matrix $\Sigma_k$ of $(c_k(1), ..., c_k(N))^\mathrm{T}$ by approximating (\ref{process-correlation-function}) by numerical quadrature using the R-function \verb|integrate|. 
We then draw a realisation of the latent spatio-temporal process, and draw $250$ observation locations from a uniform distribution within the square $[0.2, 0.8]^2$ by computing (\ref{eq:spectral-approximation}) and adding i.i.d. Gaussian measurement error, i.e. we use the model described in Section \ref{sec:hierFull}, but do not include any covariates. We create $R = 30$ replicates, and denote the true values of the simulated latent states by $\boldsymbol{x}^{n,r}$ and the observations $\boldsymbol{y}^{n,r}$ for time steps $n = 1,\ldots, 45$ and replicates $r = 1, \ldots, 30$.

\subsection{Candidate models and evaluation}
The model in Section \ref{sec:hierFull} without covariates has parameters $\boldsymbol{\theta}=(r_s, r_t, \nu_s, \nu_t, \beta_s, \sigma, \sigma_\mathrm{obs})^\mathrm{T}$. For estimation, we used the spatial resolution $M_{\mathrm{inf}} = 8^2 = 64$ and the same temporal resolution as the truth. We estimate these seven parameters using maximum likelihood inference as described in Section \ref{sec:HierParInf}. The model in which all seven parameters are estimated is denoted the 'Full' model.

As a baseline for comparison, we also introduce a reduced, separable model, which we denote the 'Simple' model. In the reduced model, we fix $\nu_\mathrm{t}=0.5$, $\nu_\mathrm{s}=1.0$ and $\beta_\mathrm{s}=0$, and estimate the remaining four parameters. The choice of $\nu_\mathrm{t}=0.5$ and $\beta_\mathrm{s}=0$ is motivated by the the fact that this yields a simple AR(1) process in time driven by coloured noise, whereas $\nu_\mathrm{s}=1.0$ was chosen to match the true value of $\nu_\mathrm{s}$. The model is discretized in the same way as 'Full', but does not include any rational approximation as it corresponds to using $\gamma = 1$ and $\alpha = 0$ in Equation \eqref{eq:SPDE:noise}. Note, this means that we cannot achieve temporal smoothness $\nu_\mathrm{t} = 1$ and that $\beta_\mathrm{s}$ is not estimated.

We use the maximum likelihood estimates $\hat{\boldsymbol{\theta}}_{\mathrm{MLE}}^r$, for realization $r=1,\ldots,R$, as the predicted parameter values, and compare parameter predictions through violin plots that illustrate the bias and spread in parameter predictions across realizations.

For each replicate $r$ we simulate a new base truth $\boldsymbol{x}^{n,r}_{\text{pred}}$ and a new dataset $\boldsymbol{y}^{n,r}_{\text{pred}}$, using the same method as described in Section \ref{sec:SimStudy:Mot}. We apply the Kalman filter to this new dataset, with the inferred parameters $\hat{\boldsymbol{\theta}}_{\mathrm{MLE}}^r$, i.e. the test dataset is independent of the training dataset. Our predictor for $u_{M_{\mathrm{sim}}}^{n,r}(\boldsymbol{s})$ is the posterior mean of $u_{M_{\mathrm{inf}}}^{n,r}(\boldsymbol{s})$ at each location $\boldsymbol{s}\in\mathcal{D}$ for time point $t$ conditional on estimated parameters. For filtering and one-step ahead forecasts, we compute
$$\text{RMSE} = \sqrt{\frac{1}{R}\sum_{r = 1}^{R} \int_\mathcal{D}\left(u_{M_{\mathrm{sim}}}^{n,r}(\boldsymbol{s})-\widehat{u_{M_{\mathrm{inf}}}^{n,r}(\boldsymbol{s})}\right)^2\mathrm{d}\boldsymbol{s}}=\sqrt{ \frac{1}{R} \sum_{r = 1}^{R} \sum_{k = 1}^{M_{\mathrm{sim}}} \left(c_k^{n,r} - \widehat{c_k^{n,r}} \right)^2 },$$
where $\widehat{c_k^{n,r}}$ is the posterior mean and we interpret $\widehat{c_k^{n,r}} = 0$ for $k = M_{\mathrm{inf}},M_{\mathrm{inf}}+1, \ldots, M_{\mathrm{sim}}$. For CRPS, we use an equidistant grid $\{\boldsymbol{s}_{i,j}\}_{i,j \in \{1,2,...,21\}}$ of $21 \times 21$ nodes on the unit square, going from 0 to 1.
$$\text{CRPS} := \frac{1}{21^2 R} \sum_{r = 1}^{R} \sum_{i = 1}^{21} \sum_{j = 1}^{21} \text{crps}\left(u_{M_{\mathrm{sim}}}^{n,r}(\boldsymbol{s}_{i,j}),\mathcal{N}\left(\widehat{u_{M_{\mathrm{inf}}}^{n,r}(\boldsymbol{s}_{i,j})}, \mathrm{Var}(\widehat{u_{M_{\mathrm{inf}}}^{n,r}(\boldsymbol{s}_{i,j})} )\right)\right),$$
where the first argument of $\mathrm{crps}$ is the true value, and the second argument is the predictive Gaussian distribution.

\subsection{Results}

\begin{figure}
    \centering
    \includegraphics[width=0.3\linewidth]{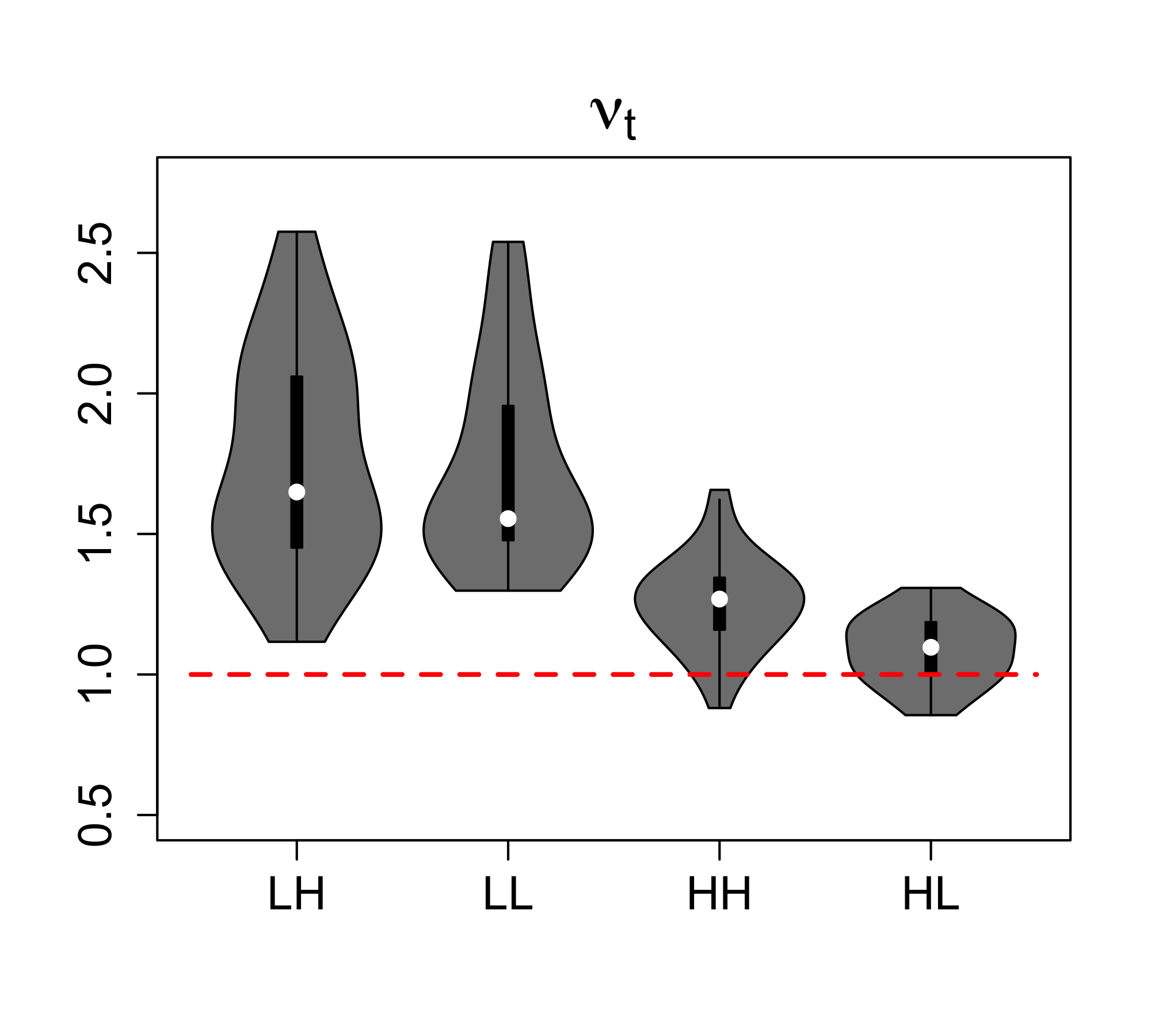}
    \includegraphics[width=0.3\linewidth]{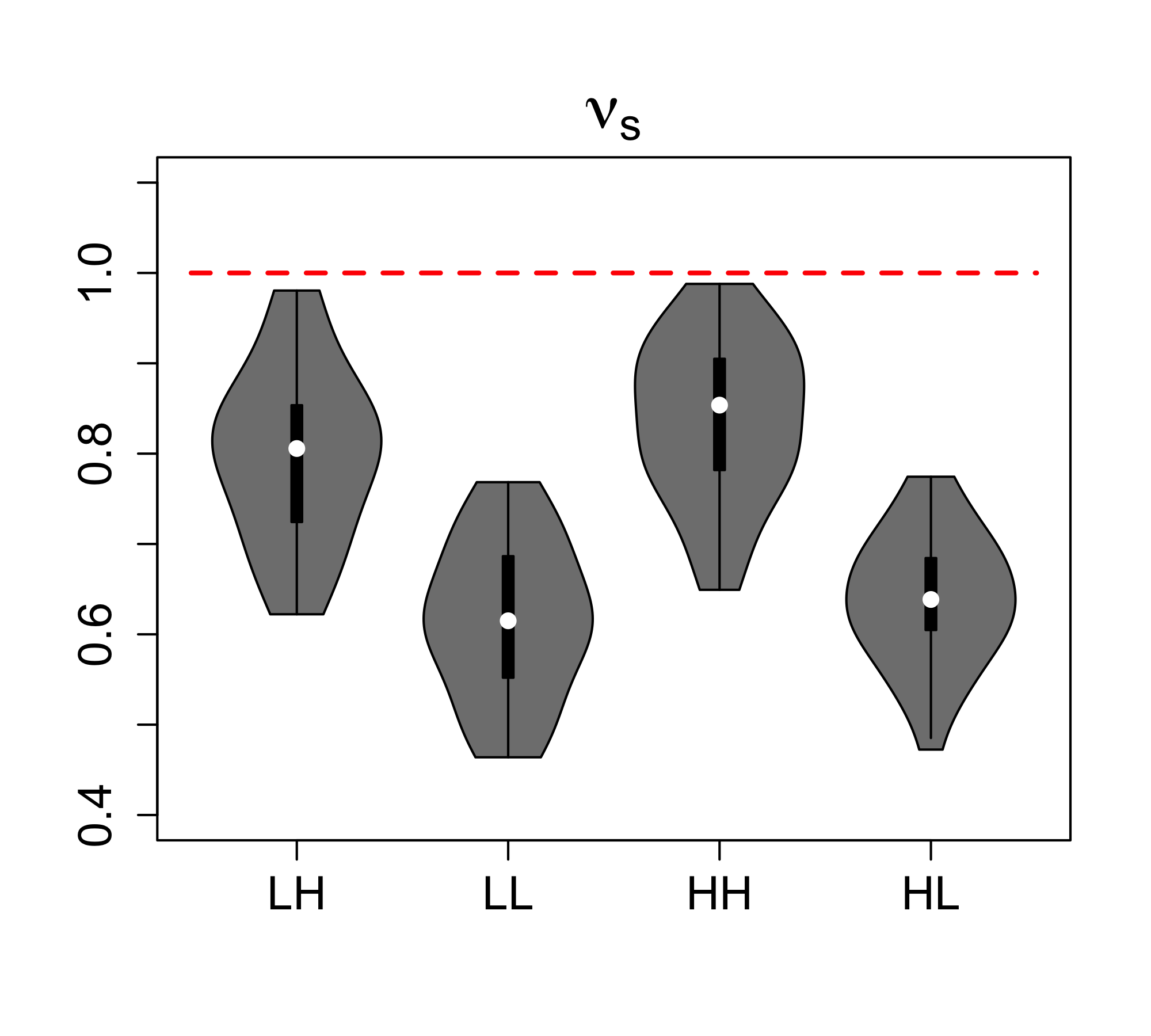}
    \includegraphics[width=0.3\linewidth]{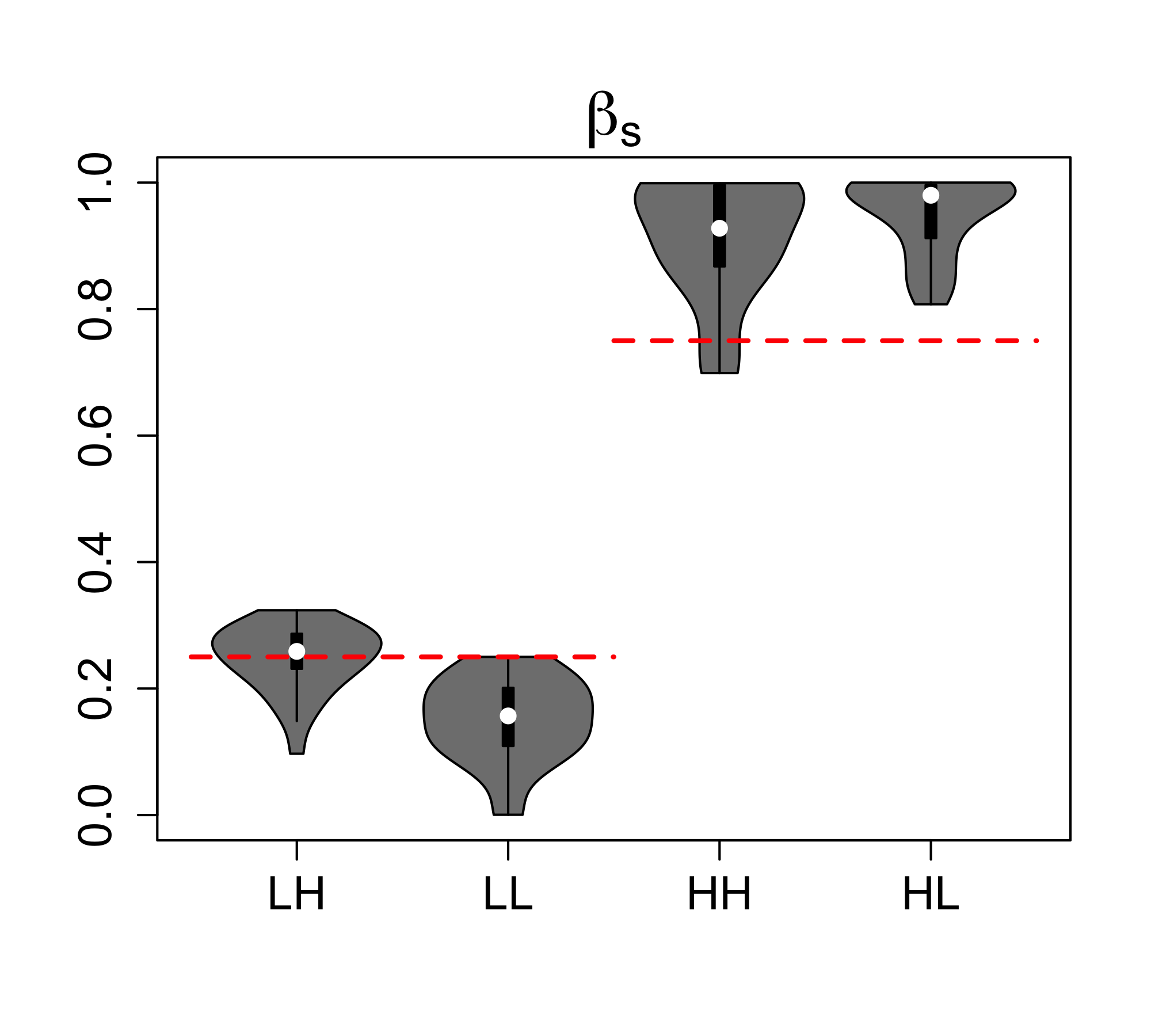}
    \includegraphics[width=0.45\linewidth]{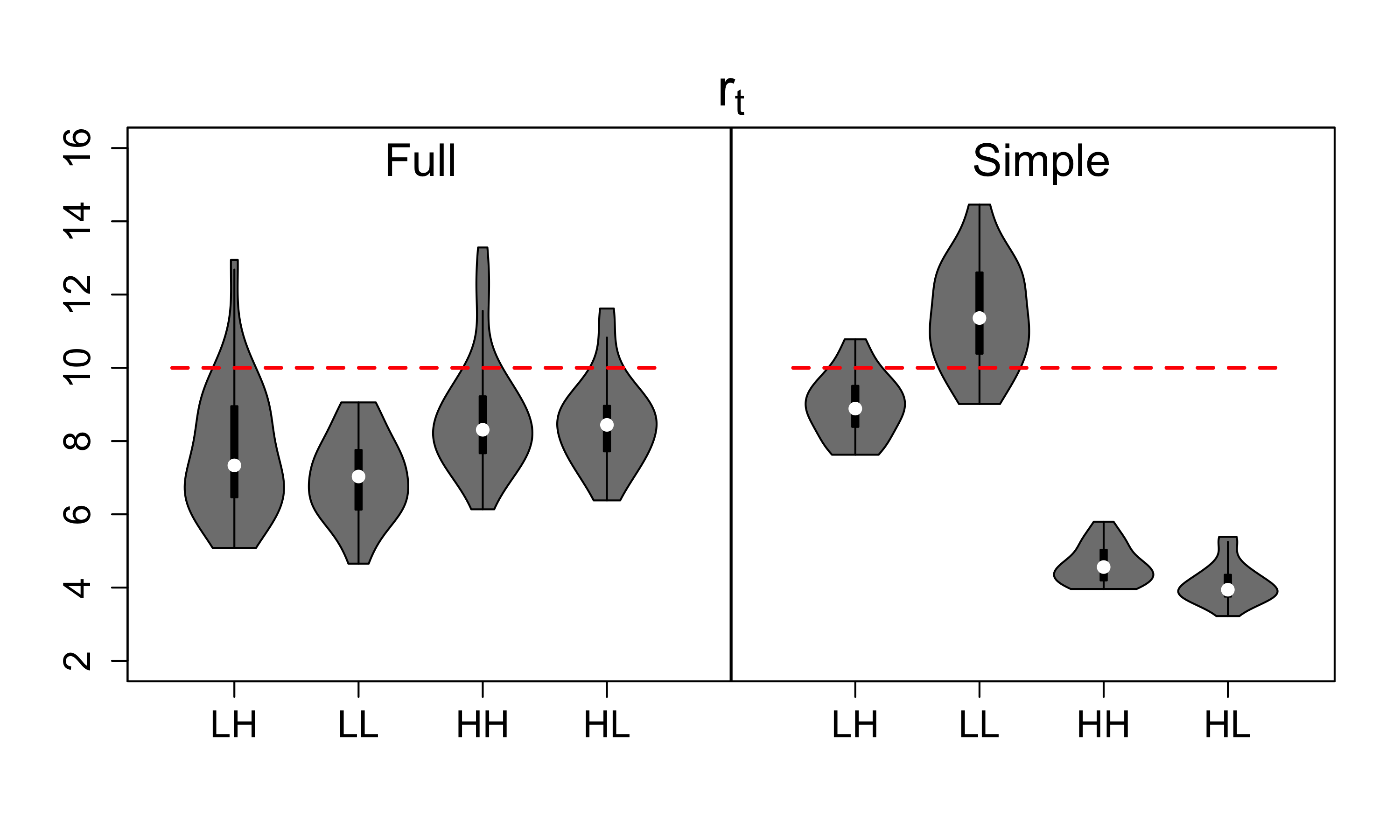}
    \includegraphics[width=0.45\linewidth]{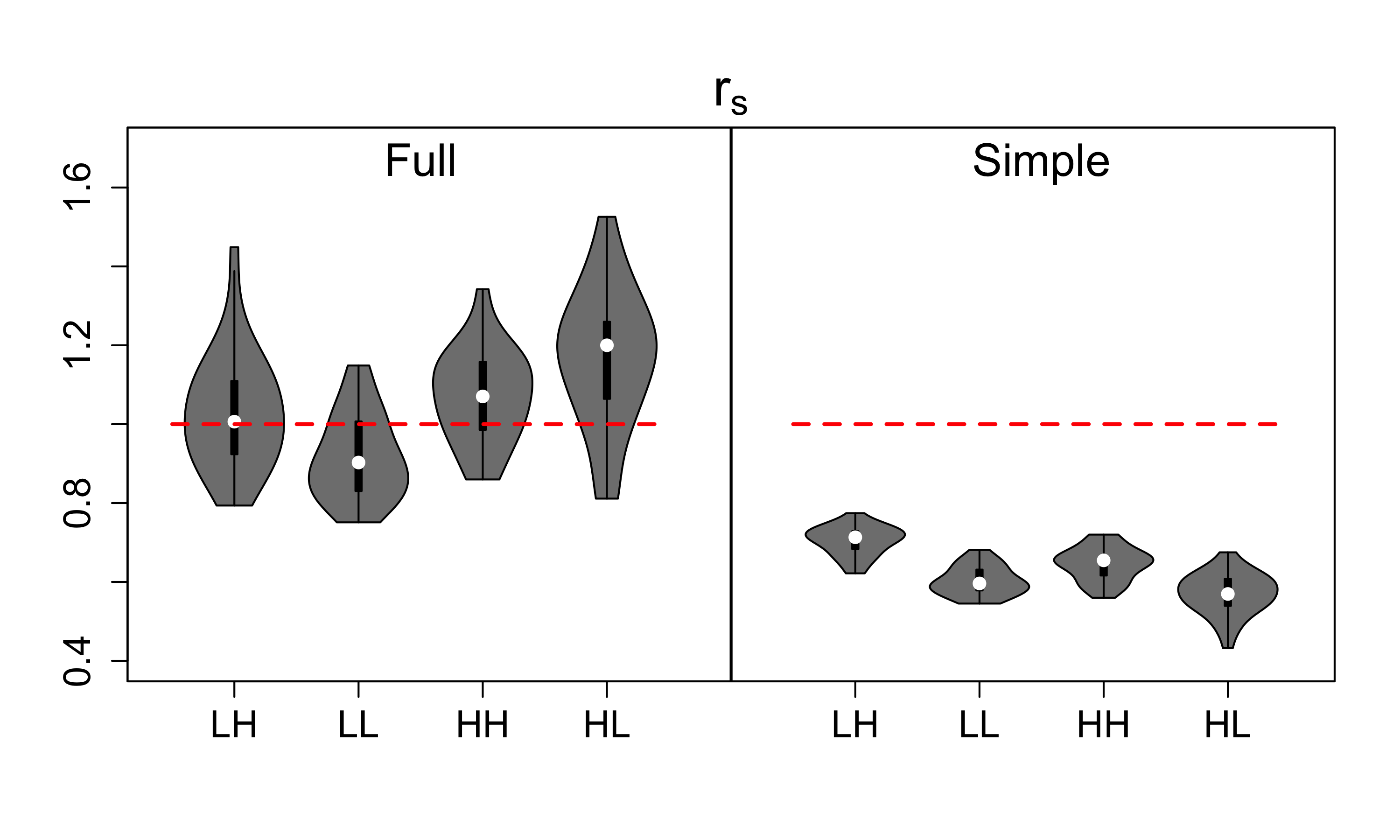}
    \includegraphics[width=0.45\linewidth]{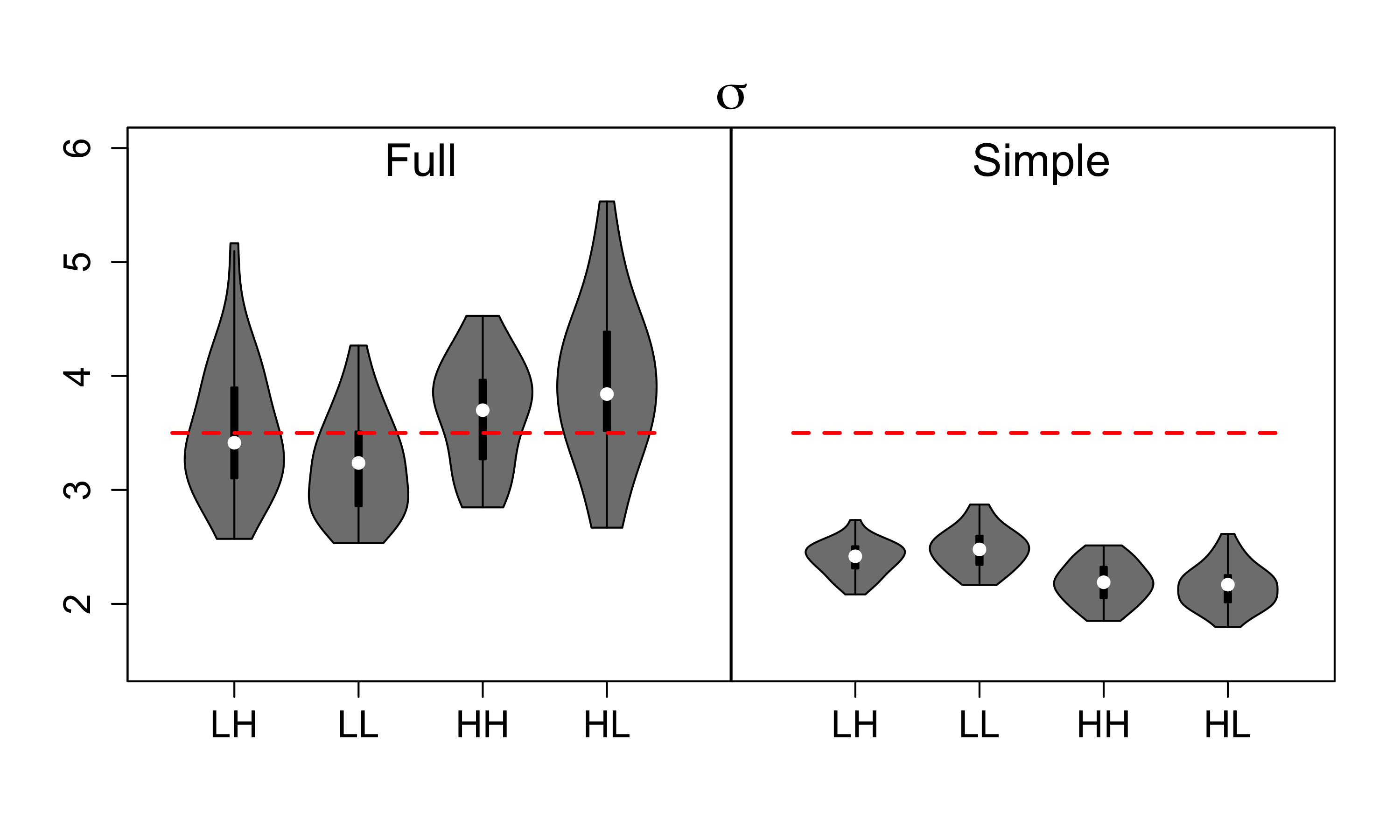}
    \includegraphics[width=0.45\linewidth]{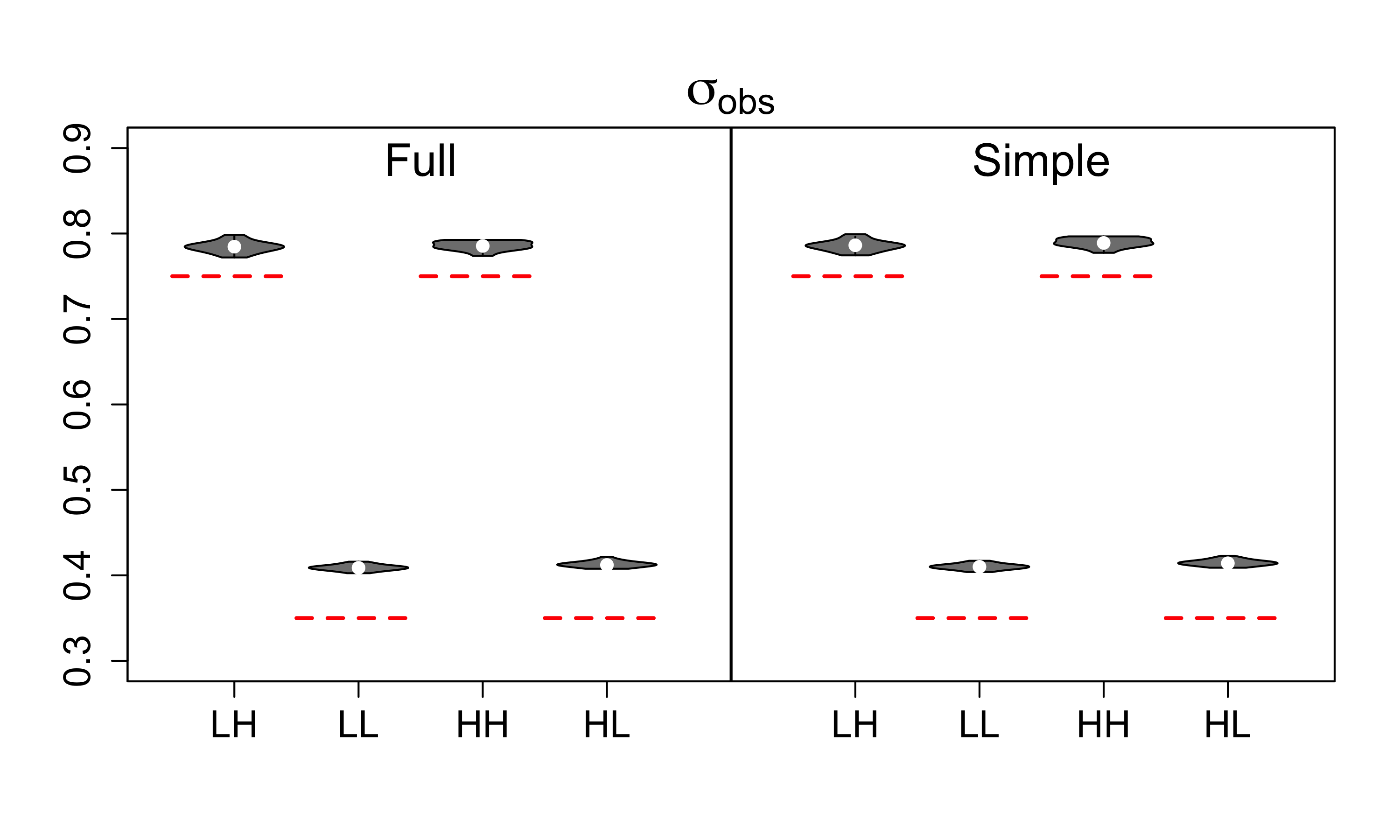}

    \caption{Violin-plots of the estimated parameters over the $30$ replicates, for both the 'Full' model and the 'Simple' model. Top row only includes 'Full' model since these parameters are not estimated in the 'Simple' model. The true value of the parameter is marked in red.\label{fig:parameter-estimates}}
\end{figure}

\begin{figure}
    \centering
    \begin{tabular}{>{\centering\arraybackslash}m{0.04\textwidth} >{\centering\arraybackslash}m{0.45\textwidth} >{\centering\arraybackslash}m{0.45\textwidth}}
         & Filtered predictions & Forecasted predictions \\ 
         \rotatebox[origin=c]{90}{RMSE} & \includegraphics[width=\linewidth]{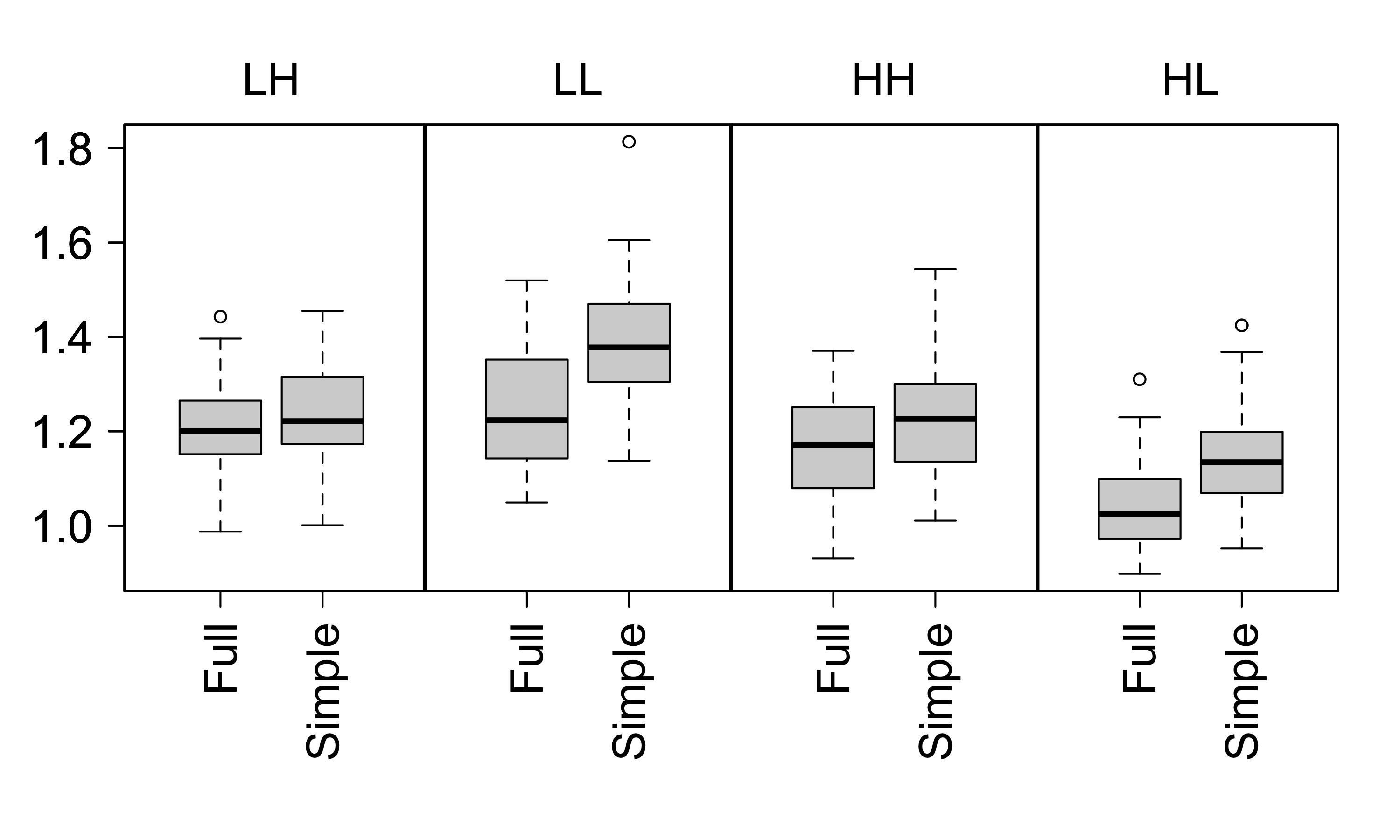} & \includegraphics[width=\linewidth]{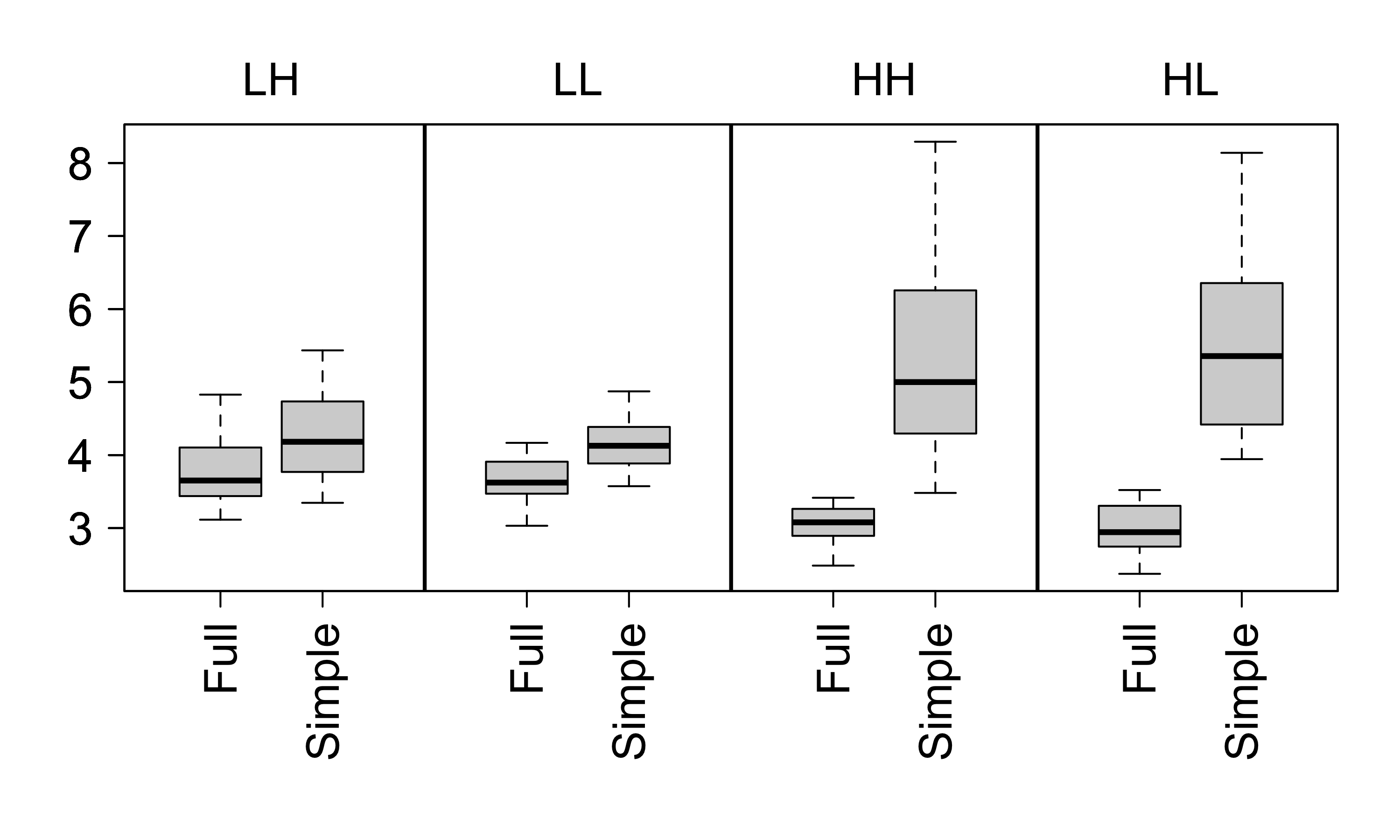} \\
         \rotatebox[origin=c]{90}{CRPS} & \includegraphics[width=\linewidth]{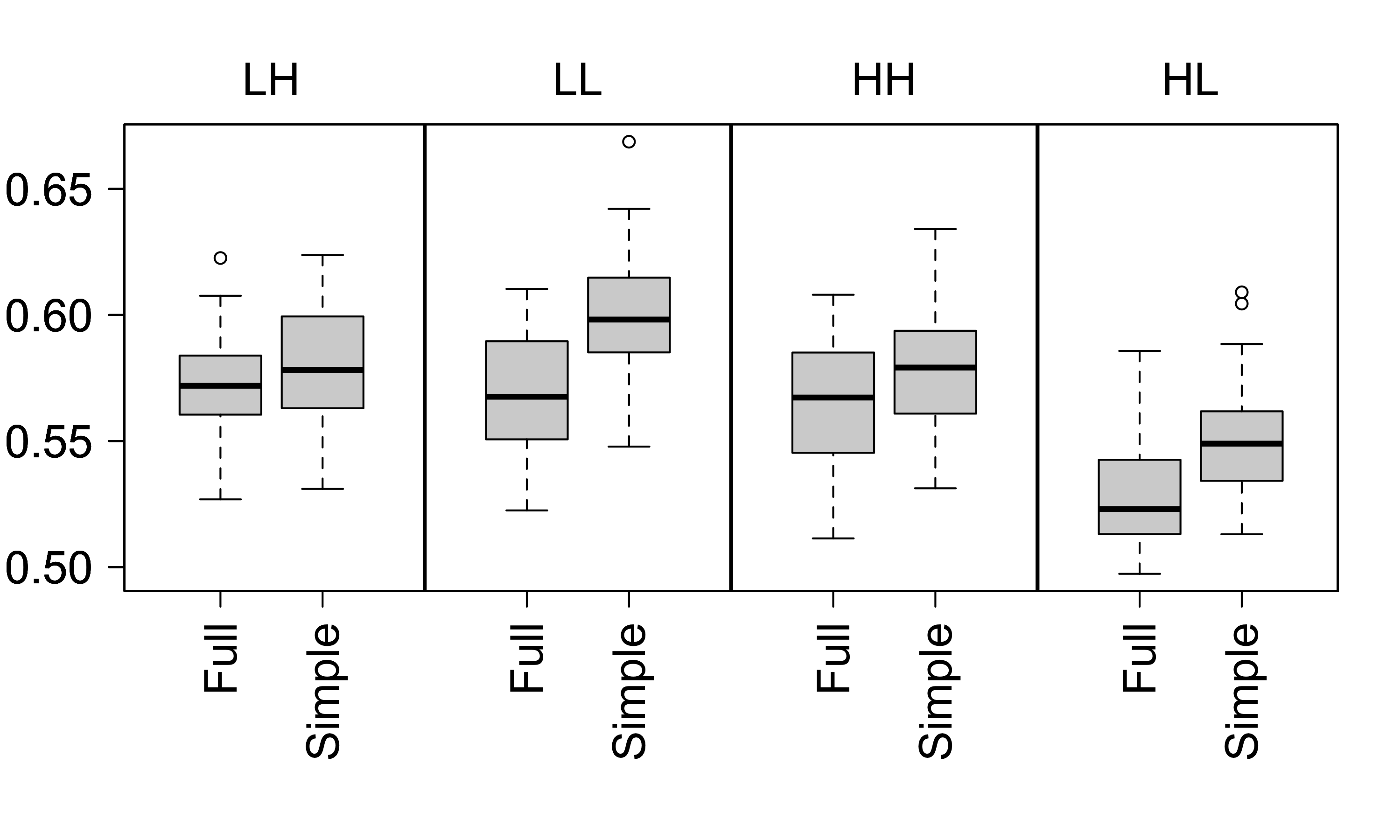} & \includegraphics[width=\linewidth]{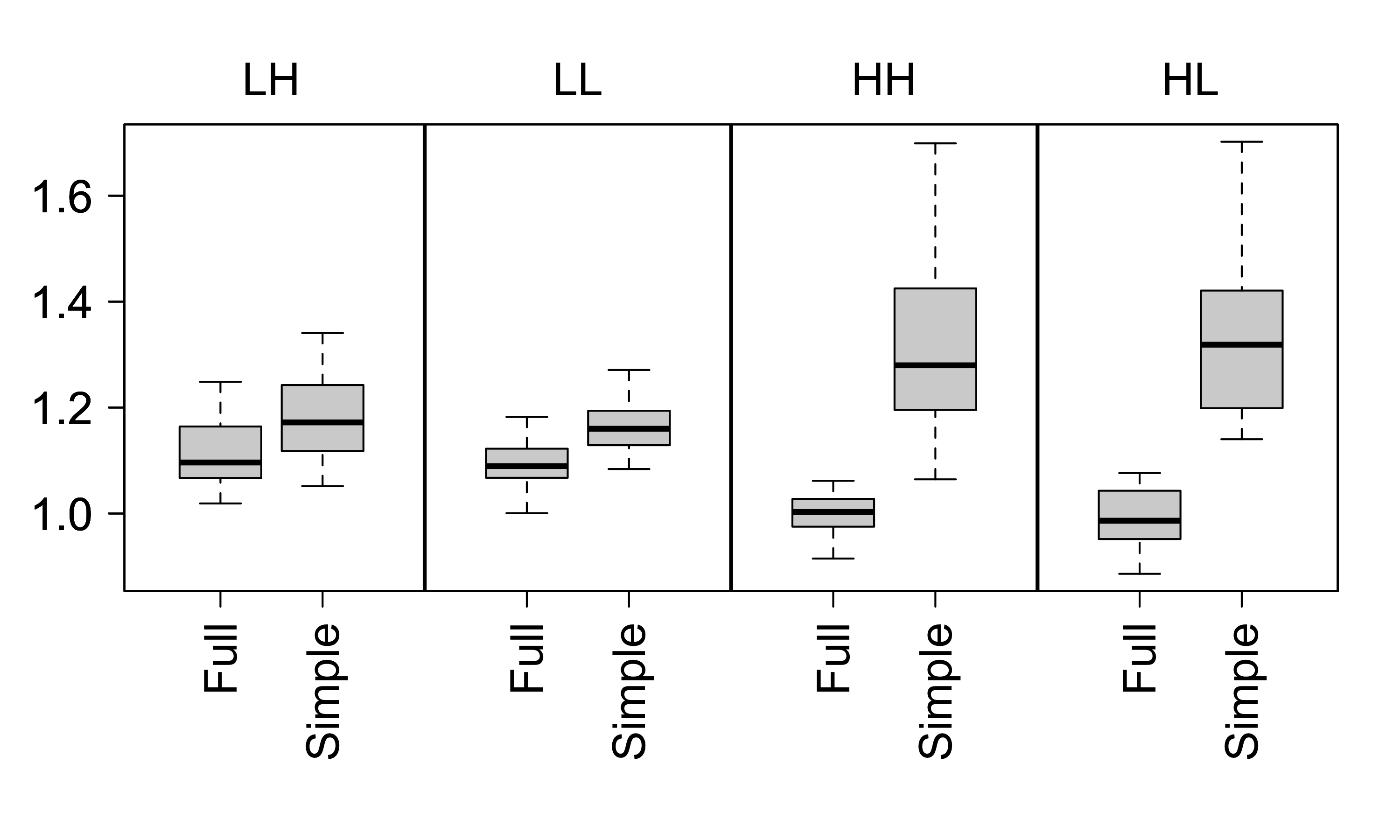} \\
    \end{tabular}
    \caption{Box-plots of RMSE and CRPS of both filtered and forecasted predictions on new datasets with $30$ replicates, for both the 'Full' model and the 'Simple' model.\label{fig:sim-predictions}}
\end{figure}

Violin plots of the inferred parameters over the 30 realizations are shown in Figure \ref{fig:parameter-estimates}. A key observation for the 'Full' model is that we separate well between the scenarios of low and high non-separability as estimated by $\beta_\mathrm{s}$. Further, there is a high spread in the estimates for the smoothness parameters $\nu_\mathrm{s}$ and $\nu_\mathrm{t}$, and we observe overestimation of temporal smoothness and underestimation of spatial smoothness. The estimate of $\nu_\mathrm{s}$ is better for higher measurement error variance, while the estimates for $\nu_\mathrm{t}$ are better for higher non-separability. The spatial range $r_\mathrm{s}$ is estimated well, and the temporal range $r_\mathrm{t}$ is slightly underestimated. Estimation of these two range parameters behave similarily under the different scenarios. The standard deviation of the latent spatio-temporal process $\sigma$ is estimated well across the scenarios, but there is some overestimation of the measurement error standard devation $\sigma_{\mathrm{obs}}$. The latter can be explained by the fact that the truth was simulated under higher spatial resolution than the model used for inference. This extra spatial variation is likely captured by $\sigma_{\mathrm{obs}}$.

Under the 'Simple' model, we perform similar in terms of estimating $\sigma_{\mathrm{obs}}$, but $\sigma$ is underestimated. The spatial range $r_{\mathrm{s}}$ has high bias in all scenarios, and $r_\mathrm{t}$ is underestimated when the non-separability is high. Note that discrepancies are expected as the 'Simple' model is using an incorrect temporal smoothness, which follows from the AR(1) structure, and does not allow for non-separability.

Variation in RMSE and CRPS for filtering and one-step ahead forecasting on the separate test data sets using the estimated parameters from the training datasets are shown in Figure \ref{fig:sim-predictions}. As expected we are better able to discriminate between the models in the scenarios with higher signal-to-noise ratios. Interestingly, low or high non-separability does not seem to make much of a difference for the predictive ability for filtering, but the 'Full' model is considerably better than the 'Simple' model for one-step ahead forecasting when non-separability is high. In all cases, the 'Full' model is equally good or better measured both in RMSE and in CRPS.

\section{Application to daily temperature in France}\label{sec:application}
\subsection{Data analysis}
We consider daily mean temperatures (00:00--23:59) in France, from January 1, 2023, to March 31, 2023. Mainland France has an area of approximately $644000 \, \, \text{km}^2$ with varying elevation as shown in Figure \ref{fig:france-data}. The daily mean temperatures dataset\footnote{Non-blended ECA dataset, daily mean temperature TG} is openly available from the website of the European Climate Assessment \& Dataset project, see the \hyperref[sec:data-availability]{Data availability statement}. The website provides temperature measurements from 10023 stations all over Europe, in the time period January 1, 1756, to December 31, 2025. We have filtered out only those stations that are located in mainland France, leaving 936 stations, with $84097$ total measurements in the selected time window. All maps and models use Lambert-93 coordinates.

Typical spatial variation in daily mean temperatures is shown for January 1, February 1 and March 1 in Figure \ref{fig:france-data}, and Figure \ref{fig:france-temporal-predictions} shows the temporal evolution for two selected measurement stations. The covariates are also displayed in Figure \ref{fig:france-data}. Elevation is measured in kilometres above sea level, and ranges across the stations from $0$ to $3.85$. Proximity to ocean ranges from $0$ to $1$, and is defined to be $1$ at the coast, and decays exponentially to $0$ as a function of the distance to the nearest piece of coastline, $d$ by the transformation $p(d) = 0.05^{d/200}$, such that the covariate is $0.05$ when $200$ km inland. Latitude is measured in degrees, ranging from $42.40$ to $50.96$.

\begin{figure}
    \centering
    \includegraphics[width=0.9\linewidth]{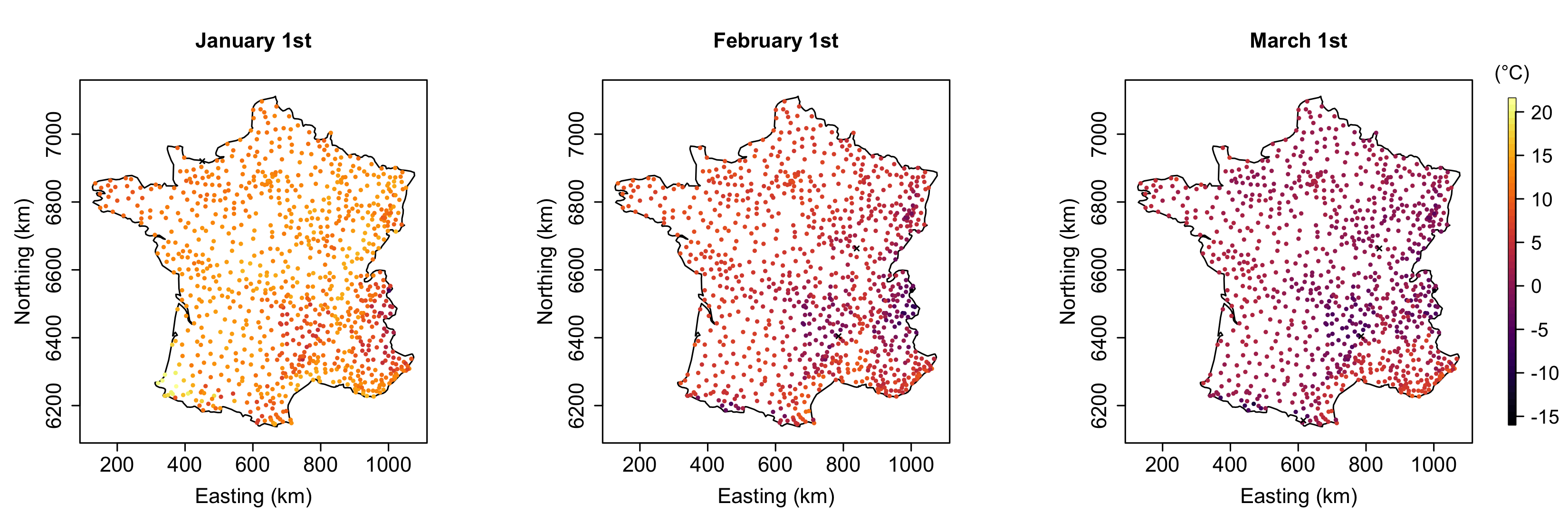}
    \includegraphics[width=0.9\linewidth]{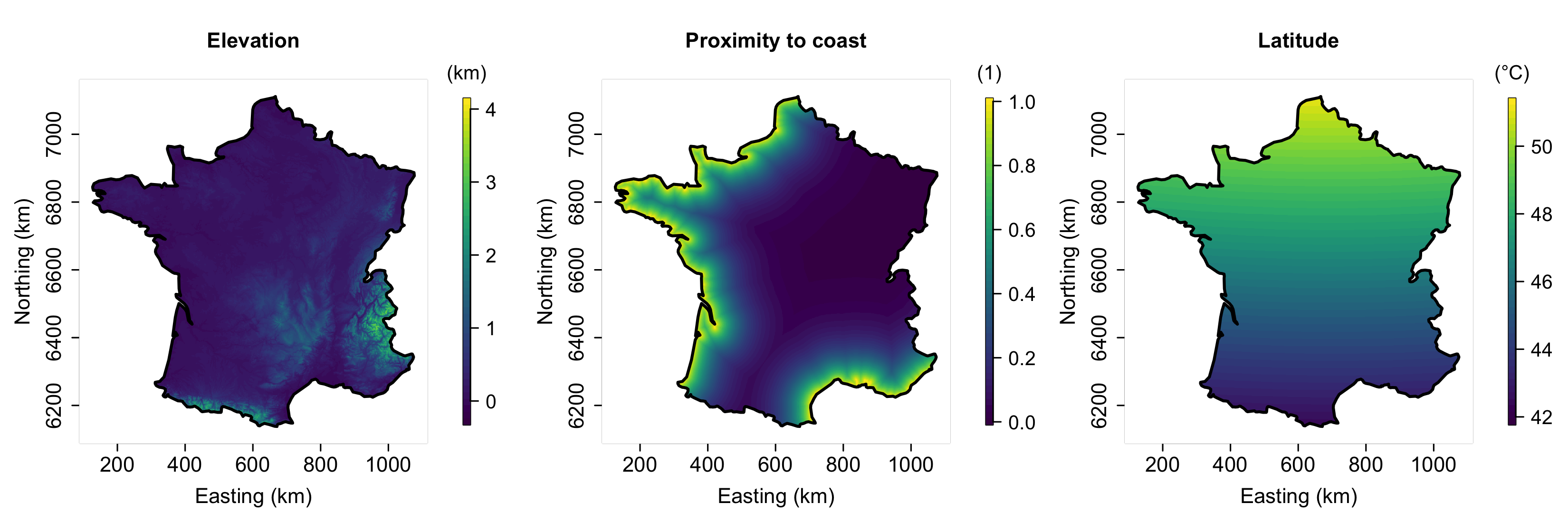}
    \caption{Top: Daily mean temperature measurements in France for January 1, February 1 and March 1, 2023. Bottom: The three covariates we include in our model. Elevation (left), proximity to ocean (middle), and latitude (right).}
    \label{fig:france-data}
\end{figure}

\begin{figure}
    \centering
    \includegraphics[width=\linewidth]{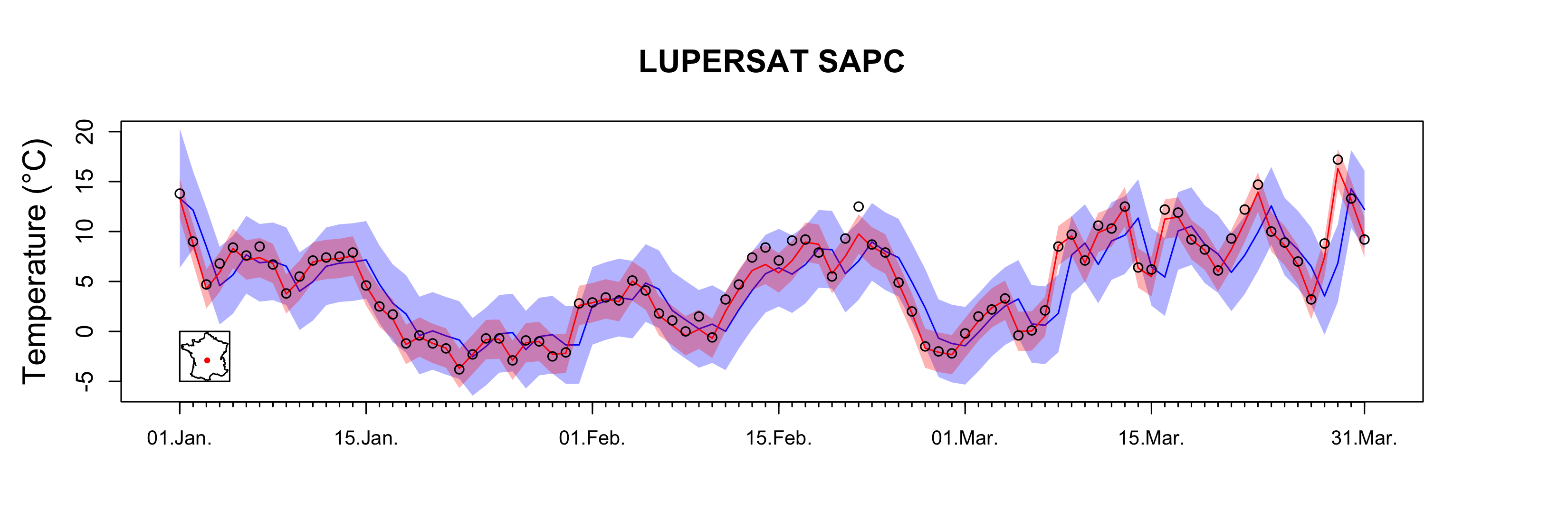}
    \includegraphics[width=\linewidth]{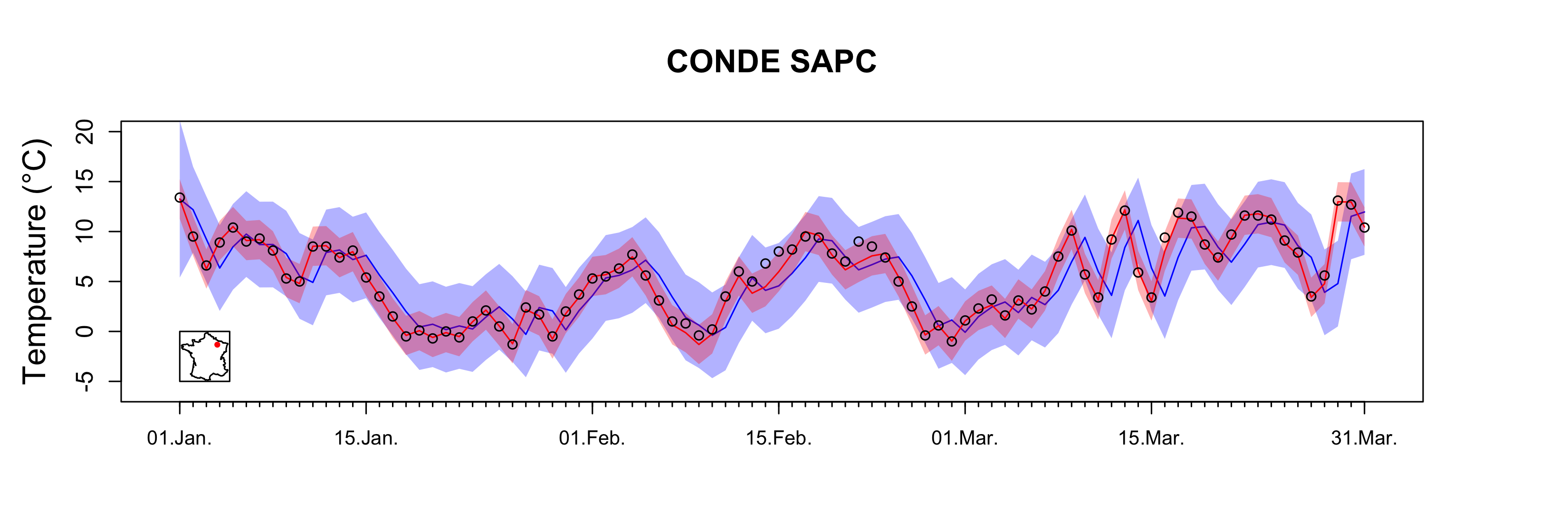}
    \caption{Daily mean temperatures in France from January 1 to March 31 at two arbitrarily selected stations. Blue line displays one step forecasts with corresponding $95\%$ blue shaded prediction region. Red line displays filtered predictions with corresponding $95\%$ red shaded prediction region. The temperature measurements are displayed as black points. The approximate location of the station is displayed in the bottom left corner.}
    \label{fig:france-temporal-predictions}
\end{figure}

The observations suggest that there are no substantial temporal trends or clear heterogeneities over the time period selected. We consider the square spatial domain $\mathcal{D}=[90050, 1113553]\times[6112891, 7136394]$ (in unit meters), and model daily mean temperatures as in Section \ref{sec:hierFull} with signal
\[
    \mu(\boldsymbol{s},t) = \beta_{0} + g_1(\boldsymbol{s})\beta_1 + g_2(\boldsymbol{s})\beta_2 + g_3(\boldsymbol{s})\beta_3 + u(\boldsymbol{s}, t), \quad \boldsymbol{s}\in\mathcal{D},
\]
where $g_1(\cdot)$, $g_2(\cdot)$, and $g_3(\cdot)$ are respectively elevation, proximity to coast, and latitude, $\beta_0$ is the intercept, and $\beta_1$, $\beta_2$ and $\beta_3$ are the effects of the covariates, and $u$ is the non-separable spatio-temporal GRF. Let the $n_{\mathrm{obs}} = 936$ observation locations be given by $\boldsymbol{s}_1, \ldots, \boldsymbol{s}_{n_{\mathrm{obs}}}\in\mathcal{D}$ and denote the $T = 90$ time steps by $n = 1, \ldots, T$. We assume conditionally independent measurement errors $y_i^n | \mu(\boldsymbol{s}_i, t_n), \sigma_{\mathrm{obs}}^2 \sim \mathcal{N}(\mu(\boldsymbol{s}_i, t_n), \sigma_{\mathrm{obs}}^2)$ for $i = 1, \ldots, n_{\mathrm{obs}}$ and $n = 1, \ldots, T$. Note that some observations are missing so that there are only 84097 observations, but that this is excluded from the notation to make it easier to read.

We use a two-step estimation procedure: 1) we estimate $\boldsymbol{\beta} = (\beta_0, \beta_1, \beta_2, \beta_3)^\mathrm{T}$ using ordinary least squares and compute empirical residuals, and 2) estimate the covariance 
parameters $\boldsymbol{\theta}_{\mathrm{cov}}=(\nu_t, \nu_s, \beta_s, r_t, r_s, \sigma, \sigma_{\mathrm{obs}})^\mathrm{T}$ based on the residuals. 
Denote the observation at location $i$ and time step $n$ by $y_i^n$. Then we first fit the regression model
$$y_{i}^n = \beta_{0} + g_1(\boldsymbol{s})\beta_1 + g_2(\boldsymbol{s})\beta_2 + g_3(\boldsymbol{s})\beta_3  + v_i^n,$$
where all $v_i^t$ are i.i.d. $\mathcal{N}(0, \sigma_{\mathrm{Res}}^2)$, where $\sigma_{\mathrm{Res}}^2$ is a nuisance parameter, using least squares. This produced the estimated coefficients $\hat{\beta}_0 = 21.232$, $\hat{\beta}_1 = -4.969$, $\hat{\beta}_2 = 1.559$, and $\hat{\beta}_3 = -0.299$, and empirical residuals $\hat{v}_i^n$.
These residuals were used as data to estimate the covariance structure through treating them as observations
$$\hat{v}_i^n = y^n_i-\hat{\beta}_0-g_1(\boldsymbol{s}_i)\hat{\beta}_1-g_2(\boldsymbol{s}_i)\hat{\beta}_2-g_3(\boldsymbol{s}_i)\hat{\beta}_3=u(\boldsymbol{s}_i, n) + \epsilon_{i}^n \, ,$$
using the procedure described in Section \ref{sec:HierParInf} with spatial resolution $M$. Here the variables $\epsilon_i^n$ are i.i.d. $\mathcal{N}(0, \sigma_{\mathrm{obs}}^2)$.

We estimate $\boldsymbol{\theta}_{\mathrm{cov}}$ with $M = 8^2$, $M = 12^2$, and $M = 16^2$. The parameter estimates are found in Table \ref{table:france-full-parameter-estimates}. As we can see the largest differences lies in differing spatial smoothness and spatial ranges estimated for different resolutions. We also compute filtered estimates for the signal $\mu(\cdot, 32)$ for February 1 using the estimated parameters both for $M = 8^2$, $M = 12^2$, and $M = 16^2$. As shown in Figure \ref{fig:france-spatial-predictions}, the model with higher spatial resolution is able to pick up more fine-scale variation in the spatial field and in the prediction standard deviations, and look visually better than the lower resolution model. We also computed filtered estimates and one-step ahead forecasts for two arbitrarily selected stations using $M = 16^2$ in Figure \ref{fig:france-temporal-predictions}. Note that the prediction intervals plotted in Figure \ref{fig:france-temporal-predictions} are predicting new measurements, not the signal, and hence include the measurement error $\sigma_{\mathrm{obs}}$. 

\begin{figure}
    \centering
    \includegraphics[width=0.9\linewidth]{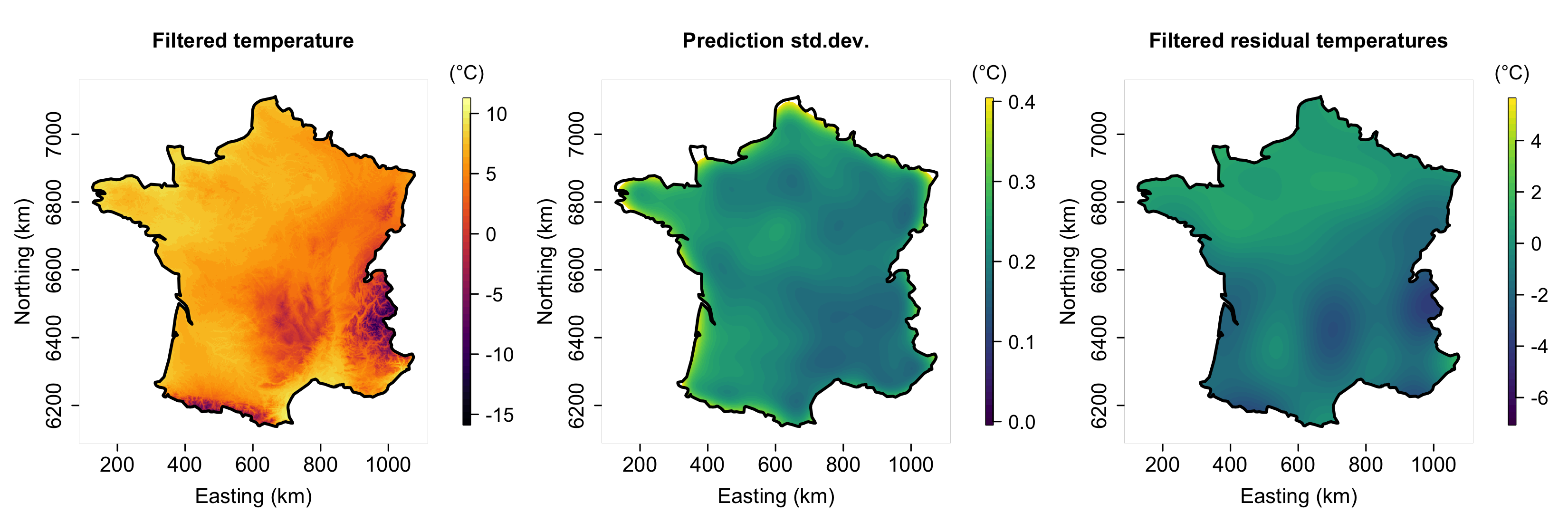}
    \includegraphics[width=0.9\linewidth]{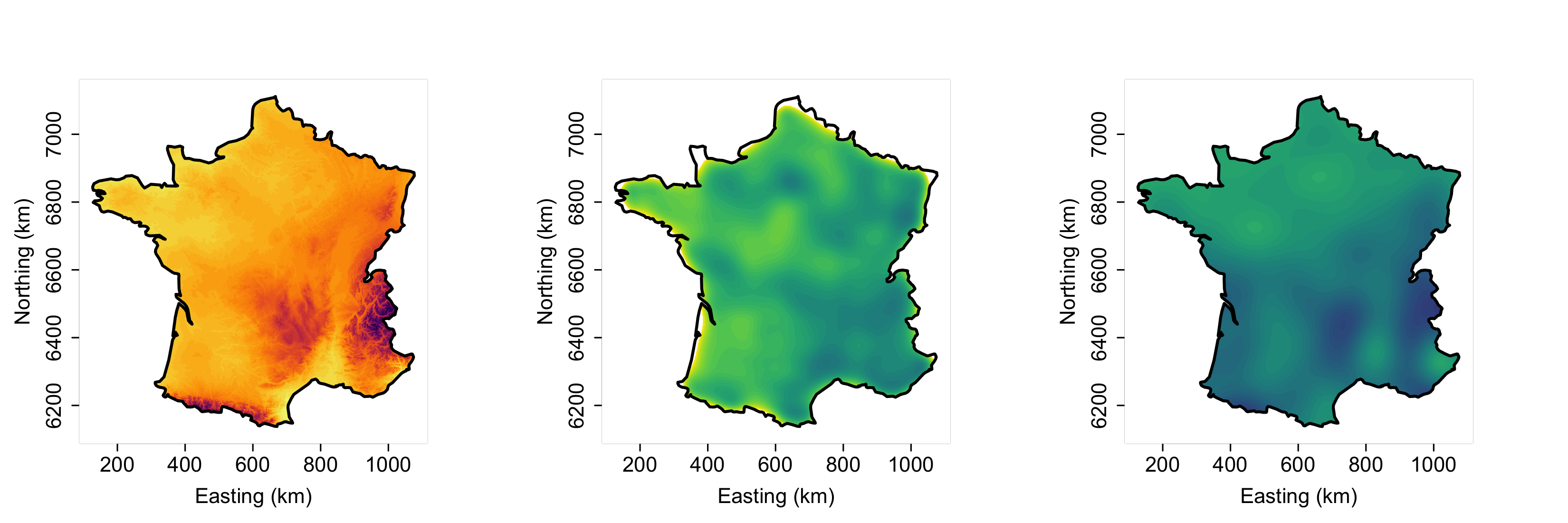}
    \includegraphics[width=0.9\linewidth]{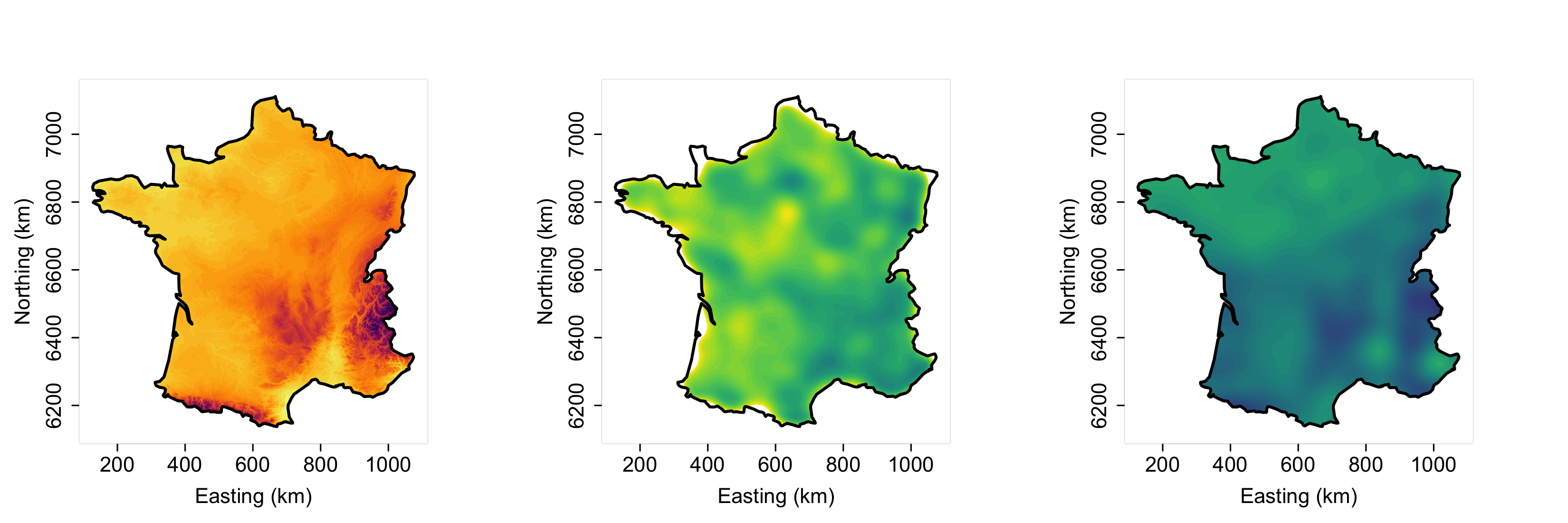}
    
    \caption{Interpolated daily mean temperature in France on February 1, 2023. Filtered temperature means with linear model (left), prediction error (middle), and filtered temperature means without linear model (right). $M = 8^2$ (top), $M^2 = 12^2$ (middle), and $M = 16^2$ (bottom). Axes are in kilometres.}
    \label{fig:france-spatial-predictions}
\end{figure}

\begin{table}[]
    \centering
    \caption{Tabulated parameter estimates for the full dataset.}
    \label{table:france-full-parameter-estimates}
    \begin{tabular}{c|c|c|c|c|c|c|c}
     & $\nu_t$ & $\nu_s$ & $\beta_s$ & $r_t$ (days) &  $r_s$ (km.) & $\sigma$ (°C) & $\sigma_{\mathrm{obs}}$ (°C) \\ \hline
     $M = 8^2$ & 0.477 & 0.457 & 0.002 & 10.568 & 726.554 & 3.086 & 1.026 \\ \hline
     $M = 12^2$ & 0.419 & 0.659 & 0.002 & 17.925 & 711.676 & 3.522 & 0.975 \\ \hline
     $M = 16^2$ & 0.438 & 0.722 & 0.002 & 15.919 & 704.184 & 3.390 & 0.945 \\ \hline
    \end{tabular}
\end{table}

\begin{figure}
    \centering
    \includegraphics[width=0.9\linewidth]{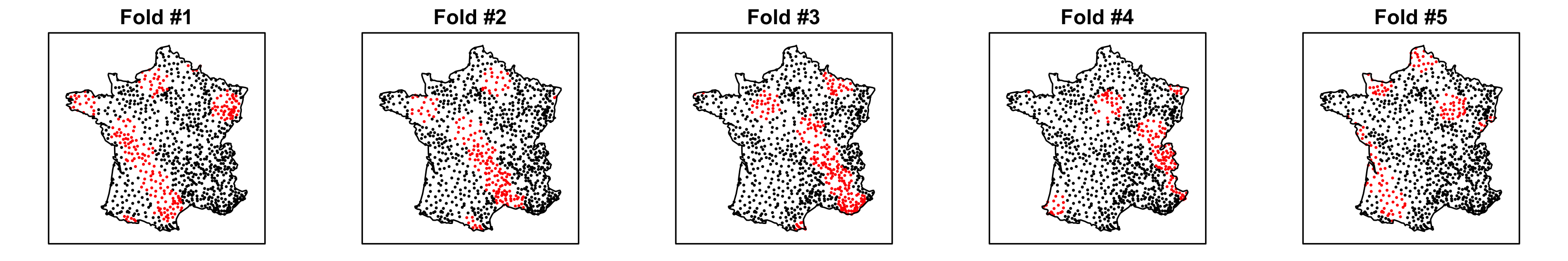}
    \caption{The five folds used for spatial cross validation. Station in the training set are marked in black; stations in the test set are marked in red.\label{fig:cross-validation-folds}}
\end{figure}

\begin{figure}
    \centering
\includegraphics[width=0.9\linewidth]{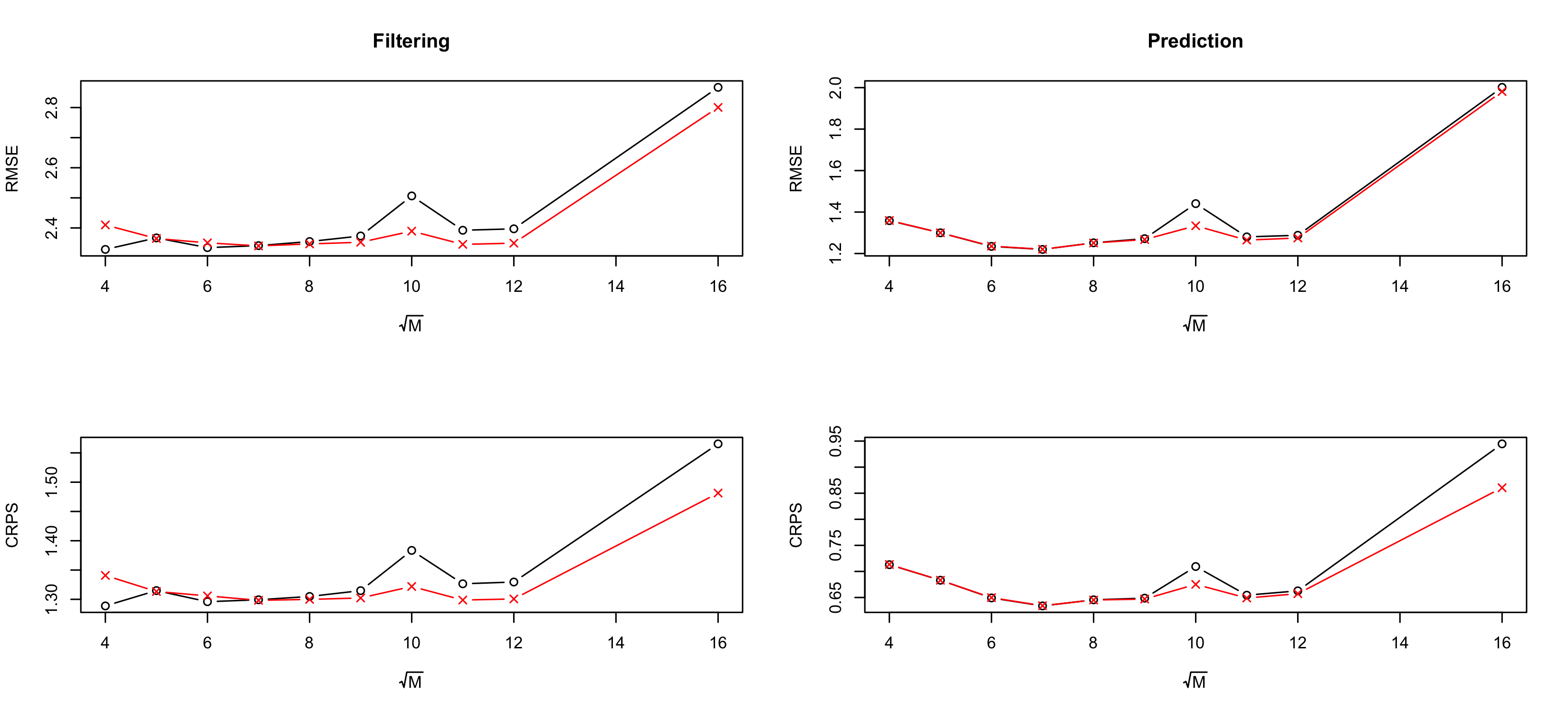}
    \caption{Cross-validation scores for $\sqrt{M} = 4,5,6,7,8,9,10,11,12,16$. "o"-lines are for the 'Full' model. "x"-lines are scores for the 'Simple' model. The same $M$ was used both for parameter inference and prediction.}
    \label{fig:placeholder}
\end{figure}

\subsection{Sensitivity to spatial resolution}
\label{sec:sensitivity}
We conduct five-fold cross-validation for 
$M \in \{4, 5, \ldots, 12, 16\}$, where we in each case compare the 'Full' model to the 'Simple' model with fixed $\nu_\mathrm{t} = 0.5$, $\nu_\mathrm{s} = 0.5$, and $\beta_\mathrm{s} = 0.0$. The choice of $\nu_\mathrm{t}=0.5$ and $\beta_\mathrm{s}=0$ is again motivated by the the fact that this yields an AR(1) process in time driven by spatially coloured noise, whereas $\nu_\mathrm{s}=0.5$ was chosen so that the spatial smoothness is that of an exponential GRF. The five folds correspond to regions selected using the function \verb|cv_spatial| from the \textsf{R} package \texttt{blockCV}, using \verb|selection = "systematic"| and \verb|size = 150000| (m. in Lambert-93). This gives the five folds with spatial structure shown in Figure \ref{fig:cross-validation-folds}. 

Parameter estimates for the fixed effects and for the covariance parameters, for each fold and spatial resolution, are given in Section \ref{suppl:case-study-parameter-estimate-tables} in the Supplementary Materials. For each fold, RMSE and CRPS are computed across the hold-out locations for filter distributions and for forecast distributions. They are combined to a single RMSE and CRPS by a mean square average and an arithmetic mean, respectively, and the results are plotted against spatial resolution in Figure \ref{fig:placeholder}. RMSE and CRPS scores for each fold are given in Section \ref{suppl:case-study-cv-prediction-score-tables} in the Supplementary Materials. Note that since we are predicting new measurements and not the signal, the standard deviation used in the computation of the CRPS score includes the estimated measurement error $\hat{\sigma}_{\mathrm{obs}}$. 

For this dataset, we estimate a non-separability close to 0 and a spatial smoothness a bit below 0.5, and so it is not unexpected that the 'Simple' model will perform well. For coarser resolutions, the 'Simple' model is slightly better with very small difference for filtering. However, the 'Full' model gives the advantage that we are able detect the degree of non-separability and that we can estimate the fractional smoothnesses in space and time. When the resolution is $\geq10^2$ the parameter inference tends towards extreme values in the temporal range, with detrimental effect on predictions. The source of this effect is not clear and further study is needed in future work. When the parameters were fixed to the estimated values for resolution $M = 8^2$, we did not observe any adverse effects on predictive performance when increasing the resolution and the predictive performance remained similar as for $M = 8^2$.

\section{Discussion}\label{sec:Discussion}
The proposed discretization approach extends the practical use of the diffusion-based model proposed by \citet{LindgrenBakka} to all smoothnesses in space and time. The arbitrary smoothness in space is handled by a traditional spectral approximation, and the different smoothnesses in time are achieved through a novel ARMA approximation. The discretization is fast for simulation since the different frequencies can be simulated independently of each other. We propose a sequential Kalman filter for parameter inference and prediction that has a complexity that is the square of the spatial resolution times the number of observations for a fixed order of the rational approximation.  

In the simulation study, we are able to reliably identify the degree of non-separability in the model that generated the observations. In terms of prediction, the largest improvement over a separable model is seen for one-step ahead forecasts, and in that case the gain is highest when the truth is generated from a model with high non-separability. This suggests that to achieve the full benefit of these models, we need to discover which applications exhibit this type of non-separability by exploring different datasets with the new approach.

The application to daily mean temperature shows that it is computationally feasible to estimate the model across three months and almost 1000 measurement stations. By using the proposed approach, we are able to estimate smoothnesses that are not locked to predetermined values such as in \citet{LindgrenBakka}. We find that the non-separability is estimated low, and that we are not able to improve the predictions over the simpler non-separable model. However, for moderate spatial resolutions, the non-separable model has comparable predictive performance to the separable model, in terms of RMSE and CRPS, but there are still issues to be explored for higher spatial resolutions in the hold-areas-out cross-validation. Since non-separability is estimated low, it is not unexpected that the non-separable model would not improve predictions, but it was through the estimation of the model we were able to diagnose the lack of this type of non-separability in the daily mean temperatures. 

Further study is needed to determine under which conditions the approach can be expected to be robust. Specifically, the model has a tendency to overestimate the temporal range when the number of spatial basis functions grows large, with detrimental effect on predictions. It is not clear if this is due to a numerical error or if it is a form of overfitting. We conjecture that this effect is caused, at least to some degree, by misspecification in the spatial model. If so, it could possibly be remedied by introducing a fixed spatial effect, or a fixed random effect for each measurement location. 

This paper provides a practical way forward for discretizing fractional space-time operators where the covariance structure is not Markovian in time. We expect that the rational approximation introduced in this work can be be applied in more general equations of the form $(\partial_t + A)^{\gamma}u = \partial_t W_Q$ through a temporal semi-discretisation
\begin{equation*}
p(S(\Delta t)B)(1-S(\Delta t)B)^{\lfloor\gamma\rfloor} u(\cdot, t_n) = q(S(\Delta t)B)\Delta W_Q(\cdot,t_n),
\end{equation*}
where $S(t)$ is the analytic semi-group generated by $-A$. In this more general setting, one could replace the spectral approximation used in this work by a different spatial discretisation, e.g., a finite element method (FEM) discretisation. Advantages of FEM or surface FEM would be that one does not need to know the eigenfunctions and eigenvalues, and that there exists software for FEM that could be leveraged. A concrete example of interest is to replace $(\kappa^2-\Delta)$ by $(\kappa(\boldsymbol{s})^2-\nabla\cdot\mathbf{H}(\boldsymbol{s})\nabla)$, where $\kappa$ is a spatially varying function and $\mathbf{H}$ is a spatially varying diffusion tensor, to achieve a non-stationary spatial covariance structure with spatially varying anisotropy such as in \cite{fuglstad2015exploring,fuglstad2015does}.
We expect that the techniques used in the proof of Theorem \ref{thm:method-convergence} to establish convergence of the rational approximation could, with some work, be adapted to also provide error bounds in this more general setting.

Overall, this paper moves the discretization methods for the diffusion-based extension of the Matérn model presented by \citet{LindgrenBakka} and \citet{kirchner2022} a significant step forward, both in terms of theory and methods, by allowing all fractional smoothnesses in space and time, and allowing estimation and prediction with high temporal resolution.

\section{Data availability statement}\label{sec:data-availability} The data used in the case study in Section \ref{sec:application} is freely available at \url{https://www.ecad.eu/dailydata/predefinedseries.php}. Supporting code for the simulation study is available at \url{https://zenodo.org/records/19880427}.

\section{Declarations} 
\noindent
\textbf{Funding} \, GAF and ERJ were supported by the Research Council of Norway under Grant Agreement No. 325114 “IMod. Partial differential equations,
statistics and data: An interdisciplinary approach to data-based modelling”.

\noindent
\textbf{Conflicts of interest} \, The authors declare no conflicts of interest.

\bibliography{wileyNJD-AMA}

\newpage
\begin{titlepage}
    \centering
    \vspace*{5cm}
    
    {\Huge \textbf{Supplementary Materials}} \\[1.5cm]
    {\Large Proofs, tables, and other supporting material for \textit{ARMA approximation of a Non-separable Spatio-Temporal Model with Fractional Smoothnesses in Space and Time}} \\[2cm]
    
    \textbf{S. Knutsen Furset, Geir-Arne Fuglstad, and Espen R.
    Jakobsen} \\
    
    \vfill
\end{titlepage}
\newpage

\appendix
\renewcommand{\thesection}{S} 

\setcounter{table}{0}
\renewcommand{\thetable}{S\arabic{table}}
\setcounter{figure}{0}
\renewcommand{\thefigure}{S\arabic{figure}}
\renewcommand{\theequation}
{S\arabic{equation}}
\setcounter{equation}{0}
\setcounter{theorem}{0}
\setcounter{remark}{0}

\section{Supplementary Materials}

\subsection{Proof of Theorem \ref{thm:weak-spectral-convergence}} \label{suppl:proof-thm-1}

\begin{proof}[Proof of Theorem \ref{thm:weak-spectral-convergence}]

Applying the Sogge bound $\|f_k\|_\infty \lesssim \xi_k^{\frac{d-1}{4}}$ (\cite{Sogge1988, Smith2007}), and the Weyl bound $\xi_k \lesssim k^{\frac{2}{d}}$ (\cite{davies1995}, \cite{Lablee2015-pi}),  we get from \eqref{eq:time-process-covariance-function}, \eqref{eq:process-covariance-function}, and \eqref{eq:spectral-approximation}, that 
\begin{align*}
& \, \quad \left| C_{u}(\boldsymbol{s}_1, \boldsymbol{s}_2, h) - C^M_{u}(\boldsymbol{s}_1, \boldsymbol{s}_2, h)\right| \\
&\leq \sum_{k = M + 1}^\infty \left|f_k(\boldsymbol{s}_1)\,f_k(\boldsymbol{s}_2) \right| \frac{\lambda_k \mathrm{e}^{-\mu_k |h|}}{(2 \mu_k)^{2 \gamma - 1} \Gamma(\gamma)^2} \int_0^\infty (u + 2 \mu_k |h|)^{\gamma - 1} u^{\gamma - 1} \mathrm{e}^{-u} \, \mathrm{d}u \\
& \leq \frac{\Gamma(2 \gamma - 1)}{\Gamma(\gamma)^2 2^{2 \gamma - 1}} \sum_{k = M + 1}^\infty \xi_k^{\frac{d - 1}{2}} \lambda_k \mu_k^{1 - 2\gamma} \lesssim \sum_{k = M + 1}^\infty k^{\frac{d - 1}{d} - \frac{2}{d}(\beta + \alpha (2\gamma - 1))} \\
& \leq \int_{M}^\infty x^{\frac{d - 1}{d} - \frac{2}{d}(\beta + \alpha (2\gamma - 1))} \dint x = \frac{M^{1 + \frac{d - 1}{d} - \frac{2}{d}(\beta + \alpha (2\gamma - 1))}}{1 + \frac{d - 1}{d} - \frac{2}{d}(\beta + \alpha (2\gamma - 1))} \propto M^{-\frac{2 \nu_s}{d} + \frac{d - 1}{d}} \,,
\end{align*}
where $\nu_s = \beta + \alpha (2\gamma - 1) - \frac{d}{2}$. On rectangular domains the eigenbasis of the Laplacian with Neumann (Dirichlet) boundary conditions is a product of cosine (sine) functions  with 
$\|f_k\|_\infty \lesssim 1$. A recomputation then leads to  the improved bound
\begin{align*}
\left| C_{u}(\boldsymbol{s}_1, \boldsymbol{s}_2, h) - C^M_{u}(\boldsymbol{s}_1, \boldsymbol{s}_2, h)\right| &\lesssim M^{-\frac{2 \nu_s}{d}} \, . \qedhere
\end{align*}
\end{proof}

\subsection{Proof of Theorem \ref{thm:method-convergence}} \label{suppl:proof-thm-2}

To prove Theorem \ref{thm:method-convergence}, we first need a few lemmas.

\begin{lemma} \label{lem:frac-binom-coef-convergence}  Assume $\eta>0$ and $\ell\in\mathbb N$. Then  the fractional binomial coefficient, 
$\binom{\eta}{\ell} := \frac{\eta \cdot (\eta - 1)\cdot \dots \cdot(\eta
- \ell + 1)}{\ell!} \, ,$
satisfies the bound
$$ \left|(-1)^\ell \binom{-\eta}{\ell} \Gamma(\eta) \ell^{1 - \eta} - 1 \right| \lesssim \ell^{-1} \, .$$
\end{lemma}
\begin{remark}
It is well-known that $(-1)^\ell \binom{-\eta}{\ell} \Gamma(\eta) \ell^{1 -\eta} \rightarrow 1$  as $\ell\to \infty$, see for example \cite{Levrie2017},  but we were unable to find an explicit error bound in the literature. The bound follows easily from Euler's definition of the gamma function.
\end{remark}
\begin{proof}
 \cite{Euler1738} showed that
$\Gamma(z) = \frac{1}{z} \prod_{i = 1}^{\infty} \frac{1}{1 + \frac{z}{i}} \left( 1 + \frac{1}{i} \right)^{z}$ for $ z\in \mathbb C \setminus \{0,-1,-2,...\}  .$
Let $R(\ell) := \prod_{i = \ell}^{\infty} \frac{1}{1 + \frac{\eta}{i}} \left( 1 + \frac{1}{i} \right)^{\eta}$ and note that
\begin{align*}
    \Gamma(\eta) & = \frac{1}{\eta} \left( \prod_{i = 1}^{\ell - 1} \frac{1}{1 + \frac{\eta}{i}} \left( 1 + \frac{1}{i} \right)^{\eta} \right) R(\ell) 
    = \frac{1}{\eta} \left( \prod_{i = 1}^{\ell - 1} \frac{-i}{-\eta - i} \left( 1 + \frac{1}{i} \right)^{\eta} \right) R(\ell) \, \\
    & = (-1)^{\ell} (\ell - 1)! \left(\frac{1}{-\eta} \cdot \frac{1}{-\eta - 1} \cdot \frac{1}{-\eta - 1} \cdots \frac{1}{-\eta - (\ell - 1)} \right) \left( \frac{2}{1} \cdot \frac{3}{2} \cdots \frac{\ell}{\ell - 1} \right)^{\eta} R(\ell) \, \\[0.2cm]
    & = (-1)^{\ell} \frac{\ell! \ell^{-1+\eta} }{-\eta(-\eta - 1)(-\eta - 2) \cdots (-\eta - (\ell - 1))}\, R(\ell)= (-1)^{\ell}\frac{ \ell^{-1+\eta}}{\binom{-\eta}{\ell}}\, R(\ell)  \, ,
\end{align*}
or equivalently,
$(-1)^\ell \binom{-\eta}{\ell} \Gamma(\eta) \ell^{1 - \eta} = R(\ell)$.  To conclude, we need to show that $|R(\ell) - 1| \lesssim \ell^{-1}$. 
We first show that $\log(R(\ell)) \lesssim \ell^{-1}$. 
By a Taylor expansion, $\log(1 + x) = x - \frac{1}{2} x^2 - h(x) x^2$  where 
$h(x) \rightarrow 0$ as $x \rightarrow 0$. Now take $\delta$ such that $h(x) < 1$ for $|x| < \delta$ and $N$ such that $\max\left(\frac{\eta}{N}, \frac{1}{N} \right) < \delta$. For $\ell \geq N$ we have that
\begin{align*}
\abs{\log(R(\ell))} &\leq \sum_{i = \ell}^{\infty} \abs{-\log \left(1 + \frac{\eta}{i} \right) + \eta \log \left(1 + \frac{1}{i} \right) } \\
& = \sum_{i = \ell}^{\infty} \abs{-\frac{\eta}{i} + \frac{1}{2} \left(\frac{\eta}{i}\right)^2 - h \left(\frac{\eta}{i} \right) \left(\frac{\eta}{i}\right)^2 + \eta \left( \frac{1}{i} - \frac{1}{2} \left( \frac{1}{i} \right)^2 + h \left(\frac{1}{i} \right) \left(\frac{1}{i}\right)^2 \right) } \\
& \leq \sum_{i = \ell}^{\infty}  \frac{\abs{\frac{1}{2} \eta^2 -  \frac{1}{2} \eta - h \left(\frac{\eta}{i} \right) \eta^2 + h \left(\frac{1}{i} \right)\eta}}{i^2} \lesssim \sum_{i = \ell}^\infty \frac{1}{i^2} \lesssim \ell^{-1} ,
\end{align*}
where the constant is independent of $i$ and $\ell$. If $\ell$ is sufficiently large we can now bound
\begin{equation*}
 \big|R(\ell) - 1 \big| = \big|\mathrm{e}^{\log(R(\ell))} - 1\big| \leq 2 \big| \log(R(\ell)) \big| \lesssim \ell^{-1} \,. \qedhere
\end{equation*}
\end{proof}
\begin{lemma} \label{lem:toy-model-convergence}  Assume $\gamma>\tfrac12$, $\beta + \alpha (2 \gamma - 1) - 1 > 0$, $\Delta t>0$, 
$e_\ell = \{\delta_{j,\ell}\}_j$, $B$ is the backshift operator mapping $e_j \mapsto e_{j + 1}$, $\rho_c$ is  
defined in \eqref{process-correlation-function}, $c := \{c_j\}_{j = -\infty}^\infty \in \ell^2(\mathbb{Z})$ satisfies the recurrence relation 
\begin{align*}(1 - \mathrm{e}^{-\Delta t \mu_k} B)^{\gamma} c := \sum_{j = 0}^\infty \binom{\gamma}{j} (-1)^j \mathrm{e}^{-\mu_k j\Delta t } B^j c=e_0 \, ,
\end{align*}
and define
\begin{align*}
\Tilde{\rho}(n) := C \Delta t^{2\gamma - 1} \sum_{j = 0}^\infty c_j c_{j + n} \quad \text{for}\quad C := \tfrac{\Gamma(\gamma)^2 (2 \mu_k)^{2 \gamma - 1}}{\Gamma(2 \gamma - 1)}\, .
\end{align*}

Then for $n\in\mathbb Z$,
\begin{align*}
& \qquad (i) \ \ c_n = \begin{cases} \binom{-\gamma}{n} (-1)^n \mathrm{e}^{-\mu_k n  \Delta t }, \, & n \geq 0 \\
0, & n < 0
\end{cases} \\
\qquad\text{and} & \qquad 
(ii) \ \ \abs{\rho_c( n\Delta t ) - \Tilde{\rho}(n)} \lesssim \mu_k^{2 \gamma - 1} \mathrm{e}^{- \mu_k  n\Delta t } \Big( \Delta t \abs{\log(\Delta t)} + \Delta t^{\min(2\gamma - 1, 1)} \Big) \, ,
\end{align*}
\end{lemma}
\begin{remark} It might seems strange that we define the backshift operators by shifting the basis forward (i.e. $e_j \mapsto e_{j + 1}$), however this "forward" shift of the basis corresponds to a "backwards" shift of the coefficients, i.e. $B\left( \sum_{j = 0}^\infty c_j e_j \right) = \sum_{j = 1}^\infty c_{j - 1} e_j$, hence the name "backshift" operator. 
\end{remark}
\begin{proof}
(i) We solve for $c$ using the $Z$-transform (see e.g. \cite{Oppenheim}),
%
\begin{align*}
    1 & = Z[e_0] = \sum_{j = 0}^\infty \binom{\gamma}{j} (-1)^j \mathrm{e}^{-\mu_k j\Delta t } Z[B^j c] 
    = \sum_{j = 0}^\infty \binom{\gamma}{j} (-1)^j \mathrm{e}^{-\mu_k j\Delta t } \sum_{\ell = -\infty}^\infty c_\ell z^{-(\ell + j)} \\
    &= \bigg(\sum_{j = 0}^\infty \binom{\gamma}{j} (-1)^j (\mathrm{e}^{-\mu_k \Delta t} z^{-1})^j \bigg) \bigg(\sum_{\ell = -\infty}^\infty c_\ell z^{-\ell} \bigg) 
    = (1 - \mathrm{e}^{- \mu_k \Delta t} z^{-1})^\gamma \bigg(\sum_{\ell = -\infty}^\infty c_\ell z^{-\ell} \bigg) \, ,
\end{align*}
as functions on $\mathcal{L}^2(\mathbb{T})$, where $\mathbb{T} = \{z \in \mathbb{C} \quad \text{s.t.} \quad |z| = 1\}$, from which it follows that
\begin{align*}
    & \quad \sum_{\ell = -\infty}^\infty c_\ell z^{-\ell} = (1 - \mathrm{e}^{- \mu_k \Delta t} z^{-1})^{-\gamma} = \sum_{\ell = 0}^\infty \binom{-\gamma}{\ell} (-1)^\ell \mathrm{e}^{-\mu_k \ell\Delta t } z^{-\ell} \, , \\
    \quad \text{and hence} & \quad c_{\ell} = \begin{cases} \binom{-\gamma}{\ell} (-1)^{\ell} \mathrm{e}^{-\mu_k {\ell}  \Delta t }, \, & \ell \geq 0 \\
0, & \ell < 0
\end{cases}.
\end{align*}
\noindent (ii) Let $\overline{c}_j :=\frac{1}{\Gamma(\gamma)} \,j^{\gamma - 1} e^{- \mu_k j\Delta t}$ for $j\neq 0$ and $\overline{c}_{j}=0$ for $j \leq 0$, 
and use the triangle inequality to get
\begin{align*}
\big|\rho_c( n\Delta t ) - \Tilde{\rho}(n)\big| & \leq \Big|\rho_c( n\Delta t ) - C \Delta t^{2 \gamma - 1} \sum_{j = 1}^{\infty} \overline{c}_j \overline{c}_{j +n} \Big| + \Big| C \Delta t^{2 \gamma - 1} \sum_{j = 1}^{\infty} \overline{c}_j \overline{c}_{j + n} - \Tilde{\rho}(n)\Big| \\ & = (1) + (2) \, .
\end{align*}
To bound (1), we note that by a change of variables  in \eqref{process-correlation-function},
$$\rho_c(t) =  \frac{\mathrm{e}^{-\mu_k t} (2 \mu_k)^{2 \gamma - 1}}{\Gamma(2 \gamma - 1)} \int_{0}^{\infty}  g(u)
\, \dint u, \qquad \text{for}\qquad g(u)=u^{\gamma - 1} (u +  n\Delta t)^{\gamma - 1}\mathrm{e}^{-2 \mu_k u}$$
so
\begin{align*}
(1) & = \bigg| \frac{\mathrm{e}^{-\mu_k  n\Delta t } (2 \mu_k)^{2 \gamma - 1}}{\Gamma(2 \gamma - 1)} \int_{0}^{\infty}  g(u) 
\, \dint u - \frac{C \Delta t^{2 \gamma - 1} \mathrm{e}^{-\mu_k n\Delta t}}{\Gamma(\gamma)^2} \sum_{j = 1}^{\infty} \mathrm{e}^{-2 \mu_k j \Delta t } j^{\gamma - 1} (j + n)^{\gamma - 1} \bigg|\\
 & = \frac{\mathrm{e}^{- \mu_k  n\Delta t }(2 \mu_k)^{2 \gamma - 1}}{\Gamma(2 \gamma - 1)} \bigg| \int_{0}^{\infty}  g(u) 
 \, \dint u - \Delta t \sum_{j = 1}^{\infty}  g(j\Delta t) 
 \bigg|
\, .
\end{align*}
Note that $\Delta t \sum_{j = 1}^{\infty}  g(j\Delta t) 
$ is the right Riemann sum of $\int_{0}^{\infty}  g(u) 
\, \dint u $. It is well known that the right Riemann sum of a $C^1$ function on a compact domain $[a,b]$ satisfies $|\int_a^b f - R.S.|  \leq \int_a^b|f'|\,\dint x\,\Delta t \leq (b - a)\|f'\|_\infty \Delta t$. We will now extend this error bound to our setting. If $\gamma > 2$, then 
\begin{align*}
   I:=\ & \bigg| \int_{0}^{\infty}  g(u)  
   \, \dint u - \Delta t \sum_{j = 1}^{\infty}  g(j\Delta t) 
   \bigg| 
    \leq 
    \, \bigg| \int_{0}^{N\Delta t}  g(u)  
    \, \dint u - \Delta t \sum_{j = 1}^{N}  g(j\Delta t) 
    \bigg| + 2 \int_{N\Delta t}^\infty  g(u)  
    \, \dint u \\
    \lesssim & \, N\Delta t^2 + \mathrm{e}^{-2 \mu_k N \Delta t} \simeq \Delta t \log \left(\Delta t^{-1}\right) + \Delta t = \Delta t \abs{ \log(\Delta t) } + \Delta t,
\end{align*}
where we have assumed $\Delta t \leq 1$ so that $\log(\Delta t^{-1}) = \abs{\log(\Delta t)}$ and chosen $N = \frac{\log(\Delta t^{-1})}{2 \mu_k \Delta t}$. If $\frac{1}{2} < \gamma < 2$ then the integrand is not continuously differentiable at $0$, so we  we split the integral two more times:
\begin{align*}
    I\lesssim & \, 
     \, \int_0^{\Delta t} g(u)
    \, \dint u +\Delta t\; g(\Delta t)
     +\bigg| \int_{\Delta t}^{1} 
      g(u)  
     \, \dint u - \Delta t \sum_{j\geq 2\,:\, j\Delta t\leq 1}  g(j\Delta t) 
     \bigg|  
     \\    & 
     +\bigg| \int_{1}^{N\Delta t}  g(u) 
    \, \dint u - \Delta t \sum_{ j\leq N\,:\, j\Delta t> 1}  g(j\Delta t) 
    \bigg| 
    + 2\int_{N\Delta t}^\infty  g(u) 
    \, \dint u \\[0.2cm]
    \lesssim & \ 
    \Delta t^{2\gamma-1} + \Delta t^{\min(2\gamma-1,1)} +\Delta t^{2\gamma-1}  + N\Delta t^2  + \mathrm{e}^{-2 \mu_k N \Delta t}  \lesssim \Delta t^{2\gamma-1} +\Delta t +\Delta t \log \left(\Delta t^{-1}\right)  
    \\[0.2cm]
    \lesssim &\ \Delta t^{\min( 2\gamma-1, 1)}+\Delta t \abs{ \log \left(\Delta t
    \right) }
    ,
\end{align*} 
since e.g. $g(\Delta t)\lesssim (\Delta t)^{2\gamma-2}$ and $\Delta t\int_{\Delta t}^1|g'(u)|\,\dint u \lesssim \Delta t^{\min(2\gamma-1,1)}$.
Thus 
\begin{align}\label{(1)}
(1) \lesssim \mu_k^{2 \gamma - 1} \mathrm{e}^{- \mu_k  n\Delta t}\Big(\Delta t \abs{\log(\Delta t)} +  \Delta t^{\min( 2\gamma-1, 1)}\Big)\quad\text{for}\quad \gamma>\tfrac12.
\end{align}
To bound $(2)$ we first note that
\begin{align*}
    (2) & \leq C \Delta^{2 \gamma - 1} \sum_{j = 1}^{\infty} \abs{ \overline{c}_j \overline{c}_{j + n} - c_j c_{j + n} } \\
    & \lesssim \mu_k^{2 \gamma - 1} \mathrm{e}^{- \mu_k  n\Delta t } \Delta t^{2 \gamma - 1} \mathrm{e}^{- 
     \mu_j n\Delta t} \sum_{j = 1}^{\infty} \abs{ \frac{j^{\gamma - 1} (j + n)^{\gamma - 1}}{\Gamma(\gamma)^2} - \binom{-\gamma}{j} \binom{-\gamma}{j + n} (-1)^{2j + n}} \mathrm{e}^{-2 \mu_k j \Delta t }.
\end{align*}
Then we apply Lemma \ref{lem:frac-binom-coef-convergence} three times to bound
\begin{align*}
    & \abs{\binom{-\gamma}{j} \binom{-\gamma}{j + n} (-1)^{2j + n} - \frac{j^{\gamma - 1} (j + n)^{\gamma - 1}}{\Gamma(\gamma)^2}} \\
    & \leq \abs{\binom{-\gamma}{j} (-1)^{j} - \frac{j^{\gamma - 1}}{\Gamma(\gamma)}} \binom{-\gamma}{j + n} (-1)^{j + n} + \abs{\binom{-\gamma}{j + n} (-1)^{j + n} - \frac{(j + n)^{\gamma - 1}}{\Gamma(\gamma)}} \frac{j^{\gamma - 1}}{\Gamma(\gamma)} \\
    & \lesssim \abs{\binom{-\gamma}{j} (-1)^{j} j^{1 - \gamma} \Gamma(\gamma) - 1} j^{2 \gamma - 2} + \abs{\binom{-\gamma}{j + n} (-1)^{j + n} (j + n)^{1 - \gamma} \Gamma(\gamma) - 1} j^{2 \gamma - 2} \\
    & \lesssim j^{2 \gamma - 3} \, ,
\end{align*}
and conclude that
\begin{align*}
    (2) 
    & \lesssim  \mu_k^{2 \gamma - 1} \mathrm{e}^{- \mu_k  n\Delta t } \Delta t \,  J(\Delta t) \qquad\text{where}\qquad J(\Delta t)= \Delta t \sum_{j = 1}^{\infty} (j\Delta t)^{2 \gamma - 3} \mathrm{e}^{-2 \mu_k j \Delta t } \, .
\end{align*}
For $\gamma >1$, 
$ J(\Delta t)\rightarrow \int_{0}^\infty x^{2 \gamma - 3} \mathrm{e}^{-2 \mu_k x} \, \dint x \lesssim 1$. For $\frac12<\gamma  < 1$,
\begin{align*}
 J(\Delta t) &\leq \Delta t^{2\gamma - 2} +  \int_{\Delta t}^{\infty} x^{2 \gamma - 3} \mathrm{e}^{-2 \mu_k x} \, \dint x \leq \Delta t^{2\gamma - 2} +\int_{\Delta t}^{1} x^{2 \gamma - 3} \, \dint x +\int_{1}^{\infty}  \mathrm{e}^{-2 \mu_k x} \, \dint x 
\\& 
\lesssim  \Delta t^{2\gamma - 2} + \Delta t^{2\gamma - 2} +\frac1{\mu_k}  \, ,
\end{align*}
For $\gamma=1$, we get 
$ J(\Delta t)\lesssim 1 + |\log(\Delta t)|$ in a similar way.  
Thus
\begin{equation} \label{J}
(2) \lesssim  \mu_k^{2 \gamma - 1} \mathrm{e}^{- \mu_k  n\Delta t } \cdot\ \begin{cases}
      \Delta t^{2 \gamma - 1} & \, , \quad \frac12<\gamma < 1 ,\\[0.2cm]
    \Delta t \abs{\log(\Delta t)}& \, , \quad \gamma = 1, \\[0.2cm]
     \Delta t & \, , \quad \gamma > 1.
\end{cases} 
\end{equation}
Thus by \eqref{(1)} and \eqref{J}
\begin{equation*}
\big|\rho_c(n\Delta t) - \Tilde{\rho}(n)\big| \leq (1) + (2) \lesssim \mu_k^{2 \gamma - 1} \mathrm{e}^{- \mu_k  n\Delta t } \Big(\Delta t \abs{\log(\Delta t)} + 
\Delta t^{\min(2\gamma - 1, 1)} \Big) \, . \qedhere
\end{equation*}

\end{proof}

\begin{lemma} \label{lem:inversion-trick}   Assume $\eta > 0$, $\epsilon > 0$, and there are  polynomials $p$ and $q$ of order $n$ with real coefficients, such that 
%
$$ \sup_{|z| \leq 1} \left|(1 - z)^{\eta} - p(z) / q(z) \right| < \epsilon \, .$$
Then for  any $\delta \in( 0,1)$ such that $\delta^{-\eta} \epsilon < 1$,
$$\sup_{|z| \leq 1 - \delta} \big| (1 - z)^{-\eta} - q(z) / p(z) \big| < \frac{\delta^{-2 \eta}\epsilon}{1 - \delta^{-\eta} \epsilon} \, .$$
\end{lemma}

\begin{proof} 
 Let $\abs{z}\leq 1-\delta$. Then  $\abs{1 - z} \geq 1 - \abs{z} \geq \delta$, $\abs{1 - z}^{-\eta} \leq \delta^{-\eta}$  and 
\begin{align*} 
\text{er}(z) &\!:= \abs{(1-z)^{-\eta} - q(z) / p(z)} 
= \abs{(1-z)^{-\eta} q(z) / p(z)} \cdot \abs{(1-z)^{\eta} - p(z) / q(z)} \\ & \leq \delta^{-\eta} \epsilon \big|q(z) / p(z)\big| \, .
\end{align*}
 Since $ \abs{q(z) / p(z)} \leq \abs{q(z) / p(z) - (1 - z)^{-\eta}} +  \abs{(1 - z)^{-\eta}} < \text{er}(z) + \delta^{-\eta}$ and $\delta^{-\eta} \epsilon < 1$, it follows that 
\begin{equation*}
\text{er}(x) < \delta^{-\eta} \epsilon \, (\text{er}(x) + \delta^{-\eta}) 
%
\qquad \iff\qquad 
\text{er}(x) < \frac{\delta^{-2 \eta}\epsilon}{1 - \delta^{-\eta} \epsilon} \, . \qedhere
\end{equation*}
\end{proof}

\begin{proof}[Proof of Theorem \ref{thm:method-convergence}]

Let $c$ be the sequence from Lemma \ref{lem:toy-model-convergence} and define the time series $x = \{x_j\}_{j = -\infty}^{\infty}$ by
$$x = \sum_{j = -\infty}^\infty \psi_j B^j c \, ,$$
where $\psi = \{\psi_j\}_{j = -\infty}^\infty$
is a sequence of i.i.d. Gaussian variables with variance $\sigma^2_{ k}$. Element-wise this is equivalent to the MA($\infty$)-process $x_j = \sum_{\ell = 0}^\infty \psi_\ell c_{j - \ell}$, since $c_{\ell} = 0$ for $\ell < 0$. The autocorrelation function of $x$ is then $\Tilde{\rho}$ from Lemma \ref{lem:toy-model-convergence}, and $x$ satisfies $(1 - \mathrm{e}^{-\mu_k \Delta t} B)^\gamma x = \psi$ since
\begin{align*}
    (1 - \mathrm{e}^{-\mu_k \Delta t} B)^\gamma x &= \sum_{j = 0}^\infty \binom{\gamma}{j} (-1)^j \mathrm{e}^{- \mu_k j\Delta t } B^j x = \sum_{j = 0}^\infty \binom{\gamma}{j} (-1)^j \mathrm{e}^{- \mu_k j\Delta t } B^j \sum_{\ell = -\infty}^{\infty} \psi_\ell 
    B^\ell c\\
    &= \sum_{\ell = -\infty}^\infty 
    \psi_\ell B^\ell \sum_{j = 0}^{\infty} \binom{\gamma}{j} (-1)^j \mathrm{e}^{- \mu_k j\Delta t } B^j c = \sum_{\ell = -\infty}^\infty \psi_\ell B^\ell (1 - \mathrm{e}^{- \mu_k \Delta t} B)^\gamma c\\
    &= \sum_{\ell = -\infty}^\infty \psi_\ell B^\ell e_0 = \sum_{\ell = -\infty}^{\infty} \psi_\ell e_\ell = \psi\, .
\end{align*}
We now consider the sequence $\Tilde{c} = \{\Tilde{c}_\ell\}_{\ell = 0}^\infty$ which satisfy
\begin{equation*}
p(\mathrm{e}^{-\Delta t \mu_k} B) (1 - \mathrm{e}^{-\Delta t \mu_k} B)^{\lfloor \gamma \rfloor} \Tilde{c} = q(\mathrm{e}^{-\Delta t \mu_k} B) e_0 \, .
\end{equation*}
By a similar argument as above $\Tilde{x} = \sum_{\ell = -\infty}^\infty \psi_\ell B^\ell \Tilde{c}$ satisfies (\ref{eq:time-discrete}), and has autocorrelation function $\rho_{\Tilde{x}}$ defined by $$\rho_{\Tilde{x}}(m) = \Tilde{C} \sum_{j = 0}^\infty \Tilde{c}_j \Tilde{c}_{j + m}, \qquad \text{where}\qquad \Tilde{C} = \Big(\sum_{j = 0}^\infty \Tilde{c}_j^2\Big)^{-1}.$$
It is easy to show (cf. the proof of Lemma \ref{lem:toy-model-convergence} (i)) that the $Z$-transforms $f$ and $\Tilde{f}$ of $c$ and $\Tilde{c}$, respectively, satisfy
$$f(z) = (1 - \mathrm{e}^{-\mu_k \Delta t} z^{ -1 })^{-\gamma} \qquad \text{and}
\qquad \Tilde{f}(z) = \frac{q(\mathrm{e}^{-\Delta t \mu_k} z^{ -1 })}{p(\mathrm{e}^{-\Delta t \mu_k} z^{ -1 })}(1 - \mathrm{e}^{-\mu_k \Delta t} z^{ -1 })^{-\lfloor \gamma \rfloor} \, .$$
By picking $\delta = 1 - \mathrm{e}^{- \mu_k \Delta t}$, Lemma \ref{lem:inversion-trick}, the fact that the $Z$-transform is an 
isometric 
 isomorphism  between $\ell^2(\mathbb{Z})$ and $\mathcal{L}^2(\mathbb{T})$ (a consequence of Parseval's theorem  for the $z$-transform), the bound $1 - \cos(\theta) \geq \frac{2}{\pi^2} \theta^2$, a change of variable in the final step to get the bound $\lesssim \delta^{1 - 2 \lfloor \gamma \rfloor}$ on the integral,  and finally $\eta=\gamma-\lfloor\gamma\rfloor$,  we get
\begin{align}
\sum_{\ell = 0}^\infty (c_\ell - \Tilde{c}_\ell)^2 &= \int_{-\pi}^{\pi} \big|f(\mathrm{e}^{i \theta}) - \Tilde{f}(\mathrm{e}^{i \theta})\big|^2 \dint \theta 
\nonumber\\
& = \int_{-\pi}^{\pi} \big|(1 - \mathrm{e}^{-\mu_k \Delta t} \mathrm{e}^{i \theta})^{-\lfloor \gamma \rfloor}\big|^2 \bigg|(1 - \mathrm{e}^{-\mu_k \Delta t} \mathrm{e}^{i \theta})^{-\eta} -\frac{q(\mathrm{e}^{-\Delta t \mu_k} \mathrm{e}^{i \theta})}{p(\mathrm{e}^{-\Delta t \mu_k} \mathrm{e}^{i \theta})} \bigg|^2 \, \dint \theta \nonumber\\
& \leq \left(\frac{\delta^{-2 \eta}\epsilon}{1 - \delta^{-\eta} \epsilon}\right)^2 \int_{-\pi}^{\pi} (1 + 
\mathrm{e}^{-2 \mu_k \Delta t} - 2 \mathrm{e}^{\mu_k \Delta t} \cos(\theta))^{-\lfloor \gamma \rfloor} \, \dint \theta \nonumber\\
& = \left(\frac{\delta^{-2 \eta}\epsilon}{1 - \delta^{-\eta} \epsilon}\right)^2 \int_{-\pi}^{\pi} ((1 - \mathrm{e}^{-\mu_k \Delta t})^2 + 2 \mathrm{e}^{-
\mu_k \Delta t} (1 - \cos(\theta))^{-\lfloor \gamma \rfloor} \, \dint \theta \nonumber\\
& \leq \left(\frac{\delta^{-2 \eta}\epsilon}{1 - \delta^{-\eta} \epsilon}\right)^2 \int_{-\pi}^{\pi} \left(\delta^2 + \frac{4 \mathrm{e}^{-
\mu_k \Delta t}}{\pi^2} \theta^2 \right)^{-\lfloor \gamma \rfloor} \, \dint \theta \nonumber\\
& \lesssim \frac{\delta^{-4 \eta} \epsilon^2}{(1 - \delta^{-\eta} \epsilon)^2} \delta^{1 - 2 \lfloor \gamma \rfloor} = \frac{\delta^{1 - 2\gamma - 2
\eta}\epsilon^2}{(1 - \delta^{-\eta} \epsilon)^2} \, .\label{c-c}
\end{align}
 We also need bounds on $\sum_{j = 1}^\infty c_j^2$ and $\sum_{j = 1}^\infty \Tilde{c}_j^2$ to proceed.
Since $\gamma>\tfrac12$, 
\begin{align}\label{sumc}
\Delta t^{2 \gamma - 1} \sum_{j = 1}^\infty c_j^2 \lesssim \Delta t^{2 \gamma - 1} \sum_{j = 1}^\infty j^{2 \gamma - 2} \mathrm{e}^{- 2 \mu_k j\Delta t} = \Delta t \sum_{j = 1}^\infty (j\Delta t )^{2 \gamma - 2} \mathrm{e}^{- 2 \mu_k j\Delta t} \lesssim 
\tfrac 1{2\gamma-1} \, .
\end{align}
It follows that $\Delta t^{2 \gamma - 1} \sum_{j = 1}^\infty \Tilde{c}_j^2 \lesssim 1$, since $\tilde c^2_j\leq 2(c_j-\tilde c_j)^2+2 c^2_j$.
Let $\tilde{\rho}$ be defined in Lemma \ref{lem:toy-model-convergence}. We estimate $\left|\Tilde{\rho}(n) - \rho_{\Tilde{x}}(n) \right|$ next,
\begin{align*}
&\left|\Tilde{\rho}(n) - \rho_{\Tilde{x}}(n) \right| =
\bigg|C \Delta t^{2 \gamma - 1} \sum_{j = 0}^{\infty} c_j c_{j + n} -  \Tilde{C} \sum_{j = 0}^{\infty} \Tilde{c}_j \Tilde{c}_{j + n} \bigg|\\
& \leq C \Delta t^{2 \gamma - 1} \bigg|\sum_{j = 0}^{\infty} c_j c_{j + n} -  \sum_{j = 0}^{\infty} \Tilde{c}_j \Tilde{c}_{j + n} \bigg| + \big| C \Delta t^{2 \gamma - 1} - \Tilde{C} \big| \bigg| \sum_{j = 0}^\infty \Tilde{c}_j \Tilde{c}_{j + n} \bigg| \\
& = (1) + (2).
\end{align*}
 We can then bound $(1)$ using \eqref{c-c},
\begin{align*}
(1) & 
\leq C \Delta t^{2 \gamma - 1} \bigg(\bigg|\sum_{j = 0}^{\infty} c_j (\Tilde{c}_{j + n} - c_{j + n}) \bigg| +  \bigg| \sum_{j = 0}^{\infty} (\Tilde{c}_j - c_j) \Tilde{c}_{j + n} \bigg| \bigg) \\[0.2cm]
& \leq C \Delta t^{2 \gamma - 1} \bigg(\sum_{j = 0}^{\infty} (c_j - \Tilde{c}_j)^2\bigg)^{\frac12} \bigg( \bigg(\sum_{j = 0}^\infty c_j^2\bigg)^{\frac12} + \bigg(\sum_{j = 0}^\infty \Tilde{c}_j^2\bigg)^{\frac12}\bigg) \\& 
\lesssim \mu_k^{2 \gamma - 1} \mathrm{e}^{- \mu_k  n\Delta t } \frac{\epsilon \Delta t^{\gamma - \frac{1}{2}} \delta^{ \frac{1}{2} -\gamma - \eta}}{1 - \delta^{-\eta} \epsilon} \, .
\end{align*}
Recall the defintion $\rho_c$ in \eqref{process-correlation-function}.  We can then bound $(2)$ by  \eqref{c-c},\eqref{sumc}, Lemma \ref{lem:toy-model-convergence}, and $\rho_c(0)=1$,
\begin{align*}
(2) &\leq \big| C \Delta t^{2 \gamma - 1} - \Tilde{C} \big| \bigg| \sum_{j = 0}^\infty \Tilde{c}_j^2 \bigg| = \bigg| C \Delta t^{2 \gamma - 1} \sum_{j = 0}^\infty \Tilde{c}_j^2 - 1 \bigg| \\
& \leq \bigg| C \Delta t^{2 \gamma - 1} \sum_{j = 0}^\infty \Big[(\Tilde{c}_j - c_j)^2 
+ 2c_j(\Tilde{c}_j - c_j)
+ c_j^2  \Big] - \rho_c(0) \bigg| + \left| \rho_c(0) - 1 \right| \\
& \leq C \Delta t^{2 \gamma - 1} \sum_{j = 0} (\Tilde{c}_j - c_j)^2 + 2C \Delta t^{2 \gamma - 1} \bigg(\sum_{j = 0}^\infty c_j^2\bigg)^{\frac12} \bigg(\sum_{j = 0}^\infty (\Tilde{c}_j - c_j)^2\bigg)^{\frac12} + |\Tilde{\rho}(0) - \rho_c(0)| \\ 
& \lesssim \mu_k^{2 \gamma - 1} \left( \frac{\epsilon\Delta t^{\gamma - \frac{1}{2}} \delta^{\frac{1}{2}-\gamma - \eta}}{1 - \delta^{-\eta} \epsilon} + \Delta t \abs{\log(\Delta t)} + \Delta t^{\min(2 \gamma - 1, 1)} \right)\, .
\end{align*}
Finally,  since  $\delta = 1 - \exp(-\mu_k \Delta t)
\geq\mu_k \Delta t$  and $\gamma + \eta > \tfrac12$,  by the above bounds and Lemma \ref{lem:toy-model-convergence},
\begin{align*}
\left|\rho_c(n\Delta t) - \rho_{\Tilde{x}}(n) \right| & \leq \left| \rho_c(n\Delta t) - \Tilde{\rho}(n) \right| + \left|\Tilde{\rho}(n) - \rho_{\Tilde{x}}(n) \right| \\
& \lesssim (1) + (2) + \mu_k^{2 \gamma - 1} \mathrm{e}^{- \mu_k  n\Delta t } \left( \Delta t \abs{\log(\Delta t)} + \mu_k^{2 \gamma - 1}\Delta t^{\min(2\gamma - 1, 1)} \right)  \\[0.2cm]
& \lesssim \mu_k^{2 \gamma - 1} \left(\frac{\epsilon \Delta t^{-\eta} }{1 - \delta^{-\eta} \epsilon} + \Delta t \abs{\log(\Delta t)} + \Delta t^{\min(2\gamma - 1, 1)} \right)\, . \hfill\qedhere
\end{align*}
\end{proof}

\subsection{Proof of Theorem \ref{thm:full-convergence}} \label{suppl:proof-thm-3}

\begin{proof}[Proof of Theorem \ref{thm:full-convergence}]

\begin{align*}
    & \, \quad \abs{C_u(\mathbf{s}_1, \mathbf{s}_2, {n \Delta t}) - \Tilde{C}_u^M(\mathbf{s}_1, \mathbf{s}_2, n)} \\
    & \leq \abs{C_u(\mathbf{s}_1, \mathbf{s}_2, {n \Delta t}) - C_u^M(\mathbf{s}_1, \mathbf{s}_2, {n \Delta t})} + \abs{C_u^M(\mathbf{s}_1, \mathbf{s}_2, {n \Delta t}) - \Tilde{C}_u^M(\mathbf{s}_1, \mathbf{s}_2, n)} \\
    & = (1) + (2) \, .
\end{align*}
By Theorem \ref{thm:weak-spectral-convergence}, $(1) \lesssim M^{-\frac{2\nu_s}{d}}$. We bound $(2)$ by 
\begin{align*}
    (2) &\leq \sum_{k = 1}^M \frac{\lambda_k \Gamma(2 \gamma - 1)}{(2 \mu_k)^{2 \gamma - 1}\Gamma(\gamma)^2} \abs{\rho_{c_k}({n \Delta t}) - \rho_{\Tilde{c}_k}(n)} \abs{f_k(\mathbf{s}_1)} \abs{f_k(\mathbf{s}_2)} \\
    & \lesssim \left( \frac{\epsilon \Delta t^{- \eta}}{1 - r} + \Delta t \abs{\log(\Delta t)} + \Delta t^{\min(\gamma - \frac{1}{2}, 1)} \delta_{\gamma}(\Delta t) \right) \frac{\Gamma(2 \gamma - 1)}{2^{2 \gamma - 1} \Gamma(\gamma)^2} \sum_{k = 1}^M \lambda_k\\
    & \lesssim C(M) \left(\frac{\epsilon \Delta t^{ - \eta}}{1 - r} + \Delta t \abs{\log(\Delta t)} + \Delta t^{\min(\gamma - \frac{1}{2}, 1)} \delta_{\gamma}(\Delta t)\right) \, ,
\end{align*}
where $C(M) := \sum_{k = 1}^M \lambda_k$, we applied Theorem \ref{thm:method-convergence}, and used that $\|f_k\|_\infty \lesssim 1$ for the cosine/sine basis on rectangular domains. 
 By the Weyl bound for eigenvalues of the Laplacian,  $\xi_k \lesssim k^{\frac{2}{d}}$ (\cite{davies1995,Lablee2015-pi}), we have  $\lambda_k\lesssim k^{-\frac{2\beta}d}$ and $C(M) \lesssim \sum_{k = 1}^\infty k^{-\frac{2\beta}{d}}$.
If $\beta > \frac{d}{2}$, then  $C(M) 
\lesssim 1$ (the covariance operator $Q = (\kappa^2 - \Delta)^{-\beta}$ is Hilbert-Schmidt). If $\beta < \frac{d}{2}$, 
$C(M)\leq  1 + \int_{1}^M x^{-\frac{2\beta}{d}} \dint x\lesssim M^{1 - \frac{2 \beta}{d}}$. Finally if $\beta = \frac{d}{2}$, $C(M)\lesssim \log(M)$.\qedhere
\end{proof}

\subsection{The sequential Kalman filter} \label{suppl:seq-Kalman-filter}

The main idea behind the sequential Kalman filter is to modify the joint update and replace it with a sequential update. This section should be read together with Section \ref{sec:ModelAndInf} in the main text, and notation is defined there. Let $\hat{\mathbf{S}}_{n,0} = \tilde{\mathbf{S}}_n$ and $\hat{\mathbf{m}}_{n,0} = \tilde{\mathbf{m}}_n$, and, for $j = 1, \ldots, n_{\mathrm{obs}}^n$, let
\begin{equation*}
    a_{j}^n = \boldsymbol{h}_{n,j}^\mathrm{T} \Hat{\mathbf{S}}_{n,j - 1} \boldsymbol{h}_{n,j} + \sigma_{\mathrm{obs}}^2, \quad \text{and} \quad
    \Hat{y}_{n,j} := \boldsymbol{h}_{n,j}^\mathrm{T} \Hat{\boldsymbol{m}}_{n,j-1} +\boldsymbol{g}_{n,j}^T
    \boldsymbol{\beta},
\end{equation*}
where $\boldsymbol{h}_{n,j}^\mathrm{T}$ is row $j$ in $\mathbf{H}_n$ and $\boldsymbol{g}_{n,j}^T$ is row $j$ in $\mathbf{G}_n$. We then let
\(
    \boldsymbol{k}_{n,j} = (\Hat{\mathbf{S}}_{n,j-1} \boldsymbol{h}_{n,j})/a_{n,j} \, ,
\)
and 
\begin{equation*}
    \Hat{\boldsymbol{m}}_{n,j} = \Hat{\boldsymbol{m}}_{n,j - 1} + \boldsymbol{k}_{n,j} \left(y_j^n - \Hat{y}_{n,j}\right), \quad \text{and}\quad
    \Hat{\mathbf{S}}_{n,j} = \Hat{\mathbf{S}}_{n,j - 1} - \boldsymbol{k}_{n,j} \boldsymbol{h}_{n,j}^\mathrm{T} \Hat{\mathbf{S}}_{n,j - 1}  \, .
\end{equation*} 
Finally, we set $\Hat{\boldsymbol{m}}_n = \Hat{\boldsymbol{m}}_{n,n_{\mathrm{obs}}^n}$ and $\Hat{\mathbf{S}}_n = \Hat{\mathbf{S}}_{n, n_{\mathrm{obs}}^n}$. The final form of the log-likelihood becomes
\begin{equation*}
    \ell(\boldsymbol{\theta}; \boldsymbol{y}^1,\ldots,\boldsymbol{y}^N) = \sum_{n = 1}^N \sum_{j = 1}^{n_{\mathrm{obs}}^n} \log(\phi(y_j^n; \Hat{y}_{n,j}, a_{n,j})) \, ,
\end{equation*}
and the estimated parameters using maximum likelihood are
\[
    \hat{\boldsymbol{\theta}}_{\mathrm{MLE}}=\arg\max_{\boldsymbol{\theta}}  \ell(\boldsymbol{\theta}; \boldsymbol{y}^1,\ldots,\boldsymbol{y}^N).
\]
For a more in-depth description of the (sequential) Kalman filter, we refer the reader to \cite{Anderson2005-hn}, or an equivalent textbook on Kalman filtering.

\subsection{Case study - CV parameter estimates} \label{suppl:case-study-parameter-estimate-tables}

In section contains tables with parameter estimates for each of the five folds used in the cross-validation in Section \ref{sec:sensitivity} in the main text. Parameter estimates for the linear model can found in Table \ref{tab:param_est_lin}. Parameter estimates for the SPDE-model for different spatial resolutions can be found in Tables \ref{tab:param_est_4}-\ref{tab:param_est_16}.

\begin{table}[H]\footnotesize\centering
    \caption{Parameter estimates in the linear model across the five folds}
    \label{tab:param_est_lin}
    \begin{tabular}{c|c|c|c|c}
    & \multicolumn{4}{c}{Linear regression} \\ \hline
    & $\beta_0$ & $\beta_1$ & $\beta_2$ & $\beta_3$ \\ \hline
    
Fold \#1 & 21.114 & -4.949 & 1.587 & -0.296\\ \hline
Fold \#2 & 21.430 & -4.945 & 1.638 & -0.304\\ \hline
Fold \#3 & 21.120 & -5.031 & 1.438 & -0.296\\ \hline
Fold \#4 & 21.130 & -4.897 & 1.404 & -0.296\\ \hline
Fold \#5 & 20.669 & -4.978 & 1.785 & -0.287\\ \hline
\end{tabular}
          \end{table}
         
\begin{table}[H]\footnotesize\centering
    \caption{Parameter estimates across five folds for $M=4^2$}
    \label{tab:param_est_4}
    \begin{tabular}{c|c|c|c|c|c|c|cV{3}c|c|c|c}
    & \multicolumn{7}{cV{3}}{Full model} & \multicolumn{4}{c}{Reduced model} \\ \hline
    & $\nu_t$ & $\nu_s$ & $\beta_s$ & $r_t$ &  $r_s$ & $\sigma$ & $\sigma_{\mathrm{obs}}$ & $r_t$ &  $r_s$ & $\sigma$ & $\sigma_{\mathrm{obs}}$ \\ \hline
    
Fold \#1 & 0.668 & 0.263 & 0.447 & 6.842 & 947.817 & 3.418 & 1.241 & 6.313 & 681.016 & 2.501 & 1.241\\ \hline
Fold \#2 & 0.500 & 0.277 & 0.226 & 12.480 & 854.076 & 3.202 & 1.170 & 6.394 & 666.932 & 2.532 & 1.170\\ \hline
Fold \#3 & 0.483 & 0.283 & 0.242 & 14.086 & 833.805 & 3.352 & 1.132 & 6.827 & 631.080 & 2.579 & 1.132\\ \hline
Fold \#4 & 0.343 & 0.258 & 0.303 & 18.446 & 942.113 & 3.682 & 1.103 & 6.502 & 641.299 & 2.567 & 1.103\\ \hline
Fold \#5 & 0.448 & 0.267 & 0.235 & 14.914 & 809.928 & 3.376 & 1.220 & 6.963 & 617.563 & 2.603 & 1.220\\ \hline
\end{tabular}
          \end{table}
         
\begin{table}[H]\footnotesize\centering
    \caption{Parameter estimates across five folds for $M=5^2$}
    \label{tab:param_est_5}
    \begin{tabular}{c|c|c|c|c|c|c|cV{3}c|c|c|c}
    & \multicolumn{7}{cV{3}}{Full model} & \multicolumn{4}{c}{Reduced model} \\ \hline
    & $\nu_t$ & $\nu_s$ & $\beta_s$ & $r_t$ &  $r_s$ & $\sigma$ & $\sigma_{\mathrm{obs}}$ & $r_t$ &  $r_s$ & $\sigma$ & $\sigma_{\mathrm{obs}}$ \\ \hline
    
Fold \#1 & 0.500 & 0.290 & 0.153 & 26.143 & 779.901 & 4.514 & 1.154 & 13.879 & 571.846 & 3.433 & 1.154\\ \hline
Fold \#2 & 0.495 & 0.265 & 0.092 & 16.610 & 702.041 & 3.695 & 1.098 & 11.899 & 597.569 & 3.190 & 1.099\\ \hline
Fold \#3 & 0.504 & 0.301 & 0.152 & 16.242 & 792.585 & 3.696 & 1.070 & 9.302 & 605.900 & 2.961 & 1.071\\ \hline
Fold \#4 & 0.498 & 0.405 & 0.189 & 40.710 & 866.570 & 5.473 & 1.047 & 17.949 & 588.664 & 3.928 & 1.047\\ \hline
Fold \#5 & 0.488 & 0.310 & 0.169 & 30.926 & 819.349 & 4.817 & 1.153 & 14.427 & 576.784 & 3.529 & 1.153\\ \hline
\end{tabular}
          \end{table}
         
\begin{table}[H]\footnotesize\centering
    \caption{Parameter estimates across five folds for $M=6^2$}
    \label{tab:param_est_6}
    \begin{tabular}{c|c|c|c|c|c|c|cV{3}c|c|c|c}
    & \multicolumn{7}{cV{3}}{Full model} & \multicolumn{4}{c}{Reduced model} \\ \hline
    & $\nu_t$ & $\nu_s$ & $\beta_s$ & $r_t$ &  $r_s$ & $\sigma$ & $\sigma_{\mathrm{obs}}$ & $r_t$ &  $r_s$ & $\sigma$ & $\sigma_{\mathrm{obs}}$ \\ \hline
    
Fold \#1 & 0.507 & 0.257 & 0.057 & 10.201 & 764.848 & 3.150 & 1.111 & 8.728 & 656.861 & 2.834 & 1.112\\ \hline
Fold \#2 & 0.496 & 0.305 & 0.028 & 9.040 & 746.633 & 2.976 & 1.060 & 8.442 & 677.199 & 2.792 & 1.060\\ \hline
Fold \#3 & 0.511 & 0.263 & 0.063 & 10.098 & 778.393 & 3.140 & 1.038 & 8.359 & 668.098 & 2.812 & 1.038\\ \hline
Fold \#4 & 0.501 & 0.277 & 0.099 & 14.001 & 785.857 & 3.455 & 1.006 & 9.989 & 620.306 & 2.952 & 1.006\\ \hline
Fold \#5 & 0.505 & 0.252 & 0.067 & 10.179 & 766.764 & 3.182 & 1.109 & 8.251 & 645.958 & 2.806 & 1.109\\ \hline
\end{tabular}
          \end{table}
         
\begin{table}[H]\footnotesize\centering
    \caption{Parameter estimates across five folds for $M=7^2$}
    \label{tab:param_est_7}
    \begin{tabular}{c|c|c|c|c|c|c|cV{3}c|c|c|c}
    & \multicolumn{7}{cV{3}}{Full model} & \multicolumn{4}{c}{Reduced model} \\ \hline
    & $\nu_t$ & $\nu_s$ & $\beta_s$ & $r_t$ &  $r_s$ & $\sigma$ & $\sigma_{\mathrm{obs}}$ & $r_t$ &  $r_s$ & $\sigma$ & $\sigma_{\mathrm{obs}}$ \\ \hline
    
Fold \#1 & 0.500 & 0.305 & 0.000 & 8.952 & 731.717 & 3.019 & 1.081 & 9.606 & 670.858 & 2.928 & 1.082\\ \hline
Fold \#2 & 0.500 & 0.424 & 0.006 & 8.239 & 747.613 & 2.856 & 1.038 & 8.201 & 713.173 & 2.792 & 1.038\\ \hline
Fold \#3 & 0.478 & 0.281 & 0.001 & 11.086 & 731.936 & 3.262 & 1.011 & 11.488 & 662.553 & 3.132 & 1.012\\ \hline
Fold \#4 & 0.505 & 0.255 & 0.001 & 9.089 & 705.228 & 3.048 & 0.975 & 9.806 & 642.155 & 2.936 & 0.975\\ \hline
Fold \#5 & 0.500 & 0.257 & 0.016 & 8.776 & 743.324 & 3.054 & 1.082 & 8.930 & 656.649 & 2.873 & 1.082\\ \hline
\end{tabular}
          \end{table}
         
\begin{table}[H]\footnotesize\centering
    \caption{Parameter estimates across five folds for $M=8^2$}
    \label{tab:param_est_8}
    \begin{tabular}{c|c|c|c|c|c|c|cV{3}c|c|c|c}
    & \multicolumn{7}{cV{3}}{Full model} & \multicolumn{4}{c}{Reduced model} \\ \hline
    & $\nu_t$ & $\nu_s$ & $\beta_s$ & $r_t$ &  $r_s$ & $\sigma$ & $\sigma_{\mathrm{obs}}$ & $r_t$ &  $r_s$ & $\sigma$ & $\sigma_{\mathrm{obs}}$ \\ \hline
    
Fold \#1 & 0.474 & 0.410 & 0.003 & 11.022 & 757.352 & 3.157 & 1.063 & 10.687 & 711.665 & 3.080 & 1.064\\ \hline
Fold \#2 & 0.468 & 0.528 & 0.002 & 12.157 & 725.790 & 3.195 & 1.020 & 11.209 & 733.189 & 3.177 & 1.020\\ \hline
Fold \#3 & 0.467 & 0.375 & 0.003 & 12.256 & 760.707 & 3.317 & 0.993 & 11.953 & 699.640 & 3.220 & 0.993\\ \hline
Fold \#4 & 0.483 & 0.434 & 0.004 & 13.860 & 714.258 & 3.450 & 0.955 & 13.416 & 679.503 & 3.382 & 0.955\\ \hline
Fold \#5 & 0.487 & 0.414 & 0.003 & 11.718 & 746.185 & 3.310 & 1.064 & 11.639 & 702.842 & 3.241 & 1.064\\ \hline
\end{tabular}
          \end{table}
         
\begin{table}[H]\footnotesize\centering
    \caption{Parameter estimates across five folds for $M=9^2$}
    \label{tab:param_est_9}
    \begin{tabular}{c|c|c|c|c|c|c|cV{3}c|c|c|c}
    & \multicolumn{7}{cV{3}}{Full model} & \multicolumn{4}{c}{Reduced model} \\ \hline
    & $\nu_t$ & $\nu_s$ & $\beta_s$ & $r_t$ &  $r_s$ & $\sigma$ & $\sigma_{\mathrm{obs}}$ & $r_t$ &  $r_s$ & $\sigma$ & $\sigma_{\mathrm{obs}}$ \\ \hline
    
Fold \#1 & 0.435 & 0.443 & 0.005 & 12.998 & 790.260 & 3.293 & 1.050 & 11.286 & 743.969 & 3.167 & 1.050\\ \hline
Fold \#2 & 0.438 & 0.556 & 0.001 & 14.212 & 741.609 & 3.306 & 1.007 & 11.824 & 762.835 & 3.254 & 1.007\\ \hline
Fold \#3 & 0.478 & 0.419 & 0.003 & 9.310 & 793.232 & 2.987 & 0.980 & 9.184 & 745.906 & 2.931 & 0.980\\ \hline
Fold \#4 & 0.433 & 0.457 & 0.002 & 19.633 & 737.936 & 3.819 & 0.936 & 16.407 & 706.349 & 3.705 & 0.936\\ \hline
Fold \#5 & 0.471 & 0.429 & 0.003 & 11.384 & 781.446 & 3.240 & 1.049 & 10.936 & 736.157 & 3.167 & 1.049\\ \hline
\end{tabular}
          \end{table}
         
\begin{table}[H]\footnotesize\centering
    \caption{Parameter estimates across five folds for $M=10^2$}
    \label{tab:param_est_10}
    \begin{tabular}{c|c|c|c|c|c|c|cV{3}c|c|c|c}
    & \multicolumn{7}{cV{3}}{Full model} & \multicolumn{4}{c}{Reduced model} \\ \hline
    & $\nu_t$ & $\nu_s$ & $\beta_s$ & $r_t$ &  $r_s$ & $\sigma$ & $\sigma_{\mathrm{obs}}$ & $r_t$ &  $r_s$ & $\sigma$ & $\sigma_{\mathrm{obs}}$ \\ \hline
    
Fold \#1 & 0.356 & 0.529 & 0.000 & 35.565 & 770.106 & 4.327 & 1.033 & 18.780 & 773.414 & 4.007 & 1.033\\ \hline
Fold \#2 & 0.391 & 0.584 & 0.001 & 23.525 & 750.584 & 3.860 & 0.992 & 13.979 & 788.240 & 3.527 & 0.992\\ \hline
Fold \#3 & 0.492 & 0.514 & 0.002 & 9.050 & 775.913 & 2.940 & 0.969 & 8.757 & 775.471 & 2.912 & 0.969\\ \hline
Fold \#4 & 0.345 & 0.549 & 0.001 & 84.421 & 737.253 & 5.384 & 0.922 & 19.168 & 744.186 & 4.029 & 0.923\\ \hline
Fold \#5 & 0.373 & 0.506 & 0.001 & 42.550 & 768.166 & 4.803 & 1.031 & 22.214 & 757.225 & 4.383 & 1.031\\ \hline
\end{tabular}
          \end{table}
         
\begin{table}[H]\footnotesize\centering
    \caption{Parameter estimates across five folds for $M=11^2$}
    \label{tab:param_est_11}
    \begin{tabular}{c|c|c|c|c|c|c|cV{3}c|c|c|c}
    & \multicolumn{7}{cV{3}}{Full model} & \multicolumn{4}{c}{Reduced model} \\ \hline
    & $\nu_t$ & $\nu_s$ & $\beta_s$ & $r_t$ &  $r_s$ & $\sigma$ & $\sigma_{\mathrm{obs}}$ & $r_t$ &  $r_s$ & $\sigma$ & $\sigma_{\mathrm{obs}}$ \\ \hline
    
Fold \#1 & 0.390 & 0.549 & 0.001 & 20.786 & 779.380 & 3.740 & 1.020 & 15.183 & 797.074 & 3.662 & 1.020\\ \hline
Fold \#2 & 0.415 & 0.582 & 0.001 & 20.288 & 761.383 & 3.749 & 0.978 & 15.436 & 805.602 & 3.710 & 0.978\\ \hline
Fold \#3 & 0.452 & 0.598 & 0.003 & 11.778 & 745.517 & 3.099 & 0.960 & 9.960 & 801.864 & 3.090 & 0.960\\ \hline
Fold \#4 & 0.391 & 0.567 & 0.000 & 20.653 & 743.175 & 3.686 & 0.914 & 12.705 & 775.713 & 3.373 & 0.914\\ \hline
Fold \#5 & 0.424 & 0.511 & 0.001 & 19.905 & 774.122 & 3.861 & 1.019 & 16.282 & 775.370 & 3.815 & 1.018\\ \hline
\end{tabular}
          \end{table}
         
\begin{table}[H]\footnotesize\centering
    \caption{Parameter estimates across five folds for $M=12^2$}
    \label{tab:param_est_12}
    \begin{tabular}{c|c|c|c|c|c|c|cV{3}c|c|c|c}
    & \multicolumn{7}{cV{3}}{Full model} & \multicolumn{4}{c}{Reduced model} \\ \hline
    & $\nu_t$ & $\nu_s$ & $\beta_s$ & $r_t$ &  $r_s$ & $\sigma$ & $\sigma_{\mathrm{obs}}$ & $r_t$ &  $r_s$ & $\sigma$ & $\sigma_{\mathrm{obs}}$ \\ \hline
    
Fold \#1 & 0.371 & 0.606 & 0.001 & 26.615 & 759.509 & 3.944 & 1.011 & 16.222 & 821.697 & 3.812 & 1.011\\ \hline
Fold \#2 & 0.394 & 0.634 & 0.000 & 32.211 & 746.099 & 4.321 & 0.969 & 19.174 & 827.976 & 4.133 & 0.968\\ \hline
Fold \#3 & 0.448 & 0.625 & 0.001 & 14.130 & 726.468 & 3.284 & 0.950 & 11.211 & 813.032 & 3.255 & 0.950\\ \hline
Fold \#4 & 0.409 & 0.595 & 0.001 & 17.549 & 738.966 & 3.482 & 0.905 & 12.410 & 800.743 & 3.369 & 0.904\\ \hline
Fold \#5 & 0.431 & 0.569 & 0.001 & 20.283 & 756.280 & 3.869 & 1.009 & 16.039 & 799.179 & 3.825 & 1.009\\ \hline
\end{tabular}
          \end{table}
         
\begin{table}[H]\footnotesize\centering
    \caption{Parameter estimates across five folds for $M=16^2$}
    \label{tab:param_est_16}
    \begin{tabular}{c|c|c|c|c|c|c|cV{3}c|c|c|c}
    & \multicolumn{7}{cV{3}}{Full model} & \multicolumn{4}{c}{Reduced model} \\ \hline
    & $\nu_t$ & $\nu_s$ & $\beta_s$ & $r_t$ &  $r_s$ & $\sigma$ & $\sigma_{\mathrm{obs}}$ & $r_t$ &  $r_s$ & $\sigma$ & $\sigma_{\mathrm{obs}}$ \\ \hline
    
Fold \#1 & 0.404 & 0.639 & 0.001 & 26.826 & 756.525 & 4.114 & 0.976 & 15.582 & 869.269 & 3.828 & 0.975\\ \hline
Fold \#2 & 0.412 & 0.817 & 0.000 & 109788.854 & 684.221 & 10.250 & 0.936 & 15.983 & 905.170 & 3.928 & 0.941\\ \hline
Fold \#3 & 0.465 & 0.763 & 0.000 & 167050.760 & 686.736 & 17.576 & 0.900 & 402.847 & 903.160 & 19.271 & 0.897\\ \hline
Fold \#4 & 0.416 & 0.652 & 0.001 & 19.128 & 742.379 & 3.628 & 0.873 & 12.960 & 869.257 & 3.522 & 0.872\\ \hline
Fold \#5 & 0.363 & 0.677 & 0.000 & 8960.525 & 731.473 & 8.088 & 0.971 & 14.896 & 862.843 & 3.806 & 0.975\\ \hline
\end{tabular}
          \end{table}
        
\subsection{Case study - CV prediction scores}
\label{suppl:case-study-cv-prediction-score-tables}
This section contains tables with RMSE and CRPS scores for both filtering and one-step forecasting each of the five folds used in the cross-validation in Section \ref{sec:sensitivity}. Scores for different spatial resolutions can be found in Tables \ref{tab:cross_val_4}-\ref{tab:cross_val_16}.

\begin{table}[H]
    \footnotesize\centering
    \caption{CRPS/RMSE for filtering/forecasting across the five folds for $M=4^2$. Smallest value is marked in bold.}
    \label{tab:cross_val_4}
    \begin{tabular}{cV{3}c|c|c|cV{3}c|c|c|c}
    & \multicolumn{4}{cV{3}}{RMSE} & \multicolumn{4}{c}{CRPS} \\ \hline 
    & \multicolumn{2}{c|}{Forecasting} & \multicolumn{2}{cV{3}}{Filtering} & \multicolumn{2}{c|}{Forecasting} & \multicolumn{2}{c}{Filtering} \\ \hline
    & Full & Simple & Full & Simple & Full & Simple & Full & Simple \\ \hlineB{3}
    
Fold \#1 & \textbf{2.149} & 2.295 & 1.073 & \textbf{1.073} & \textbf{1.187} & 1.274 & 0.576 & \textbf{0.576}\\ \hline
Fold \#2 & \textbf{2.354} & 2.439 & \textbf{1.319} & 1.319 & \textbf{1.304} & 1.359 & \textbf{0.708} & 0.708\\ \hline
Fold \#3 & \textbf{2.254} & 2.351 & \textbf{1.498} & 1.500 & \textbf{1.259} & 1.318 & \textbf{0.799} & 0.800\\ \hline
Fold \#4 & \textbf{2.643} & 2.674 & \textbf{1.692} & 1.694 & \textbf{1.460} & 1.482 & \textbf{0.887} & 0.888\\ \hline
Fold \#5 & \textbf{2.212} & 2.268 & \textbf{1.109} & 1.110 & \textbf{1.234} & 1.271 & \textbf{0.595} & 0.595\\ \hline
Total:  & \textbf{2.329} & 2.410 & \textbf{1.359} & 1.359 & \textbf{1.289} & 1.341 & \textbf{0.713} & 0.713\\ \hline
\end{tabular}
          \end{table}
         
\begin{table}[H]
    \footnotesize\centering
    \caption{CRPS/RMSE for filtering/forecasting across the five folds for $M=5^2$. Smallest value is marked in bold.}
    \label{tab:cross_val_5}
    \begin{tabular}{cV{3}c|c|c|cV{3}c|c|c|c}
    & \multicolumn{4}{cV{3}}{RMSE} & \multicolumn{4}{c}{CRPS} \\ \hline 
    & \multicolumn{2}{c|}{Forecasting} & \multicolumn{2}{cV{3}}{Filtering} & \multicolumn{2}{c|}{Forecasting} & \multicolumn{2}{c}{Filtering} \\ \hline
    & Full & Simple & Full & Simple & Full & Simple & Full & Simple \\ \hlineB{3}
    
Fold \#1 & 2.253 & \textbf{2.251} & 1.122 & \textbf{1.122} & 1.250 & \textbf{1.247} & 0.593 & \textbf{0.593}\\ \hline
Fold \#2 & \textbf{2.338} & 2.343 & 1.281 & \textbf{1.280} & \textbf{1.295} & 1.299 & 0.688 & \textbf{0.687}\\ \hline
Fold \#3 & \textbf{2.298} & 2.335 & \textbf{1.454} & 1.455 & \textbf{1.285} & 1.308 & \textbf{0.768} & 0.769\\ \hline
Fold \#4 & 2.648 & \textbf{2.617} & \textbf{1.558} & 1.558 & 1.475 & \textbf{1.455} & 0.816 & \textbf{0.816}\\ \hline
Fold \#5 & 2.275 & \textbf{2.259} & 1.005 & \textbf{1.004} & 1.269 & \textbf{1.258} & 0.550 & \textbf{0.550}\\ \hline
Total:  & 2.367 & \textbf{2.365} & \textbf{1.300} & 1.300 & 1.315 & \textbf{1.313} & 0.683 & \textbf{0.683}\\ \hline
\end{tabular}
          \end{table}
         
\begin{table}[H]
    \footnotesize\centering
    \caption{CRPS/RMSE for filtering/forecasting across the five folds for $M=6^2$. Smallest value is marked in bold.}
    \label{tab:cross_val_6}
    \begin{tabular}{cV{3}c|c|c|cV{3}c|c|c|c}
    & \multicolumn{4}{cV{3}}{RMSE} & \multicolumn{4}{c}{CRPS} \\ \hline 
    & \multicolumn{2}{c|}{Forecasting} & \multicolumn{2}{cV{3}}{Filtering} & \multicolumn{2}{c|}{Forecasting} & \multicolumn{2}{c}{Filtering} \\ \hline
    & Full & Simple & Full & Simple & Full & Simple & Full & Simple \\ \hlineB{3}
    
Fold \#1 & \textbf{2.242} & 2.259 & \textbf{1.062} & 1.062 & \textbf{1.245} & 1.254 & \textbf{0.561} & 0.561\\ \hline
Fold \#2 & \textbf{2.305} & 2.312 & \textbf{1.224} & 1.225 & \textbf{1.279} & 1.284 & \textbf{0.656} & 0.656\\ \hline
Fold \#3 & \textbf{2.303} & 2.334 & \textbf{1.349} & 1.354 & \textbf{1.288} & 1.308 & \textbf{0.717} & 0.719\\ \hline
Fold \#4 & \textbf{2.587} & 2.592 & \textbf{1.489} & 1.490 & \textbf{1.434} & 1.439 & \textbf{0.777} & 0.778\\ \hline
Fold \#5 & \textbf{2.217} & 2.238 & 0.977 & \textbf{0.976} & \textbf{1.234} & 1.245 & 0.536 & \textbf{0.535}\\ \hline
Total:  & \textbf{2.335} & 2.351 & \textbf{1.234} & 1.236 & \textbf{1.296} & 1.306 & \textbf{0.649} & 0.650\\ \hline
\end{tabular}
          \end{table}
         
\begin{table}[H]
    \footnotesize\centering
    \caption{CRPS/RMSE for filtering/forecasting across the five folds for $M=7^2$. Smallest value is marked in bold.}
    \label{tab:cross_val_7}
    \begin{tabular}{cV{3}c|c|c|cV{3}c|c|c|c}
    & \multicolumn{4}{cV{3}}{RMSE} & \multicolumn{4}{c}{CRPS} \\ \hline 
    & \multicolumn{2}{c|}{Forecasting} & \multicolumn{2}{cV{3}}{Filtering} & \multicolumn{2}{c|}{Forecasting} & \multicolumn{2}{c}{Filtering} \\ \hline
    & Full & Simple & Full & Simple & Full & Simple & Full & Simple \\ \hlineB{3}
    
Fold \#1 & 2.254 & \textbf{2.251} & 1.053 & \textbf{1.053} & 1.248 & \textbf{1.245} & \textbf{0.547} & 0.547\\ \hline
Fold \#2 & \textbf{2.311} & 2.312 & \textbf{1.164} & 1.164 & \textbf{1.282} & 1.283 & \textbf{0.622} & 0.622\\ \hline
Fold \#3 & 2.273 & \textbf{2.266} & \textbf{1.300} & 1.301 & 1.271 & \textbf{1.266} & 0.688 & \textbf{0.688}\\ \hline
Fold \#4 & \textbf{2.624} & 2.626 & \textbf{1.546} & 1.548 & \textbf{1.457} & 1.459 & \textbf{0.797} & 0.798\\ \hline
Fold \#5 & \textbf{2.221} & 2.223 & \textbf{0.947} & 0.947 & 1.238 & \textbf{1.238} & \textbf{0.516} & 0.516\\ \hline
Total:  & 2.341 & \textbf{2.340} & \textbf{1.220} & 1.221 & 1.299 & \textbf{1.298} & \textbf{0.634} & 0.634\\ \hline
\end{tabular}
          \end{table}
         
\begin{table}[H]
    \footnotesize\centering
    \caption{CRPS/RMSE for filtering/forecasting across the five folds for $M=8^2$. Smallest value is marked in bold.}
    \label{tab:cross_val_8}
    \begin{tabular}{cV{3}c|c|c|cV{3}c|c|c|c}
    & \multicolumn{4}{cV{3}}{RMSE} & \multicolumn{4}{c}{CRPS} \\ \hline 
    & \multicolumn{2}{c|}{Forecasting} & \multicolumn{2}{cV{3}}{Filtering} & \multicolumn{2}{c|}{Forecasting} & \multicolumn{2}{c}{Filtering} \\ \hline
    & Full & Simple & Full & Simple & Full & Simple & Full & Simple \\ \hlineB{3}
    
Fold \#1 & 2.264 & \textbf{2.255} & 1.052 & \textbf{1.052} & 1.252 & \textbf{1.246} & 0.550 & \textbf{0.550}\\ \hline
Fold \#2 & 2.304 & \textbf{2.295} & 1.156 & \textbf{1.155} & 1.277 & \textbf{1.271} & 0.616 & \textbf{0.616}\\ \hline
Fold \#3 & 2.291 & \textbf{2.280} & 1.381 & \textbf{1.377} & 1.280 & \textbf{1.274} & 0.724 & \textbf{0.723}\\ \hline
Fold \#4 & 2.681 & \textbf{2.674} & \textbf{1.627} & 1.627 & 1.487 & \textbf{1.484} & \textbf{0.834} & 0.834\\ \hline
Fold \#5 & 2.204 & \textbf{2.201} & 0.916 & \textbf{0.915} & 1.227 & \textbf{1.225} & 0.504 & \textbf{0.504}\\ \hline
Total:  & 2.355 & \textbf{2.347} & 1.252 & \textbf{1.251} & 1.305 & \textbf{1.300} & 0.646 & \textbf{0.645}\\ \hline
\end{tabular}
          \end{table}
         
\begin{table}[H]
    \footnotesize\centering
    \caption{CRPS/RMSE for filtering/forecasting across the five folds for $M=9^2$. Smallest value is marked in bold.}
    \label{tab:cross_val_9}
    \begin{tabular}{cV{3}c|c|c|cV{3}c|c|c|c}
    & \multicolumn{4}{cV{3}}{RMSE} & \multicolumn{4}{c}{CRPS} \\ \hline 
    & \multicolumn{2}{c|}{Forecasting} & \multicolumn{2}{cV{3}}{Filtering} & \multicolumn{2}{c|}{Forecasting} & \multicolumn{2}{c}{Filtering} \\ \hline
    & Full & Simple & Full & Simple & Full & Simple & Full & Simple \\ \hlineB{3}
    
Fold \#1 & 2.285 & \textbf{2.261} & 1.063 & \textbf{1.061} & 1.263 & \textbf{1.249} & 0.553 & \textbf{0.552}\\ \hline
Fold \#2 & 2.308 & \textbf{2.287} & 1.141 & \textbf{1.140} & 1.279 & \textbf{1.266} & 0.606 & \textbf{0.606}\\ \hline
Fold \#3 & 2.273 & \textbf{2.266} & 1.338 & \textbf{1.337} & 1.271 & \textbf{1.267} & 0.707 & \textbf{0.706}\\ \hline
Fold \#4 & 2.765 & \textbf{2.724} & 1.746 & \textbf{1.733} & 1.539 & \textbf{1.513} & 0.880 & \textbf{0.874}\\ \hline
Fold \#5 & 2.192 & \textbf{2.184} & 0.904 & \textbf{0.903} & 1.222 & \textbf{1.216} & 0.498 & \textbf{0.497}\\ \hline
Total:  & 2.373 & \textbf{2.352} & 1.272 & \textbf{1.267} & 1.315 & \textbf{1.302} & 0.649 & \textbf{0.647}\\ \hline
\end{tabular}
          \end{table}
         
\begin{table}[H]
    \footnotesize\centering
    \caption{CRPS/RMSE for filtering/forecasting across the five folds for $M=10^2$. Smallest value is marked in bold.}
    \label{tab:cross_val_10}
    \begin{tabular}{cV{3}c|c|c|cV{3}c|c|c|c}
    & \multicolumn{4}{cV{3}}{RMSE} & \multicolumn{4}{c}{CRPS} \\ \hline 
    & \multicolumn{2}{c|}{Forecasting} & \multicolumn{2}{cV{3}}{Filtering} & \multicolumn{2}{c|}{Forecasting} & \multicolumn{2}{c}{Filtering} \\ \hline
    & Full & Simple & Full & Simple & Full & Simple & Full & Simple \\ \hlineB{3}
    
Fold \#1 & 2.356 & \textbf{2.286} & 1.099 & \textbf{1.096} & 1.308 & \textbf{1.264} & 0.569 & \textbf{0.567}\\ \hline
Fold \#2 & 2.344 & \textbf{2.297} & 1.158 & \textbf{1.153} & 1.300 & \textbf{1.270} & 0.610 & \textbf{0.607}\\ \hline
Fold \#3 & 2.274 & \textbf{2.274} & \textbf{1.357} & 1.357 & 1.273 & \textbf{1.273} & \textbf{0.713} & 0.713\\ \hline
Fold \#4 & 3.138 & \textbf{2.808} & 2.195 & \textbf{1.863} & 1.748 & \textbf{1.562} & 1.089 & \textbf{0.942}\\ \hline
Fold \#5 & 2.312 & \textbf{2.234} & 1.082 & \textbf{1.027} & 1.289 & \textbf{1.242} & 0.566 & \textbf{0.547}\\ \hline
Total:  & 2.506 & \textbf{2.389} & 1.441 & \textbf{1.334} & 1.384 & \textbf{1.322} & 0.709 & \textbf{0.675}\\ \hline
\end{tabular}
          \end{table}
         
\begin{table}[H]
    \footnotesize\centering
    \caption{CRPS/RMSE for filtering/forecasting across the five folds for $M=11^2$. Smallest value is marked in bold.}
    \label{tab:cross_val_11}
    \begin{tabular}{cV{3}c|c|c|cV{3}c|c|c|c}
    & \multicolumn{4}{cV{3}}{RMSE} & \multicolumn{4}{c}{CRPS} \\ \hline 
    & \multicolumn{2}{c|}{Forecasting} & \multicolumn{2}{cV{3}}{Filtering} & \multicolumn{2}{c|}{Forecasting} & \multicolumn{2}{c}{Filtering} \\ \hline
    & Full & Simple & Full & Simple & Full & Simple & Full & Simple \\ \hlineB{3}
    
Fold \#1 & 2.324 & \textbf{2.276} & 1.092 & \textbf{1.090} & 1.288 & \textbf{1.259} & 0.565 & \textbf{0.564}\\ \hline
Fold \#2 & 2.326 & \textbf{2.292} & \textbf{1.169} & 1.169 & 1.288 & \textbf{1.267} & 0.616 & \textbf{0.616}\\ \hline
Fold \#3 & 2.289 & \textbf{2.273} & \textbf{1.330} & 1.330 & 1.280 & \textbf{1.271} & 0.700 & \textbf{0.700}\\ \hline
Fold \#4 & 2.759 & \textbf{2.662} & 1.710 & \textbf{1.655} & 1.537 & \textbf{1.478} & 0.862 & \textbf{0.839}\\ \hline
Fold \#5 & 2.226 & \textbf{2.195} & 0.972 & \textbf{0.966} & 1.239 & \textbf{1.220} & 0.529 & \textbf{0.526}\\ \hline
Total:  & 2.392 & \textbf{2.345} & 1.280 & \textbf{1.265} & 1.327 & \textbf{1.299} & 0.654 & \textbf{0.649}\\ \hline
\end{tabular}
          \end{table}
         
\begin{table}[H]
    \footnotesize\centering
    \caption{CRPS/RMSE for filtering/forecasting across the five folds for $M=12^2$. Smallest value is marked in bold.}
    \label{tab:cross_val_12}
    \begin{tabular}{cV{3}c|c|c|cV{3}c|c|c|c}
    & \multicolumn{4}{cV{3}}{RMSE} & \multicolumn{4}{c}{CRPS} \\ \hline 
    & \multicolumn{2}{c|}{Forecasting} & \multicolumn{2}{cV{3}}{Filtering} & \multicolumn{2}{c|}{Forecasting} & \multicolumn{2}{c}{Filtering} \\ \hline
    & Full & Simple & Full & Simple & Full & Simple & Full & Simple \\ \hlineB{3}
    
Fold \#1 & 2.350 & \textbf{2.285} & 1.118 & \textbf{1.105} & 1.305 & \textbf{1.266} & 0.588 & \textbf{0.580}\\ \hline
Fold \#2 & 2.343 & \textbf{2.295} & 1.190 & \textbf{1.187} & 1.299 & \textbf{1.268} & 0.632 & \textbf{0.630}\\ \hline
Fold \#3 & 2.279 & \textbf{2.258} & 1.296 & \textbf{1.295} & 1.273 & \textbf{1.261} & 0.688 & \textbf{0.688}\\ \hline
Fold \#4 & 2.746 & \textbf{2.675} & 1.729 & \textbf{1.694} & 1.528 & \textbf{1.484} & 0.869 & \textbf{0.854}\\ \hline
Fold \#5 & 2.233 & \textbf{2.203} & 0.978 & \textbf{0.975} & 1.242 & \textbf{1.224} & 0.537 & \textbf{0.535}\\ \hline
Total:  & 2.397 & \textbf{2.349} & 1.288 & \textbf{1.275} & 1.330 & \textbf{1.301} & 0.663 & \textbf{0.657}\\ \hline
\end{tabular}
          \end{table}
         
\begin{table}[H]
    \footnotesize\centering
    \caption{CRPS/RMSE for filtering/forecasting across the five folds for $M=16^2$. Smallest value is marked in bold.}
    \label{tab:cross_val_16}
    \begin{tabular}{cV{3}c|c|c|cV{3}c|c|c|c}
    & \multicolumn{4}{cV{3}}{RMSE} & \multicolumn{4}{c}{CRPS} \\ \hline 
    & \multicolumn{2}{c|}{Forecasting} & \multicolumn{2}{cV{3}}{Filtering} & \multicolumn{2}{c|}{Forecasting} & \multicolumn{2}{c}{Filtering} \\ \hline
    & Full & Simple & Full & Simple & Full & Simple & Full & Simple \\ \hlineB{3}
    
Fold \#1 & 2.370 & \textbf{2.308} & 1.219 & \textbf{1.184} & 1.316 & \textbf{1.279} & 0.625 & \textbf{0.607}\\ \hline
Fold \#2 & 2.834 & \textbf{2.319} & 1.973 & \textbf{1.259} & 1.589 & \textbf{1.284} & 1.040 & \textbf{0.662}\\ \hline
Fold \#3 & \textbf{3.747} & 4.094 & \textbf{3.211} & 3.631 & \textbf{2.025} & 2.160 & \textbf{1.544} & 1.707\\ \hline
Fold \#4 & 2.685 & \textbf{2.616} & 1.643 & \textbf{1.599} & 1.493 & \textbf{1.449} & 0.831 & \textbf{0.810}\\ \hline
Fold \#5 & 2.490 & \textbf{2.215} & 1.279 & \textbf{0.946} & 1.404 & \textbf{1.235} & 0.683 & \textbf{0.517}\\ \hline
Total:  & 2.867 & \textbf{2.800} & 2.001 & \textbf{1.981} & 1.565 & \textbf{1.481} & 0.945 & \textbf{0.861}\\ \hline
\end{tabular}
          \end{table}


\end{document}